\documentclass[acmsmall]{acmart}

\acmJournal{PACMPL}
\acmVolume{1}
\acmNumber{Under review as a conference paper at POPL 2020} 
\acmArticle{1}
\acmYear{2019}
\acmMonth{1}
\acmDOI{} 
\startPage{1}

\setcopyright{none}

\usepackage{booktabs}   
\usepackage{subcaption} 

\usepackage{amsmath, listings, amsthm, amssymb, proof, xspace}
\usepackage{mathtools}
\usepackage{stmaryrd}

\usepackage[utf8]{inputenc}
\usepackage[english]{babel}
\usepackage[linesnumbered,ruled]{algorithm2e}
\usepackage{placeins}
\usepackage{stmaryrd}

\usepackage{tikz}
\usetikzlibrary{automata,arrows.meta,positioning}
\usepackage{amsthm}

\newtheorem{theorem}{Theorem}
\newtheorem{lemma}{Lemma}
\newtheorem{definition}{Definition}
\newtheorem{example}{Example}
\newtheorem{corollary}{Corollary}
\newtheorem{observation}{Observation}

\newcommand{\assumed} [1] {\mathsf{assume}~#1}
\newcommand{\assume} [2] {\mathsf{assume}~#1=#2}

\newcommand{\binopdef}     \oplus 
\newcommand{\unopdef}      \ominus 

\newcommand{\infr} [3] [] {\infer[\textsc{#1}]{#3}{#2}}

\newcommand{\observed} [2] {\mathsf{observe}(\mathsf{Dist}(#1) = #2)}

\newcommand{\NDE} [3] {\langle \cN#1 , \cD#2, \cE#3 \rangle}

\makeatletter
\DeclareRobustCommand*\cal{\@fontswitch\relax\mathcal}
\makeatother

\def\cA{{\cal A}}
\def\cB{{\cal B}}
\def\cC{{\cal C}}
\def\cD{{\cal D}}
\def\cE{{\cal E}}
\def\cF{{\cal F}}

\def\cN{{\cal N}}

\def\cP{{\cal P}}

\def\cR{{\cal R}}
\def\cS{{\cal S}}
\def\cT{{\cal T}}

\def\cV{{\cal V}}

\def\cX{{\cal X}}
\def\cY{{\cal Y}}

\newcommand{\ndi} [3] {(#1:#2) \# #3}

\newcommand{\tup} [1] {\langle #1 \rangle}

\newcommand{\infra} [2] {\infr{\begin{array}{c} #1 \end{array} }{ \begin{array}{c} #2  \end{array} } }
    \newcommand{\infral} [3] {\infer[\textsc{#3}]{\begin{array}{c} #2 \end{array} }{ \begin{array}{c} #1  \end{array} } }

\newcommand{\Nid} [1] {id#1 \leftarrow \mathsf{Fresh~ID}}
\newcommand{\ID} {\mathsf{ID}}
\newcommand{\SI} [2] {\sigma_v#1, \sigma_{id}#2 \vdash}

\newcommand{\dom} {\mathsf{dom}~}

\newcommand{\Nst} [1] {x#1 \leftarrow \mathsf{Fresh~variable~name}}

\newcommand{\Comment} [1] {}

\newcommand{\dref} [1] {}

\newcommand{\SAC} {\cS \vdash}

\begin{document}

\title[Compositional Inference Metaprogramming]{Compositional Inference Metaprogramming with Convergence Guarantees}         


\author{Shivam Handa}
\orcid{nnnn-nnnn-nnnn-nnnn}             
\affiliation{
  \institution{Massachusetts Institute of Technology}            
}
\email{shivam@mit.edu}          


\author{Vikash Mansinghka}
\affiliation{
  \institution{Massachusetts Institute of Technology}           
}
\author{Martin Rinard}
\affiliation{
  \institution{Massachusetts Institute of Technology}           
}

\begin{abstract}
%
%
%
Inference metaprogramming enables effective probabilistic programming by 
supporting the decomposition of executions
of probabilistic programs into subproblems and the deployment
of hybrid probabilistic inference algorithms that apply different
probabilistic inference algorithms to different subproblems.
We introduce the concept of independent subproblem inference
(as opposed to entangled subproblem inference in which the 
subproblem inference algorithm operates over the full program trace)
and present a mathematical framework for studying convergence
properties of hybrid inference algorithms that apply different 
Markov-Chain Monte Carlo algorithms to different parts of the
inference problem. We then use this formalism to prove asymptotic convergence results
for probablistic programs with inference metaprogramming. 
To the best of our knowledge this is the first asymptotic convergence
result for hybrid probabilistic inference algorithms defined by 
(subproblem-based) inference metaprogramming. 

\end{abstract}

\keywords{Probabilistic Programming, Metaprogramming, Inference Algorithms}  

\maketitle

\section{Introduction}

Probabilistic modeling and inference are now mainstream 
approaches deployed in many areas of computing and data analysis~\cite{russel2003artificial, thrun2005probabilistic,
 liu2008monte, murphy2012machine, gelman2014bayesian, forsyth2002computer}.
To better support these computations, researchers have developed
probabilistic programming languages, which include  constructs that directly 
support probabilistic modeling and inference within the language
itself~\cite{milch2007blog, goodman2008church, goodman2014design,
    mansinghka2014venture, gordon2014tabular, tristan2014augur,
    gordon2014probabilistic, tolpin2015probabilistic, carpenter2016stan,goodman2017pyro}. 
    Probabilistic inference strategies provide the
probabilistic reasoning required to implement these constructs.

It is well known that no one probabilistic inference strategy is
appropriate for all probabilistic inference and modeling tasks~\cite{russel2003artificial,mansinghka2018probabilistic}.
Indeed, effective inference often involves breaking an
inference problem down into subproblems, then applying different
inference strategies to different subproblems as appropriate~\cite{russel2003artificial,mansinghka2018probabilistic}.
Applying this approach to probabilistic programs, specifically by
specifying subtask decompositions and inference strategies to apply
to each subtask, is called {\em inference metaprogramming}. 
Inference metaprogramming has been shown
to dramatically improve the execution time and
accuracy of probabilistic programs (in comparison with monolithic
inference strategies that apply a single inference strategy to the
entire program)~\cite{mansinghka2018probabilistic}. 

\subsection{Probabilistic Inference and Convergence} 

Executions of probabilistic programs typically use {\em probabilistic inference} to generate 
samples from the underlying probability distribution that the program defines~\cite{mansinghka2018probabilistic}.
Many probabilistic inference algorithms are {\em iterative}, i.e., they perform multiple steps that bring
samples closer to the specified probability distribution.  A standard correctness property of such algorithms is 
{\em asymptotic convergence}, i.e.,  a guarantee that, in the limit as the number of iterations increases, 
the resulting sample will be drawn from the defined posterior distribution. Markov-Chain
Monte-Carlo (MCMC) algorithms (which include  Metropolis-Hastings~\cite{chib1995understanding}
and Gibbs sampling~\cite{meyn2012markov}) comprise a widely-used~\cite{milch2007blog, goodman2008church, goodman2014design,
    mansinghka2014venture} class of probabilistic
inference algorithms that often come with asymptotic convergence
guarantees~\cite{geyer1998markov}. 

Using inference metaprogramming to decompose and solve
inference problems into subprograms produces new hybrid probabilistic 
inference algorithms.  
Whether or not these new hybrid inference algorithms (as implemented in the
inference metaprogramming language) also asymptotically converge is often a question of interest
(because it directly relates to the compositional soundness of the
inference metaprogram).

Over the last several decades the field has developed many iterative probabilistic 
inference algorithms~\cite{meyn2012markov} and proved convergence results for these algorithms~\cite{berti2008trivial}.
Many of these algorithms compose inference steps applied to different parts of the problem and would therefore
seem to be a promising candidate for proving convergence properties of probabilistic programs
with inference metaprogramming.
Unfortunately, these algorithms, and their associated convergence proofs, 
have several onerous restrictions. Specifically, they model the state of the 
system as a product space over a fixed set of random choices and work with 
policies whose selection of random choices to resample does not depend on the 
state of the system. The basic mathematical framework (and associated convergence proofs) 
is therefore not applicable to probabilistic programming with inference metaprogramming
--- in this new setting, the random choices in the subproblems (as defined by the random variables that
the program defines and samples) may change over time, are potentially unbounded (e.g., 
Open Universe Probabilistic Models~\cite{milch2010extending, wu2016swift}), 
and may depend on the state of the system as realized in the current values of the random choices. 
For example, our framework supports programs that sample a stochastic choice, then 
compute the set of stochastic choices to include in a subproblem as a function
of the sampled stochastic choice. 

\subsection{Our Result}

We present the first asymptotic convergence result for hybrid probabilistic inference algorithms
defined by inference metaprogramming. Given a probabilistic program with posterior distribution $\pi$,
we show that the hybrid algorithms applied to that program are $\pi$-irreducible, aperiodic, 
and have $\pi$ as their stationary distribution. This result stands on two foundational new results:
\begin{itemize}
\item {\bf Independent Subproblems:} We consider executions of probabilistic programs 
that produce {\em program traces}~\cite{ mansinghka2018probabilistic, wingate2011lightweight}. 
These traces record the random choices made during the execution. With inference metaprogramming, the subproblem
inference algorithms must operate only over the random choices in the subproblem. Previous formulations
of subproblem inference, however, define subproblem inference as operating over the entire trace~\cite{mansinghka2018probabilistic}.
This approach entangles the subproblem with the full program trace and complicates the analysis of the
interaction between subproblems and inference metaprogramming. 

We instead formalize subproblem inference using a new technique that extracts each subproblem from the original 
program trace into its own independent trace.  Inference is then performed over the full extracted
trace, with the newly generated trace then stitched back into the original trace
to complete the subproblem inference. By detangling the subproblem from the full
trace, this approach delivers the clean separation of subproblems and subproblem inference 
required to state and prove the new asymptotic convergence result. 

\item {\bf Mathematical Framework:} We present a new and more general mathematical framework 
for studying the composition of probabilistic inference algorithms applied to subproblems. 
A key aspect of this framework is that it supports state-dependent embeddings between program trace spaces
defined by different probabilistic programs. The framework therefore enables us to model
subproblem inference by embedding the original program trace into the space of program
traces defined by the detangled subproblem, moving the embedded trace within this new space of
program traces, then injecting the new trace back into the original trace space. 

We note that this mathematical framework is not specifically tied to probabilistic programming. 
It supports compositional asymptotic convergence results for a range of probabilistic 
inference algorithms that operate over general probability spaces (even uncomputable ones)
by mapping subspaces of the original space into new isolated probability spaces, 
iteratively applying MCMC inference algorithms to the isolated space, then mapping the 
results back into the original space. An important property is that
the applied inference algorithms, the isolated probability spaces, and the mappings may
all be state-dependent.

\end{itemize}

Building on these results, we prove a new asymptotic convergence result for inference metaprograms
that apply asymptotically converging MCMC algorithms to appropriately defined subproblems. This 
result identifies two key restrictions on the subproblem selection strategies that the 
inference metaprogram uses to identify subproblems. These restrictions
guarantee asymptotic convergence for inference metaprograms that 
apply a large class of asymptotically converging MCMC algorithms
to the specified subproblems:

\begin{itemize}
\item {\bf Reversibility:} 
The subproblem selection strategy must be {\em reversible}, i.e., 
        given $n$ traces $t_1, t_2 \ldots t_n$ such that $t_i$ can be transformed 
        into trace $t_{i+1}$ by modifying parts of the trace $t_i$ selected by the subproblem selection strategy,
        then it must be possible to transform trace $t_n$ into trace $t_1$ by modifying
        parts of the trace $t_n$ selected by the subproblem selection strategy.
        Intuitively, it must be possible to reverse any changes that can be made by countably applying 
        the same subproblem selection strategy to a starting trace. 

\item {\bf Connectivity:} 
The combination of all of the subproblem selection strategies in the 
inference metaprogram must connect the entire probability space. Given two
traces $t$ and $t'$, we say that the subproblem selection strategies connect 
$t$ and $t'$ if it is possible to transform $t$ into $t'$ by modifying
the parts of $t$ selected by one of the subproblem selection strategies. 

The subproblem selection strategies connect the probability space
if there do not exist sets of traces $U$ and $V$ such that 1) the probability of $U \cup V$
sums to one and 2) there does not exist $t \in U$ and $t' \in V$ such that the 
subproblem selection strategies connect $t$ and $t'$. 

\end{itemize}

Conceptually, these two restrictions together ensure that the hybrid inference algorithm
defined by the inference metaprogram does not become stuck in a subset of the positive
probability space and therefore unable to sample some positive probability set. 

Effective probabilistic programming requires subproblem identification and 
hybrid probabilistic inference algorithms applied to the
identified subproblems. The results in this paper
enable the sound and complete decomposition of 
otherwise intractable probabilistic inference problems 
into tractable subproblems solved by different inference algorithms. 
It also characterizes properties that entail asymptotic convergence of these resulting hybrid
probabilistic inference algorithms.

\section{Language and Execution Model}

Our treatment of subproblem selection, extraction, and stitching works with a core 
probabilistic programming language (Figure~\ref{fig:language}) based on the lambda calculus.
A program in this language is a sequence of $\mathsf{assume}$ and
$\mathsf{observe}$ statements. Expressions are derived from the untyped lambda calculus augmented 
with the $\mathsf{Dist}(e)$ expression, which allows the program
to sample from a distribution $\mathsf{Dist}$ given parameter $e$. 

The core language supports computable distributions over computable expressions (including computable reals).
We believe it is straightforward to generalize the language to include more general probability spaces
(e.g., probability spaces including uncomputable reals) at the cost of a larger formalism. 
The mathematical framework we use to prove convergence (Section~\ref{sec:stoch}) works over 
general probability spaces including probability spaces with uncomputable objects. 

\begin{figure}[h]
\[
\begin{array}{rcl}
    e_v \in E_v &:=& x 
    ~\vert~\lambda.x~e_v
    ~\vert~ (e_v~e'_v)\\
    
    e,e_1, e_2 \in E &:=& x
    ~\vert~ \lambda.x~e
    ~\vert~ \mathsf{Dist}(e)
    ~\vert~ (e_1~e_2)\\
    
    s \in S &:=& \assume{x}{e}
    ~\vert~ \observed{e}{e_v}\\
    
    p \in P &:=& \emptyset
    ~\vert~ s;p
\end{array}
\]
\caption{Probabilistic Lambda Calculus}
    \label{fig:language}
\end{figure}

$\mathsf{Dist}(e)$ can be seen as a set of probabilistic lambda calculus expressions 
$\{e_d \vert e_d \in \mathsf{Dist}(e) \subseteq E_v \}$. Based on the parameter 
expression $e$, $\mathsf{Dist}(e)$ makes a stochastic choice and returns 
an expression $e_v \in \mathsf{Dist}(e)$. We define:
\[\mathsf{Dist}(e)[x/y] = \mathsf{Dist'}(e[x/y])
= \{e_d[x/y] \vert e_d \in \mathsf{Dist}(e[x/y]) \}\]
\[\mathsf{FreeVariables}(\mathsf{Dist}(e)) =  
\bigcup\limits_{e_d \in \mathsf{Dist}(e) \cup \{e\}}\mathsf{FreeVariables}(e_d) \]

\Comment{
Because of the nondeterminism associated with stochastic choices,
the execution strategy matters for the semantics of the language. 
We use call by value as the execution strategy and forbid the reduction
of expressions within a lambda.  Figure~\ref{fig:execcalc} presents the
execution strategy. 
{\begin{figure}
\[
    \begin{array}{c}
    \begin{array}{cc}
        \infra{}
        {x \rightarrow x}
            &
        \infra{}
        {\lambda.x~e \rightarrow \lambda.x~e}
    \end{array}
        \\
        \\
        \begin{array}{ccc}
        \infra{e \rightarrow e'\\
        e_v \in \mathsf{Dist}(e')\\
        e_v \rightarrow e'_v}
        {\mathsf{Dist}(e) \rightarrow e'_v}
        &
        \infra{e_1 \rightarrow e'_1\\
        e_2 \rightarrow e'_2\\
        e_1 \neq \lambda.x~e}
        {(e_1~e_2) \rightarrow (e'_1~e'_2)}
        &
         \infra{e_1 \rightarrow \lambda.x~e\\
        e_2 \rightarrow e'_2\\
         e[e'_2/x] \rightarrow e'}
        {(e_1~e_2) \rightarrow e'}
        \\
            \\
        \infra{}
        {\emptyset \rightarrow \emptyset}
        &
        \infra{e \rightarrow e'\\
        p[e'/x] \rightarrow p'}
        {\assume{x}{e};p \rightarrow p'}
        &
        \infra{
            e \rightarrow e' ~~~~
        p \rightarrow p'
        }
        {\observed{e}{e_v};p \rightarrow \\ 
        \observed{e'}{e_v};p'}
    \end{array}
    \end{array}
\]
    \caption{Execution Strategy for Probabilistic Programs}
    \label{fig:execcalc}
\end{figure}}}

\noindent{\bf Traces:} 
When a program executes, it produces an {\em execution trace} (Figure~\ref{fig:traces}). 
This trace records the executed
sequence of assume and observe commands, including the value of
each evaluated (sub)expression. It also assigns a unique 
identifier to each evaluated (sub)expression and stochastic choice. These identifiers
will be later used to construct a dependence graph used to define the subproblem given
a set of stochastic choices in the subproblem. 

\begin{figure}[h]
\[
\begin{array}{rcl}
    v \in V &:=&  x
    ~\vert~ \tup{\lambda.x~e, \sigma_v, \sigma_{id}}
    ~\vert~ (v_1~v_2)\\
    aa \in aA &:=& \perp~
      \vert~ x=ae\\
      
    ae \in aE &:=& \ndi{x }{x}{id}
    ~\vert~ \ndi{x(id')}{v}{id}\\
    &\vert& \ndi{\lambda.x~e}{v}{id}
    ~\vert~ \ndi{(ae_1~ae_2)aa}{v}{id}\\
    &\vert& \ndi{\mathsf{Dist}(ae\#id') = ae'}{v}{id}\\
    
as \in aS &:=& \assume{x}{ae}
~\vert~ \observed{ae\#id}{e_v}\\

t \in T &:=& \emptyset
~\vert~ as;t
\end{array}
\]
\caption{Traces}
    \label{fig:traces}
\end{figure}

We define the execution, including the generation of valid traces $t$, with the 
transition relation $\Rightarrow_s \subseteq \Sigma_v \times \Sigma_{id} \times 
P \rightarrow T$ (Figure~\ref{fig:validtraces}). Conceptually, 
the transition relation executes program $p$ 
under the environment $\sigma_v, \sigma_{id}$ to obtain a trace $t$, where 
$\sigma_v : Vars \rightarrow V$ and $\sigma_{id} : Vars \rightarrow ID$. 
$\sigma_v$ is  a map from variable name to its corresponding assigned value, 
whereas $\sigma_{id}$ gives the $id$ of the expression which assigned this value
to that variable. Because of the nondeterminism associated with stochastic choices,
the execution strategy matters for the semantics of the language. 
We use call by value as the execution strategy and forbid the execution
of expressions within a lambda. 

Given a program $p$, we define the set of all valid traces which can be obtained 
by executing $p$ as $T_p = \mathsf{Traces}(p)$.
\[t \in \mathsf{Traces}(p) \iff \emptyset, \emptyset \vdash p \Rightarrow_s t\]

\begin{figure}
    \begin{subfigure}{\textwidth}
    \[
        \begin{array}{c}
            \infral{
                \Nid{} ~~~~~
                y \notin \dom \sigma_v\\
            }
            {\SI{}{} y \Rightarrow_s y, id, \ndi{y}{y}{id} }{}
                \\
                \\
            \infral{
                \Nid{} ~~~~~
                x \in \dom \sigma_v\\
            }
            {\SI{}{} x \Rightarrow_s \sigma_v(x), id, \ndi{x(\sigma_{id}(x))}{\sigma_v(x)}{id}}
            {}
            \\
            \\
            \infral{
                \Nid{}\\
                \sigma'_v = \mathsf{RestrictKeys}(\sigma_v, \mathsf{FreeVariables}(\lambda.x~e))\\
                \sigma'_{id} = \mathsf{RestrictKeys}(\sigma_{id}, \mathsf{FreeVariables}(\lambda.x~e))\\
                v = \tup{ \lambda.x~e, \sigma'_v, \sigma'_{id} }\\
            }
            {\SI{}{} \lambda.x~e \Rightarrow_s v, id,  
            \ndi{\lambda.x~e}{v}{id}}{}
            \\
            \\
            \infral{
                \Nid{} ~~~~~
                \Nid{'}\\
                \SI{}{} e \Rightarrow_s v, id_e, ae\\
            e'_v \in \mathsf{Dist}(v) ~~~~~
            \SI{}{} e'_v \Rightarrow_s v, id_v, ae_v\\
            }
            {\SI{}{} \mathsf{Dist}(e) \Rightarrow_s  v, id, 
            \ndi{\mathsf{Dist}(ae\#id') = ae_v}{v}{id} }
            {}
            \\
            \\
            \infral{
                \Nid{} ~~~~
                \Nst{}\\
                \SI{}{} e_1 \Rightarrow_s \tup{\lambda.y~e, \sigma'_v,
                \sigma'_{id}}, id_1, ae_1 \\
                \SI{}{} e_2 \Rightarrow_s v', id_2, ae_2\\
                \SI{'[x \rightarrow v']}{'[x \rightarrow id_2]} 
                e[x/y] \Rightarrow_s v, id_e, ae_e\\
            }
            {\SI{}{} (e_1~e_2) \Rightarrow_s v, id, \ndi{(ae_1~ae_2)x=ae_e}{v}{id} }{}
            \\
            \\
            \infral{
                \Nid{}\\
                \SI{}{} e_1 \Rightarrow_s v_1, id_1, ae_1 ~~~~~
                \SI{}{} e_2 \Rightarrow_s v_2, id_2, ae_2\\
                v_1 \neq \tup{\lambda.x~e, \sigma'_v, \sigma'_{id}} ~~~~
                v = (v_1~v_2)\\
            }
            {\SI{}{} (e_1~e_2) \Rightarrow_s v, id, \ndi{(ae_1~ae_2)\perp}{v}{id}}{}
                    \end{array}
    \]
        \caption{Executing expressions, $\Rightarrow_s \subseteq \Sigma_v \times \Sigma_{id} \times E 
        \rightarrow V \times ID \times aE$}
\end{subfigure}
\begin{subfigure}{\textwidth}
    \[
        \begin{array}{c}
            \\
            \infral{}
            {\SI{}{} \emptyset \Rightarrow_s \emptyset}{}
            \\
            \\
            \infral{
                \SI{}{} e \Rightarrow_s v, id, ae ~~~~~
                \SI{[x\rightarrow v]}{[x \rightarrow id]} p
                \Rightarrow_s t
            }
            {\SI{}{} \assume{y}{e};p \Rightarrow_s 
            \assume{x}{ae};t}{}
            \\
            \\
            \infral{
                \Nid{}\\
                \SI{}{} e \Rightarrow_s e'_v , id_e, ae ~~~~~
                \SI{}{} p \Rightarrow_s t
            }
            {\SI{}{} \observed{e}{e_v};p \Rightarrow_s 
             \observed{ae\#id}{e_v};t}{}
        \end{array}
    \]
    \caption{Executing Programs, 
    $\Rightarrow_s \subseteq \Sigma_v \times \Sigma_{id} \times P \rightarrow T$}
\end{subfigure}
\caption{Valid Traces}
\label{fig:validtraces}
\end{figure}

Given a trace, we can drop the computed values and 
assigned $id$s and reroll the augmented expressions to recover the underlying program.
The transition relation $\Rightarrow_r \subseteq T \rightarrow P$ (Figure~\ref{fig:tracestoprog})
formalizes this procedure. Given a trace $t$, we define 
\[p = \mathsf{Program}(t) \iff t \Rightarrow_r p\]
Note that $\forall~t,p.~t \in \mathsf{Traces}(p) \implies p = \mathsf{Program}(t)$. The
reverse may not be true as there are additional constraints that valid traces
must satisfy. Two traces are equivalent if and only if they differ at 
most in the choice of unique identifiers selected for each augmented expression and stochastic choice.

\begin{figure}
    \begin{subfigure}{\textwidth}
            \[
                \begin{array}{c}
                \begin{array}{cc}
    \infral{}
                    {\ndi{x}{x}{id} \Rightarrow_r x}{}
    &
 \infral{}
                    {\ndi{x(id')}{v}{id} \Rightarrow_r x }{}
                    
                    \\
                    \\
    \infral{}
    {\ndi{\lambda.x~e}{v
                    }{id} \Rightarrow_r \lambda.x~e }{}
                &
                    \infral{
ae_1 \Rightarrow_r e_1 ~~~~~
ae_2 \Rightarrow_r e_2\\
}
                    {\ndi{(ae_1~ae_2)aa}{v}{id} \Rightarrow_r (e_1~e_2) }{}
                \end{array}
                    \\
\infral{ae \Rightarrow_r e}
                    {\ndi{\mathsf{Dist}(ae\#id') = ae'}{v}{id} \Rightarrow_r \mathsf{Dist}(e)}{}
                \end{array}
            \]
        \caption{Rolling back Augmented Expressions, $\Rightarrow_r \subseteq aE \rightarrow E$}
    \end{subfigure}
    \begin{subfigure}{\textwidth}
\[
    \begin{array}{c}
\begin{array}{cc}
     \infral{}
    {\emptyset \Rightarrow_r \emptyset}{}
    &
\infral{
ae \Rightarrow_r e
~~~~~~~
t \Rightarrow_r p\\
}
    {\assume{x}{ae};t \Rightarrow_r \assume{x}{e};p }{}
\end{array}
    \\
        \\
\infral{
ae \Rightarrow_r e
    ~~~~~~~~
    t \Rightarrow_r p\\
}
    {\observed{ae\#id}{e_v};t \Rightarrow_r  \observed{e}{e_v};p }{}
\end{array}
\]
        \caption{Rolling back Traces, $\Rightarrow_r \subseteq T \rightarrow P$}
    \end{subfigure}
\caption{Rolling back Traces to Probabilistic Program}
\label{fig:tracestoprog}
\end{figure}

\Comment{
Given a trace $t$, $\mathsf{Program}(t)$ is the Venture Generative 
program $p$  which generated the trace. Hence $p = \mathsf{Program}(t)$ if
$t \Rightarrow_r p$.

Note that $t \in \mathsf{Traces}(p) \implies p = \mathsf{Program}(t)$ 
but the reverse may not be true (since $t$ may not be  a valid trace). }

\Comment{
\noindent{\bf Dependence Graphs:} 
Dependence graphs capture two kinds of dependences --- data dependences
and control dependences. There is a data dependence from 
expression $e_1$ to expression $e_2$ when the value of $e_2$ is computed
directly from $e_1$. There is a control dependence from expression $e_1$ to expression $e_2$
when $e_2$ was computed as a result of evaluating the lambda expression $e_1$. 
Given a trace $t \in T$ we define the dependence graph $\NDE{}{}{} = 
\mathsf{Graph}(t)$. The dependence graph makes the data and existential dependencies 
in $t$ explicit. 
}

\begin{figure}
    \begin{subfigure}{\textwidth}
\[
\begin{array}{c}
    \begin{array}{cc}
    \infra{}
    {\ndi{x}{x}{id} \Rightarrow_g id, \tup{\{id \rightarrow \perp \},
    \emptyset, \emptyset} }
        &
    \infra{}
    {\ndi{x(id')}{v}{id} \Rightarrow_g id , 
    \tup{\{id \rightarrow \perp \}, 
    \{ \tup{id' , id} \}, \emptyset}}
    \end{array}
    \\
    \\
    \infra{}
    {\ndi{\lambda.x~e}{\tup{\lambda.x~e, \sigma_v, \sigma_{id}}}{id} 
    \Rightarrow_g id, 
     \tup{ \{ id \rightarrow \perp \} 
     ,\emptyset, \emptyset } }
    \\
    \\
    \infra{
        ae_1 \Rightarrow_g id_1, \NDE{_1}{_1}{_1}
        ~~~~~~~~
        ae_2 \Rightarrow_g id_2, \NDE{_2}{_2}{_2}
        ~~~~~~~~
        ae \Rightarrow_g id_e, \NDE{_e}{_e}{_e}\\
        \NDE{}{}{} = \NDE{_1 \cup \cN_2 \cup \cN_e}
        {_1 \cup \cD_2 \cup \cD_e}{_1 \cup \cE_2 \cup \cE_e}
    }
    {\ndi{(ae_1~ae_2)x=ae}{v}{id} \Rightarrow_g \\ 
    id, \NDE{[id \rightarrow \perp ]}
    {\cup \{ \tup{id_1 ,id} ,\tup{id_e,  id } \}}
    {\cup \{ \tup{id_1 ,id_n} \vert id_n \in \dom \cN_e\} }
    }
    \\
    \\
    \infra{
        ae_1 \Rightarrow_g id_1, \NDE{}{}{}
        ~~~~~~~~~~
        ae_2 \Rightarrow_g id_2, \NDE{'}{'}{'}\\
    }
    {
        \ndi{(ae_1~ae_2)\perp}{v}{id} \Rightarrow_g \\ id , 
        \NDE{\cup \cN'[id \rightarrow \perp ]}
        {\cup \cD'\{\tup{id_1, id}, \tup{id_2, id} \} }{\cup \cE}
    }
    \\
    \\
    \infra{ae \Rightarrow_g id_e, \NDE{}{}{} 
    ~~~~~~~~~~
    ae' \Rightarrow_g id'_e, \NDE{'}{'}{'}\\
    \NDE{_r}{_r}{_r} = \\  \NDE{\cup \cN' 
    [id' \rightarrow \mathsf{Sample}] }{\cup \cD'\cup \{\tup{id_e ,id'} \} }
    {\cup \cE'\cup 
    \{ \tup{id', id_n} \vert id_n \in \dom \cN'\} }\\
    }
    {\ndi{\mathsf{Dist}(ae\#id') = ae'}{v}{id} 
    \Rightarrow_g \\ id, \NDE{_r[
        id \rightarrow \perp] 
    }{_r\cup \{ \tup{id'_e ,id}, \tup{id', id} \}}{_r}}
    \\
\end{array}
\]
\caption{Dependence Graph generation for augmented Expressions $\Rightarrow_g \subseteq aE \rightarrow ID \times \NDE{}{}{}$}
\end{subfigure}
\begin{subfigure}{\textwidth}
    \[
        \begin{array}{c}
            \infra{ae \Rightarrow_g id, \NDE{}{}{} }
    {\assume{x}{ae} \Rightarrow_g \NDE{}{}{} }
    \\
            \\
    \infra{ae \Rightarrow id', \NDE{}{}{} }
    {\observed{ae\#id}{e_v} \Rightarrow_g \NDE{
        [id \rightarrow \mathsf{Sample}] 
    }{\cup \{ \tup{id' ,id} \} }{} }
    \end{array}
\]
    \caption{Dependence Graph generation for augmented Statements, $\Rightarrow_g \subseteq aS \rightarrow \NDE{}{}{} $}
\end{subfigure}
\begin{subfigure}{\textwidth}
    \[
        \begin{array}{cc}
        \\
            \infra{}
    {\emptyset \Rightarrow_g \tup{\emptyset, \emptyset, \emptyset}}
    &
    \infra{as \Rightarrow_g \NDE{_s}{_s}{_s}
    ~~~~~~~~~
            t \Rightarrow_g \NDE{}{}{} }
    {as;t \Rightarrow_g \NDE{\cup \cN_s}{\cup \cD_s}{\cup \cE_s}} 
\end{array}
\]
    \caption{Dependence Graph generation for Traces, $\Rightarrow_g \subseteq T \rightarrow \NDE{}{}{}$ }
\end{subfigure}
\caption{Dependency Graph for a Trace $t$}
\label{fig:depgraph}
\end{figure}

\noindent{\bf Dependence Graphs:} 
Given a trace $t$, we define the dependence graph $\NDE{}{}{} =
\mathsf{Graph}(t)$ as a 3-tuple $\NDE{}{}{}$ where $\cN : \mathsf{ID}
\rightarrow \{\perp, \mathsf{Sample}\}$ is a map from   
$ID $ to either $\perp$ (when the corresponding augmented expression
for an $id \in ID$ is a deterministic computation) or $\mathsf{Sample}$ 
(when the augmented expression for an $id \in ID$ makes a stochastic choice).

$\cD \subseteq \mathsf{ID} \times \mathsf{ID}$ are data dependence edges.
There is a data dependence edge $\tup{id_1, id_2} \in \cD$ if the value
of the augmented expression  $id_2$ directly depends on the 
augmented expression $id_1$.
$\cE \subseteq \mathsf{ID} \times \mathsf{ID}$ are existential edges.
There is a existential edge $\tup{id_1, id_2} \in \cE$ if the value
of the augmented expression $id_1$ controls whether or not an 
augmented expression $id_2$ executed. For example, in a lambda application 
$(ae_1~ae_2)x=ae_3$, all augmented expressions in $ae_3$ were executed only
because of the value of $ae_1$. Changing the value of $ae_1$ would require
dropping the augmented expression $ae_3$ and recomputing another 
expression based on the new value of $ae_1$.

We formalize the dependence graph generation procedure as a transition
relation $\Rightarrow_g \subseteq T \rightarrow \NDE{}{}{}$ (Figure~\ref{fig:depgraph}). 
We use the shorthand
$\NDE{}{}{} = \mathsf{Graph}(t)$ if $\NDE{}{}{}$ is the dependence graph 
for trace $t$ i.e. $t \Rightarrow_g \NDE{}{}{}$.

\noindent{\bf Valid Subproblems:}
Subproblem inference must 1) change only the identified subproblem and not the 
enclosing trace while 2) producing a valid trace for the full probabilistic
program. Valid subproblems must therefore include all parts of the trace
that may change if any part of the subproblem changes. We formalize this
requirement as follows. 

\noindent Given a trace $t$ with dependence graph 
$\NDE{}{}{} = \mathsf{Graph}(t)$, a valid subproblem $\cS \subseteq \dom \cN$
must satisfy two properties: 1) there are no outgoing existential edges and
2) all outgoing data dependence edges must terminate at a 
stochastic choice ($\mathsf{Sample}$ node).

The first property ensures that parts of the trace which were executed due to 
values of expressions in the subproblem are also part of the subproblem. An 
example of this is lambda evaluation $(ae_1~ae_2)x=ae_3$. If the value of $ae_1$ 
can be changed by the subproblem inference, $ae_3$ may or may not exist. Hence 
$ae_3$ should be within the subproblem to ensure that the inference algorithm 
can change it if necessary.

The second property ensures that any change made by the subproblem inference can be 
absorbed by a stochastic choice. For example, when the internal parameter of a $\mathsf{Dist}$
changes, the change can be absorbed by changing the probability of the trace
to account for the change in the probability of the value generated by
the execution of the absorbing $\mathsf{Dist}$ node. The changes are
absorbed by the stochastic choice and do not propagate further into the remaining
parts of the trace outside the subproblem.  We formalize these two properties as follows:
\begin{itemize}
\item $\forall id \in \cS.~\tup{id, id_o} \in \cE  \implies id_o \in \cS$
\item  $\forall id \in \cS.~\tup{id, id_o} \in \cD \wedge id_o \in \dom \cN 
- \cS \implies \cN(id_o) = \mathsf{Sample} $
\end{itemize}

The absorbing set $\cA \subseteq \dom \cN - \cS$ of a subproblem $\cS$ 
is the set of stochastic choices whose value directly depends on the 
nodes in the subproblem i.e. $\cA = \{id_a \vert id_a \in \dom \cN
- \cS \wedge \exists~id_i \in \cS.~\tup{id_i, id_o} \in \cD\}$.
The input boundary $\cB \subseteq \dom \cN - \cS$ of a subproblem 
$\cS$ is the set of nodes on which the subproblem directly depends, i.e.,
$\cB = \{id_b \vert id_b \in \dom \cN - \cS \wedge \forall~id_i \in \cS.
~\tup{id_b, id_i} \in \cD \}$.

\Comment{
Since given a trace $t$ and a valid subproblem $\cS$, an absorbing set 
$\cA$ is naturally defined, from this point on we will call the tuple 
$\cS, \cA$ as a valid subproblem for trace $t$.
}

\noindent{\bf Entangled Subproblem Inference:}
Following~\cite{mansinghka2018probabilistic}, we define {\em entangled subproblem inference} 
using the $\mathsf{infer}$ procedure~\cite{mansinghka2018probabilistic},
which takes as parameters a subproblem selection strategy $\mathsf{SS}$, 
an inference tactic $\mathsf{IT}$, and an input trace $t$. The
subproblem inference mutates $t$ to produce a new trace $t'$. 

\[
    \infra{\mathsf{SS}(t) = \cS ~~~~~
        t' = \mathsf{IT}(t, \cS)\\
        t' \in \mathsf{Traces}(\mathsf{Program}(t)) ~~~~~
        \SAC t \equiv t'}
    {\mathsf{infer}(\mathsf{SS}, \mathsf{IT}, t) \Rightarrow_i t'}
\]

This formulation works with arbitrary subproblem selection strategies $\mathsf{SS}$. 
The requirement is that, given a trace $t$, $\mathsf{SS}$ must
produce a valid subproblem $\cS$ over $t$. In practice, one way to satisfy this
requirement is to allow the programmer to specify a (potentially arbitrary) 
set of stochastic choices that must be in the subproblem, with the language
implementation completing these choices into a valid subproblem~\cite{mansinghka2018probabilistic}.

We also work with inference algorithms $\mathsf{IT}$ that 
take as input a full program trace $t$ and a valid subproblem $\cS$ and return
a mutated full program trace $t'$. We require that the output trace $t'$ 
1) is from the same program as the trace $t$ and 
2) $t'$ differs from $t$ only in a) the stochastic choices from the subproblem $\cS$ 
and b) the deterministic computations that depend on these stochastic choices.
We formalize these constraints as 
\begin{itemize}
    \item $t' \in \mathsf{Traces}(\mathsf{Program}(t))$
    \item $\SAC t \equiv t'$
\end{itemize}
Figure~\ref{fig:equivalence} presents 
the definition of $\equiv$. Note that operating with entangled subproblems
forces the inference tactic $\mathsf{IT}$ to take the full program 
trace $t$ as a parameter even though it must modify at most only the subproblem. 

\begin{figure*}
    \begin{subfigure}{\textwidth}
    \[
        \begin{array}{c}
        \begin{array}{cc}
            \infral{
            }
            {\SAC \ndi{x}{x}{id} \equiv \ndi{x}{x}{id'}}{}
            &
            \infral{
            }
            {\SAC \ndi{x(id_v)}{v}{id} \equiv  \ndi{x(id'_v)}{v'}{id'} }{} 
            \end{array}
            \\
            \\
            \infral{
            }
            {\SAC \ndi{\lambda.x~e}{v}{id} \equiv  \ndi{\lambda.x~e}{v'}{id'} }{}
            \\
            \\
            \begin{array}{cc}
                        \infral{
                \ID(ae_1) \notin \cS ~~~~~ 
            \SAC ae_1 \equiv ae'_1 \\
            \SAC ae'_2 \equiv ae'_2 ~~~~~
            \SAC ae'_3 \equiv ae'_3\\
            }
            {\SAC \ndi{(ae_1~ae_2)x=ae_3}{v}{id} \equiv
                \ndi{(ae'_1~ae'_2)x=ae'_3}{v'}{id}
                }{}
            \\
                \\
            \infral{
                \ID(ae_1) \notin \cS\\
            \SAC ae_1 \equiv ae'_1 ~~~~~
            \SAC ae'_2 \equiv ae'_2 \\
            }
            {\SAC \ndi{(ae_1~ae_2)\perp}{v}{id} \equiv
                \ndi{(ae'_1~ae'_2)\perp}{v'}{id'}
                }{}
            \end{array}
            \\
            \\
            \begin{array}{ccc}
            \infral{
                \ID(ae_1) \in \cS\\
            \SAC ae_1 \equiv ae'_1 ~~~~~
            \SAC ae'_2 \equiv ae'_2 
            }
            {\SAC \ndi{(ae_1~ae_2)aa}{v}{id} \equiv
               \ndi{(ae'_1~ae'_2)aa'}{v'}{id'}
                }{}
            \\
                \\
                \infral{
                id_e \notin \cS\\
                \SAC ae \equiv ae' ~~~~~
                \SAC ae_v \equiv ae'_v\\
            }
            {\SAC \ndi{\mathsf{Dist}(ae\#id_e)=ae_e}{v}{id} \equiv 
                 \ndi{\mathsf{Dist}(ae'\#id_e)=ae'_e}{v'}{id'} }{}
            \\
                \\
            \infral{
                id_e \in \cS ~~~~~
                \SAC ae \equiv ae'\\
            }
            {\SAC \ndi{\mathsf{Dist}(ae\#id_e)=ae_e}{v}{id} \equiv 
                 \ndi{\mathsf{Dist}(ae'\#id'_e)=ae'_e}{v'}{id'} }{}
            \\
            \\
                   \end{array}
    \end{array}
    \]
\caption{Equivalence Check over augmented expressions, $\equiv \subseteq \cP(ID) \times aE \times aE$}
\end{subfigure}
    \begin{subfigure}{\textwidth}
        \[
            \begin{array}{c}
            \begin{array}{cc}
            \infral{}
                {\SAC \emptyset \equiv \emptyset}{}
                &
                \infral{
            \SAC ae \equiv ae' ~~~~~
                \SAC t \equiv t'\\
            }
                {\SAC \assume{x}{ae};t \equiv \assume{x}{ae'};t'}{}
            \end{array}
            \\
            \\
                \infral{
            \SAC ae \equiv ae' ~~~~~
                \SAC t \equiv t'\\
            }
                {\SAC \observed{ae\#id}{e_v};t \equiv  \observed{ae'\#id}{e_v};t'}{}
            \end{array}
        \]
        \caption{Equivalence check over traces, $\equiv \subseteq \cP(ID) \times T \times T $}
    \end{subfigure}
\caption{Equivalence Check for correct Inference}
\label{fig:equivalence}
\end{figure*}

\section{Independent Subproblem Inference}
\label{sec:indepsubinfer}

The basic idea of independent subproblem inference is to extract an
independent subtrace $t_s$ from the original trace $t$ given a subproblem $\cS$,
perform inference over the extracted subtrace $t_s$ to obtain a new trace $t_s'$,
then stitch $t_s'$ back into $t$ to obtain a new trace for the full program. Here, consistent with standard inference
techniques for probabilistic programs~\cite{wingate2011lightweight,mansinghka2018probabilistic}, $t_s$ and $t_s'$ are
valid traces of the same program $p_s$ (the subprogram for the subtraces $t_s$
and $t_s'$). The key challenge is converting the entangled subproblem (which 
is typically incomplete and therefore not a valid trace of any program) into
a valid trace by transforming the subproblem to include external dependences
and correctly scope both internal and external dependences in the extracted
trace without giving the inference algorithm access to any external stochastic
choices (including latent choices nested inside certain lambda expressions
which would otherwise override choices outside the subproblem) which it must not change. 

\Comment{
In this section, we present Subprogram based inference metaprogramming.
Given a trace $t$ and a valid subproblem $\cS$ on that trace, subproblem
based inference aims to restrict inference algorithms to only mutate 
stochastic choices within the subproblem. A developer of inference 
algorithms has to be careful to develop algorithms which work with the knowledge 
and constraints associated with subproblem based inference.

Subprogram based inference on the other hand, given a trace $t$ and a valid 
subproblem $\cS$, aims to mutate the trace into a trace $t_s$ of a different program $p_s$. 
The mutated trace and the underlying program only contains stochastic choices 
which are within the subproblem (and unexecuted lambda expressions), constraining
stochastic choices outside the subproblem. 
We call this mutated trace, a {\em subtrace}. Note that the underlying
program of the subtrace will be different from the original program. We 
call this program {\em subprogram}. 

The goal is to use run of the mill inference algorithms over the subtrace $t_s$ 
(in theory on subprogram $p_s$) 
to achieve a new subtrace $t'_s$. We can then stitch the changes made into 
$t'_s$ back into trace $t$ to achieve an inferred trace $t'$. The purpose of this 
exercise is to map inference over the subprogram $p_s$ to subproblem inference
over $t$ (i.e. $t$ and $t'$ differ only in stochastic choices in $\cS$). We 
formalize this process using extraction and stitching procedures which We will
describe next.
}

\noindent{\bf Extract Trace:}
We define the extraction procedure $t_s = \mathsf{ExtractTrace}(t, \cS)$ 
using the transition relation 
$\Rightarrow_{ex} \subseteq \cP(ID) \times T \rightarrow T$ 
(Figure~\ref{fig:extractrace}). 
The extraction procedure removes $\mathsf{Dist}(ae\#id) = ae_v$ 
augmented expressions which are not within the subproblem and converts them 
into $\mathsf{observe}$ statements. This transformation constrains 
the value of these stochastic choices to the values present in the original trace. 
It leaves the stochastic choices in the subproblem in place and therefore
accessible to the inference algorithm. 

\begin{figure*}
    \begin{subfigure}{\textwidth}
    \[
    \begin{array}{c}
        \begin{array}{cc}
        \infral{}
        {\SAC \ndi{x(id')}{v}{id} 
            \Rightarrow_{ex}  \ndi{x(id')}{v}{id}, \emptyset }{}
        \\
            \\
        \infral{}
        {\SAC \ndi{x}{x}{id} \Rightarrow_{ex} 
            \ndi{x}{x}{id}, \emptyset }{}
        \end{array}
        \\
        \\
        \infral{}
        {\SAC \ndi{\lambda.x~e}{v }{id} 
         \Rightarrow_{ex}
         \ndi{\lambda.x~e}{v}{id}, 
            \emptyset}{}
         \\
        \\
            \begin{array}{c}
        \infral{
            \ID(ae_1) \in \cS ~~~~~
        \SAC ae_1  \Rightarrow_{ex} ae'_1 , t_s ~~~~~
        \SAC ae_2 \Rightarrow_{ex} ae'_2, t'_s\\
        }
        {\SAC \ndi{(ae_1~ae_2)aa}{v}{id} \Rightarrow_{ex} 
         \ndi{(ae'_1~ae'_2)aa}{v}{id}, t_s;t'_s}
                {}
                \\
        \\
        \infral{
            \ID(ae_1) \notin \cS ~~~~~
            \SAC ae_1 \Rightarrow_{ex} ae'_1, t_s~~~~
            \SAC ae_2  \Rightarrow_{ex} ae'_2, t'_s\\
        }
        {\SAC \ndi{(ae_1~ae_2)\perp}{v}{id}, t_p \Rightarrow_{ex}
         \ndi{(ae'_1~ae'_2)\perp}{v}{id}, t_s;t'_s
                }{}
            \\
            \\
                \infral{
            \ID(ae_1) \notin \cS ~~~~~
            \SAC ae_1 \Rightarrow_{ex} ae'_1, t_s\\
            \SAC ae_2  \Rightarrow_{ex} ae'_2, t'_s ~~~~~
            \SAC ae_3
            \Rightarrow_{ex} ae'_3, t''_s\\
            \Nst{}\\
            t'''_s = t_s;\assume{x}{ae'_1};t'_s;\assume{y}{ae'_2};t''_s\\
        }
        {\SAC \ndi{(ae_1~ae_2)y= ae_3}{v}{id}, t_p \Rightarrow_{ex}
         ae'_3, t'''_s
                }{}
                \end{array}
        \\
        \\
        \begin{array}{c}
        \infral{
            id' \notin \cS ~~~~
        \SAC ae \Rightarrow_{ex} ae', t_s ~~~~~
            ae_v \Rightarrow_r e_v ~~~~~
            \SAC ae_v \Rightarrow_{ex} ae'_v, 
            t'_s
        }
        {\SAC \ndi{\mathsf{Dist}(ae\#id')= ae_v}{v}{id}  
            \Rightarrow_{ex}  ae'_v, t_s;\observed{ae'\#id'}{e_v};t'_s }{}
        \\
            \\
        \infral{
        id' \in \cS ~~~~~
        \SAC ae \Rightarrow_{ex} ae', t_s \\
        }
        {\SAC \ndi{\mathsf{Dist}(ae\#id') = ae_v}{v}{id} \Rightarrow_{ex} 
            \ndi{\mathsf{Dist}(ae'\#id') = ae_v}{v}{id}, t_s }{}
        \end{array}
    \end{array}
        \]
        \caption{Extracting subtrace from an Augmented Expression, 
        $\Rightarrow_{ex} \subseteq \cP(ID) \times aE \rightarrow 
        aE \times T$}
    \end{subfigure}
    \begin{subfigure}{\textwidth}
        \[
            \begin{array}{c}
            \begin{array}{cc}
        \infral{}
                {\SAC \emptyset \Rightarrow_{ex} \emptyset}{}
        &
        \infral{
            id = \ID(ae)~~~~~\\
            \SAC ae  \Rightarrow_{ex} ae', t_s ~~~~
        \SAC t \Rightarrow_{ex} t'_s}
    {\SAC \assume{x}{ae};t \Rightarrow_{ex} t_s;\assume{x}{ae'};t'_s }{}
            \end{array}
                \\
                \\
        \infral{
        \SAC ae \Rightarrow_{ex} ae', t_s ~~~~~
        \SAC t \Rightarrow_{ex} t'_s
        }
                {\SAC \observed{ae\#id}{e_v};t  \Rightarrow_{ex}  t_s;\observed{ae'\#id}{e_v};t'_s}{}
    \end{array}
\]
        \caption{Extracting subtrace from a given trace, $\Rightarrow_{ex} \subseteq 
        \cP(ID) \times T \rightarrow T$}
\end{subfigure}
\caption{Extraction Relation}
\label{fig:extractrace}
\end{figure*}

\begin{figure*}
    \begin{subfigure}{\textwidth}
    \[
\begin{array}{c}
    \infra{
    }
    {
        \SAC \ndi{x}{v}{id} ,  \ndi{x}{v'}{id'} , \emptyset \Rightarrow_{st} 
         \ndi{x}{v'}{id'}
    }
    \\
    \\
    \infra{}
    {
        \SAC \ndi{x(id_v)}{v}{id},  \ndi{x(id'_v)}{v'}{id'}, \emptyset \Rightarrow_{st}
         \ndi{x(id'_v)}{v'}{id'}
    }
    \\
    \\
    \infra{
    }
    {\SAC \ndi{\lambda.x~e}{v }{id'}, 
     \ndi{\lambda.x~e}{v' }{id}, \emptyset
    \Rightarrow_{st} 
     \ndi{\lambda.x~e}{v' }{id}}
        \\
    \\
    \infra{
        \ID(ae'_1) \in \cS ~~~~~
        \SAC ae'_2, ae_2 , t''_p \Rightarrow_{st} ae''_2
        ~~~~~
        \SAC ae'_1, ae_1, t'_p \Rightarrow_{st} ae''_1\\
    }
    {\SAC \ndi{(ae'_1~ae'_2)aa'}{v'}{id'}, \ndi{(ae_1~ae_2)aa}{v}{id}, t'_p;t''_p \Rightarrow_{st}               
          \ndi{(ae''_1~ae''_2)aa}{v}{id} }
    \\
    \\
    \infra{
        \ID(ae'_1) \notin \cS ~~~~~
        \SAC ae'_3, ae_3, t'''_p \Rightarrow_{st} ae''_3\\
        \SAC ae'_2, ae_2 , t''_p \Rightarrow_{st} ae''_2  ~~~~~
        \SAC ae'_1, ae_1, t'_p \Rightarrow_{st} ae''_1\\
        t_p = t'_p;\assume{x}{ae_1};t''_p;\assume{y}{ae_2};t'''_p
    }
    { \SAC \ndi{(ae'_1~ae'_2)y=ae'_3}{v'}{id}, ae_3, t_p \Rightarrow_{st}               
          \ndi{(ae''_1~ae''_2)y=ae''_3}{\cV(ae''_3)}{id} }
    \\
    \\
    \infra{
        \ID(ae'_1) \notin \cS ~~~~~
        \SAC ae'_2, ae_2 , t''_p \Rightarrow_{st} ae''_2
        ~~~~~~~~~~~~
        \SAC ae'_1, ae_1, t'_p \Rightarrow_{st} ae''_1
    }
    { \SAC \ndi{(ae'_1~ae'_2)\perp}{v'}{id'}, \ndi{(ae_1~ae_2)\perp}{v}{id}, 
            t'_p;t''_p \Rightarrow_{st}               
          \ndi{(ae''_1~ae''_2)\perp}{v}{id}}
    \\
    \\
    \infra{ 
    id'_v \in \cS ~~~~~
    \SAC ae', ae , t_p \Rightarrow_{st} ae''
    }
    {\SAC \ndi{\mathsf{Dist}(ae'\#id'_v) = ae'_v}{v'}{id'} , 
    \ndi{\mathsf{Dist}(ae\#id_v) = ae_v}{v}{id}, t_p
    \Rightarrow_{st}
    \\ \ndi{\mathsf{Dist}(ae\#id_v) = ae_v}{v}{id}
    }
    \\
    \\
    \infra{ 
    id'_v \notin \cS ~~~~~
    \SAC ae'_v, ae_v, t''_p \Rightarrow_{st} ae''_v
       \\
    \SAC ae', ae , t'_p \Rightarrow_{st} ae''
    ~~~~ t_p = t'_p;\observed{ae\#id'_v}{e_v};t''_p
    }
    {\SAC \ndi{\mathsf{Dist}(ae'\#id_v) = ae'_v}{v'}{id'} , 
    ae_v, t_p
    \Rightarrow_{st} 
     \ndi{\mathsf{Dist}(ae\#id_v) = ae''_v}{\cV(ae''_v)}{id} 
    }
    \end{array}
    \]
    \caption{Stitching augmented expressions, $\Rightarrow_{st} \subseteq 
    \cP(ID) \times aE \times aE \times T \rightarrow aE 
    $}
\end{subfigure}
    \begin{subfigure}{\textwidth}
    \[
        \begin{array}{c}
        \begin{array}{ccc}
    \\
            \infra{}
    {\SAC \emptyset, \emptyset \Rightarrow_{st} \emptyset}
            &
    \infra{
        \SAC ae, ae' , t_p \Rightarrow_{st} ae''~~~~~
        \SAC t, t'_p \Rightarrow_{st} t'_s
    }
    {\SAC \assume{x}{ae};t , 
             t_p;\assume{x}{ae'};t'_p \Rightarrow_{st} 
             \assume{x}{ae''};t'_s}
    \end{array}
    \\ \\
            \infra{
        \SAC ae, ae', t_p \Rightarrow_{st} ae'' ~~~~~
        \SAC t, t'_p  \Rightarrow_{st} t'_s
    }
    {\SAC \observed{ae\#id}{e_v};t , 
             t_p;\observed{ae'\#id'}{e_v};t'_p \Rightarrow_{st} 
\\             \observed{ae''\#id}{e_v};t'_s}

        \end{array}
\]
        \caption{Stitching trace and subtrace, $\Rightarrow_{st} \subseteq 
        \cP(ID) \times T \times T \rightarrow T
        $ }
    \end{subfigure}
\caption{Stitching Transition Relation }
\label{fig:stitchtrace}
\end{figure*}

For augmented expressions of the form $(ae_1~ae_2)y=ae_3$, when 
$ae_1$ is within the subproblem, its value can change and hence 
existential edges place the augmented expressions in $ae_3$ within the 
subproblem. When $ae_1$ is not within the subproblem,
then some stochastic choices may or not be within the subproblem.  
If we keep the augmented expression as is, the inference algorithm may
unroll $ae_3$ and execute it again, changing some stochastic choices 
in $ae_3$ but not in the subproblem. If we modify $ae_3$, the 
constraint of $ae_3$ being a valid lambda application breaks. 
We solve this problem by introducing $\mathsf{assume}$ statements and
correctly scoping the resulting dependences.

\noindent{\bf Stitch Trace:}
Given a trace $t$, a valid subproblem $\cS$ over the trace,
and a subtrace $t_s$, the stitching procedure stitches back the trace $t_s$ 
into $t$ to get a new trace $t'$. 
We define the stitching procedure $t' = \mathsf{StitchTrace}(t, t_s, \cS)$
using a transition relation  
$\Rightarrow_{st} \subseteq \cP(ID) 
\times T \times T \rightarrow T$ (Figure~\ref{fig:stitchtrace}), 
where $\SAC t, t_s \Rightarrow_{st} t' \iff t' =\mathsf{StitchTrace}(t, t_s, \cS) $.
Stitching is the dual of extraction. It uses
the original trace to figure out the structure of the resultant trace, then 
stitches back the expressions to get a new trace $t'$.

\noindent{\bf Independent Inference:}
We define independent subproblem inference using the $\mathsf{infer}$ procedure.
$\mathsf{infer}$ takes as input a subproblem selection strategy
$\mathsf{SS}$, a trace $t$ and an inference tactic $\mathsf{IT}$. This 
differs from tangled inference in that inference tactic $\mathsf{IT}$
takes only the extracted subtrace as input and not the entire program trace. 
This approach enables the use of standard inference algorithms which are
designed to operate on complete traces (and not entangled subproblems). 

The new inference procedure works as described below:
\begin{equation}
\infra{
\mathsf{SS}(t) = \cS~~~~~
t_s = \mathsf{ExtractTrace}(t, \cS) ~~~~~ t'_s = \mathsf{IT}(t_s)\\
t'_s \in \mathsf{Traces}(\mathsf{Program}(t_s)) ~~~~~
t' = \mathsf{StitchTrace}(t, t'_s, \cS)
}
{\mathsf{infer}(\mathsf{SS}, \mathsf{IT}), t \Rightarrow_i t' }
\end{equation}

\noindent{\bf Soundness and Completeness:}
Given a trace $t$, a valid subproblem $\cS$, an inferred trace $t'$ and $t_s = 
\mathsf{ExtractTrace}(t, \cS)$,
we prove that extraction and stitching is sound
and complete.
Soundness in this context means that for 
all possible mutated subtraces $t'_s \in \mathsf{Traces}(\mathsf{Program}(t_s))$, 
the stitched trace $t' = \mathsf{StitchTrace}(t, t'_s, \cS)$ is a valid inferred trace.
Completeness means that for all possible inferred traces $t'$, there
exists a mutated subtrace $t'_s$ such that $t'_s \in \mathsf{Traces}(\mathsf{Program}(t_s))$
and $t' = \mathsf{StitchTrace}(t, t'_s, \cS)$.

We summarize the comparison between the entangled subproblem inference and independent
subproblem inference approaches in Figure~\ref{fig:subprogvsubprob}.
We present the theorems, lemmas and proofs of soundness and completeness 
in Appendix~\ref{append:sound} and \ref{append:comp}.

\begin{figure}
    \centering
\begin{tikzpicture}[shorten >=1pt,node distance=3cm,on grid, state/.style={circle,inner sep=2pt}]
\small
  \node[state]   (q_0)                {$t$};
  \node[state]   (q_1) [right=of q_0] {$t'$};
  \node[state]   (q_2) [below=of q_0] {$t, t_s$};
  \node[state]   (q_3) [right=of q_2] {$t, t'_s$};
    \path[->] (q_0) edge                node [above] {$\begin{array}{c}t' \in \mathsf{Traces}
    (\mathsf{Program}(t))\\ \SAC t \equiv t'\end{array} $} (q_1)
                  edge                node [left] {$t_s = \mathsf{ExtractTrace}(t, \cS)$} (q_2)
            (q_2) edge                node [below] {$t'_s \in \mathsf{Traces}(\mathsf{Program}(t_s))$} (q_3)
            (q_3) edge                node [right] {$t' = \mathsf{StitchTrace}(t, t'_s, \cS)$} (q_1);
\end{tikzpicture}
    \caption{Entangled vs Independent Subproblem Inference}
    \vspace{-.1in}
    \label{fig:subprogvsubprob}
\end{figure}
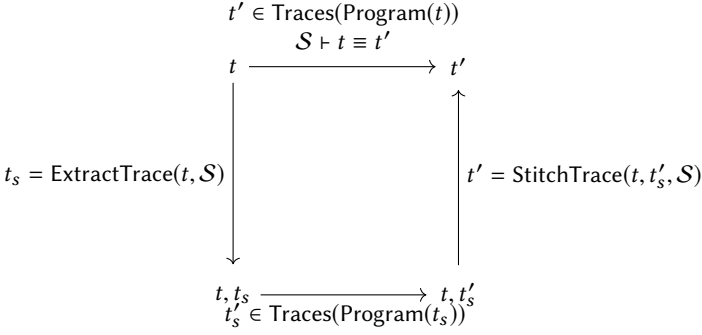

\section{Convergence of Stochastic Alternating Class Kernels}
\label{sec:stoch}

We next introduce the concept of class functions and
class kernels, which we use to prove the convergence of  hybrid inference algorithms based 
on asymptotically converging MCMC-algorithms. 
We present key definitions, theorems, and lemmas 
required to prove these results.

\subsection{Preliminaries}

We begin by introducing some key measure theory definitions ~\cite{geyer1998markov, 
meyn2012markov, tierney1994markov}. Readers familiar with measure theory may wish
to skip this subsection. 

\begin{definition}[\bf Topology]\label{def-topology}
    A {\bf Topology} on set $T$ is a collection $\cT$ of subsets of $T$ having the 
    following properties:
    \begin{enumerate}
        \item $\emptyset \in \cT$ and $T \in \cT$.
        \item $\cT$ is closed under arbitrary unions. i.e. For any collection $\{A_i \}_{i \in I}$, if for all $i \in I$,$A_i \in \cT$, then
            $\bigcup\limits_{i \in I} A_i \in \cT$.
        \item $\cT$ is closed under finitely many intersections. i.e. For any finite collection $\{A_i\}_{i \in I}$, if for all $i \in I$, 
            $A_i \in \cT$, then 
            $\bigcap\limits_{i \in I} A_i \in \cT$.
    \end{enumerate}
    
\end{definition}
Given a set $T$ and a topology $\cT$ defined on $T$, the pair $(T, \cT)$ is called a {\bf topological space}.
Given a topological space $(T, \cT)$, all sets $A \in \cT$ are called {\bf open sets}. 
From this point on, when we say $T$ is a topological space, we 
assume we are talking about any topology $\cT$ on $T$.

\begin{definition}[\bf $\sigma$-algebra]\label{def-sigma-algebra}Let $T$ be a set. A collection 
    $\Sigma$ of subsets of $T$ is a {\bf $\sigma$-field} or a 
    {\bf $\sigma$-algebra} over $T$ if and only if 
    $T \in \tau$ and $\tau$ is closed under countable unions, intersections and 
    complements, i.e.,
\begin{enumerate}
\item $T \in \Sigma$ and $\emptyset \in \Sigma$.
\item $A \in \Sigma$ implies $A^c \in \Sigma$.
\item $\Sigma$ is closed under countable unions, i.e. For any countable collection 
    $\{A_i\}_{i \in I}$, if for all $i \in I$, $A_i \in \Sigma$, then 
        $\bigcup\limits_{i \in I} A_i \in \Sigma$.
\end{enumerate}

    A {\bf measurable space} is a pair $(T, \Sigma)$ such that $T$ is a set and 
    $\Sigma$ is a $\sigma$-algebra over $T$.
\end{definition}

\begin{definition}[\bf Measure] $m : \Sigma \rightarrow \mathbb{R} \cup \{-\infty, \infty\}$ is a 
measure over a measurable space $(T, \Sigma)$ if 
\begin{enumerate}
    \item $m(A) \geq  m(\emptyset) = 0$ for all $A \in \Sigma$,
    \item For all countable collections $\{A_i\}_{i \in I}$ of pairwise disjoint
        sets in $\Sigma$, 
        \[m(\bigcup\limits_{i \in I} A_i ) = \sum\limits_{i \in I} m(A_i) \]
\end{enumerate}

\end{definition}

\noindent Given a measurable space $(T, \Sigma)$, a measure $\pi$ on $(T, \Sigma)$ is called 
{\bf probability measure} if $\pi(T) = 1$.
    We call the tuple $(T, \Sigma, \pi)$ a {\bf probability space}, if $\pi$ is
    a probability measure over the measurable space $(T, \Sigma)$.
    Given a set $T$, a collection of subsets $A_\alpha \subseteq T$ (not necessarily countable), 
    we denote the smallest $\sigma$-algebra $\Sigma$ such that $A_\alpha \in \Sigma$
    for all $\alpha$ by $\sigma(\{A_\alpha\})$.


\begin{definition}[\bf Borel $\sigma$-algebra] 
    Given a topological space $T$, 
    a Borel $\sigma$-algebra, $\cB(T)$ is the smallest $\sigma$-algebra
    containing all open sets of $T$.
\end{definition}
\noindent $(T, \cB(T))$ is called a {\bf Borel Space} when $T$ is a topological space
    and $\cB(T)$ is a Borel $\sigma$-algebra over $T$.

Consider a topological space $(T, \cT)$. For any set $A \in \cT$,
$(A, \cT_A)$ is also a topological space (where $\cT_A = \{E \cap A \vert E \in \cT\}$).
Given a measurable space $(T, \Sigma)$, for any set $A \in \Sigma$,
$(A, \Sigma_A)$ is also a measurable space (where $\Sigma_A = \{E \cap A \vert E \in \Sigma\}$).

{\noindent\bf Topology and $\sigma$-algebra over Reals:}
Consider the smallest topology $R$ over the real space $\mathbb{R}$ which
contains all intervals $(a, \infty) \subseteq \mathbb{R}$ for all $-\infty < a < \infty$.
To avoid confusion, we will refer the topological space over $\mathbb{R}$ with
the symbol $\cR$. We can now use this topology to define a Borel $\sigma$-algebra $\cB(\cR)$
over this topological space. Using the above defined topology and $\sigma$-algebra,
we can define a topological space and a $\sigma$-algebra for any open or closed intervel
in $\mathbb{R}$.

\begin{definition}[\bf Measurable Function over Measurable Spaces]
Given measurable spaces $(T_1, \Sigma_1)$ and $(T_2, \Sigma_2)$, a function 
    $h : T_1 \rightarrow T_2$ is
a measurable function from $(T_1, \Sigma_1)$  to $(T_2, \Sigma_2)$
if $h^{-1}\{B\} \in \Sigma_1$
for all sets $B \in \Sigma_2$, where $h^{-1}\{B\}  = \{x : h(x) \in B\}$.

\end{definition}

\noindent The measurable function $h$ is also known as a {\bf Random Variable}
    from measurable space $(T_1, \Sigma_1)$ to $(T_2, \Sigma_2)$.
    If $h : T_1 \rightarrow T_2$ is a measurable function from 
    measurable space $(T_1, \Sigma_1)$ 
    to a measurable space $(T_2, \Sigma_2)$, and $\pi$ is a 
    probability measure on $(T_1, \Sigma_1)$, then $\pi_h : \Sigma_2 \rightarrow [0, 1]$
    defined as $\pi_h(A) = \pi(h^{-1}(A))$
    is a probability measure on $(T_2, \Sigma_2)$. 

\begin{definition}[\bf Pushforward measure]
    Given a probability space $(T_1, \Sigma_1, \pi)$ and a measurable function
    to a measurable space $(T_2, \Sigma_2)$, then the pushforward measure of $\pi$
    is defined as a probability measure $f_*(\pi) : \Sigma_2 \rightarrow [0, 1]$ given by
    \[
        (f_*(\pi))(B) = \pi(f^{-1}(B)) \text{ for } B \in \Sigma_2
    \]
\end{definition}

\begin{definition}[\bf Measurable]\label{def-measurable}  
    A function $f$ is measurable if $f$ is a measurable function 
    from a measurable space $(T, \Sigma)$ to $(\mathbb{R}, \cB(\cR))$.
\end{definition}
  
\noindent $f$ is measurable if and only if $\forall a \in \mathbb{R}. \{x \in T \vert f(x) > a \} \in \Sigma$.
Intuitively, for every real number $-\infty < a < \infty$, there exists a set $A \in \tau$ containing all
elements which $f$ maps to real numbers greater than $a$.

\begin{definition}[\bf Simple Function] Given a measurable space $(T, \Sigma)$, 
    $s : T \rightarrow [0, \infty)$ is a simple
function if $s(t) = \Sigma_{i=1}^N a_i I_{A_i}(t)$, where 
    $a_i \in [0, \infty)$, $A_i \in \Sigma$, $I_{A_i}(t) = 1$ if
$t \in A_i$ and $0$ otherwise, and the $A_i$ are disjoint.
\end{definition}

\begin{definition}[\bf Lebesgue Integral]\label{def-int} 
    Given a measurable space $(T, \Sigma)$ 
and a measure $m$ over $(T, \Sigma)$, for each 
$A \in \Sigma$ and disjoint $A_i \in \Sigma$, we define 
\[ \int_A I_{A_i}(t) m(dt) = m(A_i \cap A) \]
\noindent Hence 
\[\int_A s(t) m(dt) = \sum\limits_{i=1}^N a_i m(A_i \cap A) \]

    \noindent Given a function $f : T \rightarrow [0, \infty)$ which is measurable, 
    we define 
\[\int_A f(t) m(dt) = \mathsf{sup}\{\int_A s(t) m(dt) \vert s(t) \text{ is simple }, 0 \leq s 
\leq f\}\]
where $s \leq f$ if $\forall t. s(t) \leq f(t)$, as the Lebesgue Integral
    of function $f$ over a set $E$ in measurable space $(T, \tau)$ with measure $m$.

\end{definition}

\noindent Given a function $f : T \rightarrow \mathbb{R}$ which is measurable, we define
    \[
        \int_A f(t) m(dt) = \int_A f_+(t) m(dt) - \int_A f_-(t) m(dt)
    \]
    where $f_+(t) = \max(f(t), 0)$ and $f_-(t) = \max(-f(t), 0)$.
    An integral of a measurable function $f$ is the sum of the integral of
    the positive part and the integral of the negative part.
    From this point on, $\int f(t) m(dt)$ denotes $\int_T f(t)m(dt)$ where
    $T$ is the set over which the measurable space and measure $m$ is defined. 

\Comment{
\begin{definition}[\bf Expected Value]
    Given a probability space $(T, \Sigma, \pi)$ and $X : T \rightarrow \mathbb{R}$
    be a measurable function.
    then we  define {\bf expected valve} $E~X = \int X(t) \pi(dt)$.
\end{definition}

\begin{definition}[\bf Conditional Expectation]
    Given a probability space $(T, \Sigma, \pi)$, a $\sigma$-field $\cF \subseteq \Sigma$
    and a random variable $X : T \rightarrow \mathbb{R}$ with $E~|X| < \infty$.

    We define the {\bf conditional expectation} of $X$ given $\cF$, $E(X \vert \cF)$, 
    to be any function $Y : T \rightarrow \mathbb{R}$, 
\end{definition}
}

\begin{definition}[\bf Markov Transition Kernel]
    Let $(T, \Sigma)$ be a measurable space. A Markov Transition kernel on $(T, \Sigma)$ is 
a map $K : T \times \Sigma \rightarrow [0, 1]$ such that :
\begin{enumerate}
    \item for any fixed $A \in \Sigma$, the function $K(., A)$ is measurable function from 
        $(T, \Sigma)$ to $[0, 1]$.
    \item for any fixed $t \in T$, the function $K(t, .)$ is a probability 
        measure on $(T, \Sigma)$.
\end{enumerate}
\end{definition}

\begin{definition}
    [\bf $\pi$-irreducible] Given a probability space $(T, \tau, \pi)$, a Markov Transition Kernel 
    $K : T\times \tau \rightarrow [0, 1]$ is 
    $\pi$-irreducible if for each $t \in T$ and each $A \in \tau$, such that $\pi(A) > 0$, 
there exists an integer $n = n(t, A) \geq 1$ such that 
\[K^n(t, A) > 0\]
    where $K^n(t, A) = \int_T K^{n-1}(t, dt')K(t', A)$  
    and $K^1(t, A) = K(t, A)$.
\end{definition}

\begin{definition}
    [\bf Stationary Distribution] Given a probability space $(T, \tau, \pi)$, $\pi$ is the stationary 
    distribution of a $\pi$-irreducible Markov Transition Kernel $K : T \times \tau \rightarrow [0, 1]$ if
\[\pi K  = \pi\]
where $(\pi K)(A) = \int K(t, A) \pi(dt)$.
\end{definition}

\begin{definition}
    [\bf Aperiodicity] 
    Given a probability space $(T, \tau, \pi)$, a $\pi$-irreducible 
    Markov Transition Kernel $K : T \times \tau \rightarrow [0,1]$ 
    is periodic if there exists an integer $d \geq 2$ and a sequence 
$\{E_0, E_1, \ldots E_{d-1}\}$ and $N$ of $d$ non-empty disjoint sets in $\tau$ such 
that, for all $i = 0, 1, \ldots d-1$ and for all $t \in E_i$,
\begin{enumerate}
    \item $ (\cup_{i=0}^d E_i) \cup N = T$
    \item $K(t, E_j) = 1 \text{ for } j = i + 1 (\mathsf{mod}~d)$
    \item $\pi(N) = 0$
\end{enumerate}
        Otherwise $K$ is aperiodic. 
\end{definition}

\begin{definition}
    [\bf Asymptotic convergence] 
    Given a probability space $(T, \tau, \pi)$ and sample $t \in T$, 
    a Markov Transition Kernel $K : T \times \tau \rightarrow [0, 1]$ 
    is said to asymptotically converge to $\pi$ if
\[\mathsf{lim}_{n \rightarrow \infty} || K^n(t, .) - \pi||  = 0\]
    where $||.||$ refers to the total variation norm of a measure $\lambda$, 
    defined over measurable space $(T, \tau)$, defined as 
\[||\lambda|| = \mathsf{sup}_{A \in \tau} \lambda(A) - \mathsf{inf}_{A 
\in \tau} \lambda(A)\]
\end{definition}

\begin{theorem}
    \label{thm:Athreyaconv}
    Given a probability space $(T, \tau, \pi)$ and a Markov Transition Kernel 
    $K : T \times \tau \rightarrow [0,1]$.
    If $K$ is $\pi$-irreducible, aperiodic, and $\pi K = \pi$ holds, 
    then for $\pi$-almost all $t$, 
    \[\mathsf{lim}_{n \rightarrow \infty}|| K^n(t, .)  - \pi || = 0  \]
    i.e., $K$ converges to $\pi$.
    $\pi$-almost all $t$ means that there exists a set $D \subseteq T$ such 
    that $\pi(D) = 1$ and for all $t \in D$, $\mathsf{lim}_{n \rightarrow 
    \infty} || K^n(t, .) - \pi || = 0$.
\end{theorem}
\noindent Athreya, Doss, and Sethuraman proved this theorem~\cite{athreya1996convergence}.
    All popular asymptotically converging Markov Chain Algorithms, like variants of 
    the Metropolis Hasting and Gibbs Algorithm, when parameterized over probability space 
    $(T, \tau, \pi)$, are 
    $\pi$-irreducible and aperiodic with stationary distribution $\pi$.

\begin{definition}[\bf Subalgebra] $\cE$ is a $\sigma$-subalgebra of a measurable space $(T, \tau)$ if $\cE$ 
is a $\sigma$-algebra of some set $A \subseteq T$ and $\cE \subseteq \tau$.
\end{definition}

\begin{definition}[\bf Induced Probability space] Given a probability space $(T, \tau, \pi)$
and $A \in \tau$, define $\tau_A = \{B \cap A \vert B \in \tau\}$. 
    Note that since $\tau$ is a $\sigma$-algebra, $\tau_A$ is a sub-algebra over 
    set $A$. $(A, \tau_A)$ is a measurable space. If  $\pi(A) > 0$, then function
    $\pi_A : \tau_A \rightarrow [0, 1]$, defined as $\pi_A(x) = \pi(x) / \pi(A)$, 
    is a probability measure over $(A, \tau_A)$. 
    $(A, \tau_A, \pi_A)$  is the probability space induced by $A \in \tau$. 
\end{definition}

\begin{definition}[\bf Regular Conditional Probability Measure over Product Space]
    Given a product probability space $(T_1 \times T_2, \Sigma_1 \otimes \Sigma_2, \pi)$, a regular 
    conditional probability measure $v : T_1 \times \Sigma_2 \rightarrow [0, 1]$ is a transition kernel
    such that 
    \begin{itemize}
        \item For all $A \in \Sigma_2$, $v(., A)$ is a measurable.
        \item For all $t \in T_1$, $v(t, .)$ is a probability measure over $(T_1, \Sigma_1)$.
        \item For all $B \in \Sigma_1$, $\pi(B \times A) = \int_B v(t, A) \pi(dt \times T_2)$.
    \end{itemize}
\end{definition}

\subsection{Class Functions and Class Kernels}
We next introduce class functions and class kernels, which formalize the 
concept of subproblem based inference metaprograms. Gibbs sampling is a special case
of this framework. To aid the reader, we highlight
how the framework is specialized to Gibbs sampling as we introduce the definitions,
lemmas, and theorems. 

\begin{definition}[\bf Two-way measurable function]
    Given a measurable space $(T_1, \Sigma_1)$
    and a measurable space $(T_2, \Sigma_2)$, 
    a measurable function $f$ from $(T_1, \Sigma_1)$
    to $(T_2, \Sigma_2)$ is a {\bf two-way
    measurable function}, if for 
    all sets $A \in \Sigma_1$ there exists a set $B \in \Sigma_2$, 
    such that 
        \[B = \{f(t) \vert t \in A\}\]

    i.e., the map of any set $A \in \Sigma_1$ is a set in $\Sigma_2$.

\end{definition}

Given a two-way measurable function $f$, we can extend 
$f$ to a function $g : \Sigma_1 \rightarrow \Sigma_2$  
between sets in $\Sigma_1$ to sets in $\Sigma_2$, 
    where $g(A) = \{f(t) \vert t \in A\}$. Since $f$ is a measurable function, 
    the function $f^{-1} : \Sigma_2 \rightarrow \Sigma_1$ is also 
    defined which maps sets in $\Sigma_2$ to sets in $\Sigma_1$.

{\bf Note:} From this point on, 
    given a two-way measurable function $f$, we will simply use it to 
    represent function $g$ defined above, mapping sets from $\Sigma_1$ 
    to $\Sigma_2$. We also use $f^{-1}$ as a reverse map from $\Sigma_2$ 
    to $\Sigma_1$ defined above. Note that $f^{-1}$ may or may not be the
    inverse of the function $f$.

\begin{example} (Gibbs)
    Given measurable spaces $(X, \cX)$, $(Y, \cY)$ and  
    product space $(X \times Y, \cX \otimes \cY)$, the projection 
    functions $\mathsf{proj}_x$ and $\mathsf{proj}_y$ are two-way 
    measurable functions from $(X \times Y, \cX \otimes \cY)$
    to $(X, \cX)$ and $(Y, \cY)$, respectively.
\end{example}

\begin{definition}[\bf Generalized Product Space] 
    Given countable sets of 
    disjoint measurable spaces $\{(X_i, \cX_i) \vert i \in I\}$ and 
    $\{(Y_i, \cY_i) \vert i \in I\}$, we 
    define the {\bf generalized product space} $(C, \cC)$ as:
    \[
        (C, \cC) = \big(\bigcup\limits_{i \in I} X_i \times Y_i, 
        \sigma(\bigcup\limits_{i \in I}\cX_i \otimes \cY_i) \big)
    \]
\end{definition}

    \Comment{
    construct a countable set of product spaces
    \[
        \{(X_i \times Y_i, \cX_i \otimes \cY_i) \vert i \in I\}
    \]
    We assume, for any probability measure $\pi_i$ on any of the product spaces
    $(X_i \times Y_i, \cX_i \otimes \cY_i)$, we can construct a regular
    conditional probability measure $v : X_i \times \cY_i \rightarrow [0, 1]$.}

\noindent Given a probability measure $\pi : \cC \rightarrow [0,1]$ on a generalized product space $(C, \cC)$, 
    we can define a conditional distribution $\pi_i$ over each product space 
    $(X_i \times Y_i, \cX_i \otimes \cY_i)$ such that 
    \[\pi(A) = \sum\limits_{i \in I} \pi_i(A \cap X_i \times Y_i) \pi(X_i \times Y_i)\]
    
\noindent In the following we require, for each product probability space 
    $(X_i \times Y_i, \cX_i \otimes \cY_i, \pi_i)$, that we can construct a regular
    conditional probability measure $v_{\pi_i} : X_i \times \cY_i \rightarrow [0, 1]$,
    such that for each $U \times V \in \cX_i \otimes \cY_i$
    \[
        \pi_i(U \times V) = \int_U v_{\pi_i}(x, V) \pi_i(dx \times Y_i)
    \]

\noindent When constructing a generalized product space, we always 
    prove the above assumption, i.e., a regular conditional probability measure exists.

\begin{example} (Gibbs)
    Given measurable spaces $(X, \cX)$, $(Y, \cY)$, consider
    the product probability spaces $(X \times Y, \cX \otimes \cY, \pi)$
    and $(Y \times X, \cY \otimes \cX, \pi')$, where 
    \[
        \pi(U \times V) = \pi'(V \times U)
    \]

    A Gibbs sampler, sampling over $(X \times Y, \cX \otimes \cY, \pi)$, 
    requires regular conditional probabilities 
    $v_{\pi} : X \times \cY \rightarrow [0, 1]$ and 
    $v_{\pi'} : Y \times \cX \rightarrow [0, 1]$ for independence sampling. Therefore
    both $(X \times Y, \cX \otimes \cY)$ and $(Y \times X, \cX \otimes \cY)$
    are {\bf generalized product spaces}.
\end{example}

\begin{definition}[\bf Class Functions]
    Given a measurable space $(T, \Sigma)$ and a
    generalized product space $(C, \cC)$, 
    a {\bf class function} is a one-to-one two-way measurable function 
    $f$ from $(T, \Sigma)$ to $(C, \cC)$.

\end{definition}
\noindent     Given a class function $f$ and the target product space 
    $(C, \cC)$, the {\bf projection functions} 
    \[f_x = \mathsf{proj}_x \circ f : T \rightarrow 
    \bigcup\limits_{i \in I} X_i \text{ and } f_y = \mathsf{proj}_y \circ f : T
    \rightarrow \bigcup\limits_{i \in I} Y_i\]
    are also two-way measurable functions from 
    space $(T, \Sigma)$ to projection spaces \linebreak  $\big(\bigcup\limits_{i \in I} X_i, 
    \sigma(\bigcup\limits_{i \in I}\cX_i)\big)$ and  $\big(\bigcup\limits_{i \in I} Y_i, 
    \sigma(\bigcup\limits_{i \in I}\cY_i)\big)$ respectively. 

Section~\ref{sec:metaprog} uses class functions to model subproblem selection strategies --- each
class function produces a tuple $\langle x, y \rangle$ that 
identifies the parts of the trace that are outside ($x$) and inside ($y$)
the selected subproblem. Because the class function may depend on the input
trace $t$, our framework supports subproblem selection strategies that
depend on the input trace (and specifically on the values of stochastic choices in the trace). 

\begin{example} (Gibbs)
    Given measurable space $(X, \cX)$, $(Y, \cY)$ and 
    generalized product spaces $(X \times Y, \cX \otimes \cY)$ and 
    $(Y \times X, \cY \otimes \cX)$, the identity function 
    $id : X \times Y \rightarrow X \times Y$ and the reverse function 
    $re : X \times Y \rightarrow Y \times X$ (i.e. $re(\tup{x, y}) = \tup{y, x}$)
    are class functions from $(X \times Y, \cX \otimes \cY)$ to 
    $(X \times Y, \cX \otimes \cY)$ and $(Y \times X, \cY \otimes \cX)$.
    Note that $id_x = \mathsf{proj}_x$ and $re_x = \mathsf{proj}_y$. 
\end{example}

    Given a probability space $(T, \Sigma, \pi)$,
    a class function $f$ to a
    generalized product space $(C^f, \cC^f) = \big(\bigcup\limits_{i \in I}X^f_i \times Y^f_i,
    \sigma(\bigcup\limits_{i \in I}\cX^f_i \otimes \cY^f_i)\big)$, 
    $f$ maps each point $t \in T$ to some point $\tup{x, y}$ in $X^f_i \times Y^f_i$ for some $i$.

    For each $i \in I$, we require a function $K_i : X^f_i \rightarrow (Y^f_i \times \cY^f_i) \rightarrow [0, 1]$
    such that 
    \begin{itemize}
        \item $K_i(x) : Y^f_i \times \cY^f_i \rightarrow [0, 1]$ is a 
            $v_{f_*(\pi)_i}(x, .)$-irreducible, aperiodic Markov Transition 
            Kernel with $v_{f_*(\pi)_i}(x, .)$ as it's stationary distribution.
        \item $K_i(.)(y, A)$ is measurable for all $y \in Y^f_i$ and all $A \in \cY^f_i$.
    \end{itemize}
When we apply the framework to prove the convergence of inference metaprograms (Section~\ref{sec:metaprog}), 
each $K_i$ represents an function that, when given a parameter $x$, returns a converging Markov kernel 
based on $x$, where $x$ corresponds loosely to the parts of the program trace that lie outside the 
scope of the selected subproblem. These parameterized $K_i$ support sophisticated subproblem inference
strategies that depend on how the subproblem decomposes the program trace. For example, the metaprogram
may apply one inference strategy to subproblems with discrete random choices and another to subproblems
with continuous choices. 

\begin{definition}[\bf Class Kernels]
    Given $(T, \Sigma, \pi), f,$ and $K_i$ as above, we define a  
    {\bf class kernel} $K_f : T \times \Sigma \rightarrow [0, 1]$ as a Markov Transition Kernel 
    defined as 
    \[
        K_f(t, A) = K_i(x)(y, V)I(x, U)
    \]
    where $f(t) = \tup{x, y} \in X_i \times Y_i$ and $U \times V = f(A) \cap (X_i \times Y_i)$.
\end{definition}

\begin{example} (Gibbs)
    Consider a probability space $(X \times Y, \cX \otimes \cY, \pi)$ and 
    class functions $id$ and $re$, Markov transition Kernels 
    ~$K_{id}(\tup{x, y}, U \times V) = v_{\pi}(x, V)I(x, U)$ and
    ~$K_{re}(\tup{x, y}, V \times U) = v_{\pi'}(y, U)I(y, V)$ 
    representing independence samplers
    are 
    class kernels.
\end{example}

\Comment{
\begin{example}[\bf Conditional Independent Sampling Kernels] 
    Consider a probability space $(T, \Sigma, \pi)$ and
    a class function $f$ to a generalized product space 
     $(C, \cC) = \big(\bigcup\limits_{i \in I}X_i \times Y_i,
    \sigma(\bigcup\limits_{i \in I}\cX_i \otimes \cY_i)\big)$. 
\end{example}
}

\subsection{Properties of Class Kernels}
Consider a probability space $(T, \Sigma, \pi)$ and 
a class function $f$ to a generalized 
product space  $(C^f, \cC^f) = \big(\bigcup\limits_{i \in I}X^f_i \times Y^f_i,
    \sigma(\bigcup\limits_{i \in I}\cX^f_i \otimes \cY^f_i)\big)$.
Below, we prove properties of a Class Kernel $K_f$ within this context.

\begin{lemma}
    \label{lemma:Kfexpand}
    For all $t \in T$ and $A \in \Sigma$, 
    \[K^n_f(t, A) = K^n_i(x)(y, V)I(x, U)\]
    where $f(t) = \tup{x, y} \in X_i \times Y_i$ and $U \times V = f(A) \cap X_i \times Y_i$. 
\end{lemma}
We present the proof of this lemma in Appendix~\ref{append:conv} (Lemma~\ref{appendlemma:Kfexpand}).

\begin{lemma}
    \[
        \pi(A) = \int_{t \in T} K_f(t, A) \pi(dt)
    \]
\end{lemma}
We present the proof of this lemma in Appendix~\ref{append:conv} (Lemma~\ref{lem:convappend2}).

\begin{lemma}
 $K_f$ is aperiodic if for at least one $x \in X_i$ for some $i \in I$, 
    $K_i(x) : Y_i \times \cY_i \rightarrow [0, 1]$ 
    is aperiodic.
\end{lemma}
We present the proof of this lemma in Appendix~\ref{append:conv} (Lemma~\ref{lem:classaper}).
\subsection{Connecting the probability space}

We use a finite set of class functions $\cF = \{f_1, f_2 \ldots f_n\}$
to model the inference steps of the hybrid inference metaprogram (Section~\ref{sec:metaprog}).
A critical concept here is that, together, the $\cF$ must connect the underlying probability 
space --- conceptually, starting at any positive probability set contained within the
space, it must be possible to reach any other positive probability set by
following a path of positive probability sets as defined by $\cF$. If $\cF$ does not
connect the space, it is possible for the inference metaprogram to become stuck
within an isolated subspace, with some positive probability sets unreachable
even in the limit.

\begin{definition}[\bf Connecting the space $(T, \Sigma, \pi)$]
    \label{def:connect}
    Given a probability space $(T, \Sigma, \pi)$, 
    a finite set of class functions $\cF = \{f_1, f_2 \ldots f_n\}$\Comment{(to some generalized product spaces 
    $(C^{f_1}, \cC^{f_1}),  (C^{f_2}, \cC^{f_2}) \ldots  (C^{f_n}, \cC^{f_n})$ respectively)}
    connect the probability space $(T, \Sigma, \pi)$, if
    for all sets $A \in \Sigma$ and any two functions $f, g \in \cF$
  
    \begin{equation*}
       \begin{aligned}
       \pi(f^{-1}_x(f_x(A)) \cap g^{-1}_x(g_x(A)^c)) =  \pi(f^{-1}_x(f_x(A)^c) \cap g^{-1}_x(g_x(A)) ) = 0
       \\
    \implies \pi(f^{-1}_x(f_x(A)^c)) = 0 \text{ or } \pi(g^{-1}_x(g_x(A)^c)) = 0 
       \end{aligned}
       \end{equation*}
\end{definition}

\Comment{
As shown below that the above condition of {\bf Connecting the space $(T, \Sigma, \pi)$ }
reduces to the condition under which a two component gibbs sampler converges to the 
correct distribution (Theorem~\ref{thm:gibbsconv}).
}

For standard two-component Gibbs sampling, there are only two class functions, 
specifically $id$ and $re$:

\begin{example}[\bf Connected product space] (Gibbs)
    Given a probability space $(X \times Y, \cX \otimes \cY, \pi)$
    and class functions $id$ and $re$, then 
    $id$ and $re$ connect the space $(X \times Y, \cX \otimes \cY, \pi)$, 
    if for all sets $U \times V \in \cX \otimes \cY$, 
    \[
        \pi(U \times V^c) = \pi(U \times V^c) = 0
        \implies \pi(U^c \times Y) = 0  \text{ or } \pi(X \times V^c) = 0
    \]
    as 
    \[\pi(id_x^{-1}(id_x(U \times V)) \cap re^{-1}_x(re^{-1}_x(U \times V)^c)
        = \pi(U \times Y \cap X \times V^c) = \pi(U \times V^c)
    \]
    and 
   \[\pi(id_x^{-1}(id_x(U \times V)^c) \cap re^{-1}_x(re^{-1}_x(U \times V))
        = \pi(U^c \times Y \cap X \times V) = \pi(U^c \times V)
    \]
\end{example}

\subsection{Stochastic Alternating Class Kernels}
\label{sec:staltclk}
Consider a probability space $(T, \Sigma, \pi)$ where
$\Sigma$ is countably generated. 
We construct a new Markov Chain Transition Kernel 
using a finite set of class functions $\cF = \{f_1, f_2 \ldots f_m\}$ 
and respective class Kernels $K_{f_1}, K_{f_2} \ldots K_{f_m}$.

\begin{definition}[Stochastic Alternating Markov Chain Transition Kernel]
Given $m$ positive real numbers $p_k \in (0, 1)$ which sum to $1$ (i.e. $\sum\limits_{k=1}^m p_k = 1$), 
we define a stochastic alternating Markov Chain Transition Kernel 
$K : T \times \Sigma \rightarrow [0, 1]$ as
\[
    K(t, A) = \sum\limits_{k=1}^m p_k K_{f_k}(t, A)
\]
\end{definition}

This transition kernel corresponds to randomly picking a class kernel $K_{f_k}$
with probability $p_k$ and using it to transition into the next Markov Chain State.

\begin{example} (Gibbs)
The Markov Transition Kernel for a 2 component Gibbs sampler 
    over product probability space $(X \times Y, \cX \otimes \cY, \pi)$ is
    \[
        K(\tup{x, y}, U \times V) =  
        p v_{\pi}(x, V)I(x, U) + (1 - p) v_{\pi'}(y, U)I(y, V)
    \]
\end{example}

We now consider the question of convergence of the Stochastic Alternating Transition Kernel.

Let $R_{t}^k = \{
    A \in \Sigma \wedge \exists 1 \leq n \leq k.K^n(t, A) > 0 \}
$ be the set of all sets in $\Sigma$ which 
are reachable by kernel $K$ in $k$ steps
starting from element $t \in T$.

Consider the limiting case $R_t^\infty$. Let $B_t^\infty = \{t \in T \vert \forall~A \in \Sigma.A \notin R_t^\infty \implies t \notin A\}$
which is the set of elements $t' \in T$ such that any set $A$ in $\Sigma$ which contains $t'$ is reachable by kernel $K$.

\begin{lemma}
    For any $f \in \cF$, any element $t' \in B_t^\infty$ such that 
    $f(t') = \tup{x, y} \in X_i^f \times Y_i^f$, and any 
    set $U \times V \in \sigma(\cX_i^f \otimes \cY_i^f)$ such that $A = f^{-1}(U \times V)$, 
    the following condition holds:
    \[
        v_{f_*(\pi)_i}(x, V) > 0 \wedge x \in U \implies A \in R_t^\infty
    \]
\end{lemma}
We present the proof of this lemma in Appendix~\ref{append:conv} (Lemma~\ref{lem:appendconv3}).

\begin{lemma}
    \label{lemma:fxA}
    For any positive probability set $A$ and any function $f \in \cF$, 
    if $A \subseteq f^{-1}_x(f_x(B_t^\infty))$ then $A \in R_t^\infty$.
\end{lemma}
We present the proof of this lemma in Appendix~\ref{append:conv} (Lemma~\ref{appendlemma:fxA}).

\begin{lemma}
    If $\cF$ connects the space $(T, \Sigma, \pi)$ then there does not exist a positive probability set $A \in \Sigma$, 
    such that $A \subseteq \bigcap_{f \in \cF} f_x^{-1}((f_x(B^\infty_t))^c)$.
\end{lemma}
We present the proof of this lemma in Appendix~\ref{append:conv} (Lemma~\ref{lem:appendconv4}).

\begin{theorem} 
    \label{thm:convirr}
    If $\cF$ connects the space $(T, \Sigma, \pi)$ then the Markov Transition Kernel 
    $K$ is $\pi$-irreducible.    
\end{theorem}
We present the proof of this theorem in Appendix~\ref{append:conv} (Theorem~\ref{appendthm:convirr}).

\begin{theorem}
    $\pi$ is the stationary distribution of Markov Kernel $K$, i.e.
    \[\int K(t, A)\pi(dt) = \pi(A) \text{ for all } A \in \Sigma\]
    \label{thm:convst}
\end{theorem}
We present the proof of this theorem in Appendix~\ref{append:conv} (Theorem~\ref{appendthm:convst}).

\begin{theorem}
    \label{thm:convap}
    The Markov Transition Kernel $K$ is aperiodic if at least one of the 
    class kernels $K_{f_j}$ is aperiodic.
\end{theorem}
We present the proof of this theorem in Appendix~\ref{append:conv} (Theorem~\ref{appendthm:convap}).

\begin{theorem}
    Markov Transition Kernel $K$ converges to probability distribution $\pi$.
\end{theorem}
\begin{proof}
    Using Theorems \ref{thm:Athreyaconv}, \ref{thm:convirr}, \ref{thm:convst}, and \ref{thm:convap}.
\end{proof}

\newcommand{\sembr} [1] {\mathsf{pdf}\llbracket #1 \rrbracket}
\newcommand{\prob} [1] {\mathsf{\mu}\llbracket #1 \rrbracket}
\newcommand{\salg} [1] {f\llbracket #1 \rrbracket}

\section{Inference Metaprograms}
\label{sec:metaprog}
We next formalize the concept of the probability of a trace, introduce 
inference metaprogramming, and use the results in Section \ref{sec:stoch} to prove
the convergence of inference metaprograms.

\subsection{Preliminaries}
Here we relate the concepts in Section \ref{sec:stoch} to concepts used in probabilistic programming.

\noindent{\bf Probability of a Trace:}
Because two traces are equivalent if they differ (if at all) only in the choice of unique identifiers, 
within this section, for clarity, we drop the $id$'s associated with augmented expressions
and stochastic choices in traces and augmented expressions wherever they are not required.

Assuming a countable set of variable names allowed in our probabilistic lambda 
calculus language, a countable number of expressions and a countable number of 
programs can be described in our probabilistic lambda calculus language. 

Within our probabilistic lambda calculus language, we assume that all
stochastic distributions $\mathsf{Dist}_k : V \times \cP(E) \rightarrow [0, 1]$
are functions from a tuple of value and set of lambda expressions in our language to a real number between $0$ and $1$, 
such that for any $v \in V$, $\mathsf{Dist}_k(v, .)$ is a probability measure over probability space $(E, \cP(E))$.
We assume that for each distribution $\mathsf{Dist}_k$, we are given a probability density function 
$\mathsf{pdf}_{\mathsf{Dist}_k} : V \times E \rightarrow [0, 1]$, such that
\[
    \mathsf{Dist}_k(v, A) = \sum\limits_{e \in A} \mathsf{pdf}_{\mathsf{Dist}_k}(v, e)
\]
Let $T_p = \mathsf{Traces}(p)$ be the set of valid traces of a program $p$. $T_p$ contains a countable number of traces.
We define a $\sigma$-algebra $\Sigma_p$ over set $T_p$ such that, for all $t \in T_p$, $\{t\} \in \Sigma_p$.
Given a trace $t$, $\sembr{t}$ is the unnormalized probability density of the trace $t$
(Figure~\ref{fig:probmeas}). The normalized probability distribution $\mu_p : \Sigma_p \rightarrow [0, 1]$
for a given program $p$ is defined as 
\[
    \mu_p(A) = \frac{\sum\limits_{t \in A} \sembr{t}}{\sum\limits_{t \in T_p} \sembr{t}}
\]

\begin{figure}
\[
\begin{array}{rcl}
    \sembr{x:x} &=& 1\\
    \sembr{x:v} &=& 1\\
    \sembr{\lambda.x~e: v} &=& 1\\
    \sembr{(ae_1~ae_2)\perp: v} &=& \sembr{ae_1}*\sembr{ae_2}\\
    \sembr{(ae_1~ae_2)x=ae: v} &=& \sembr{ae_1}*\sembr{ae_2}*\sembr{ae}\\
    \sembr{\mathsf{Dist_i}(ae) = ae_e:v_e} &=& \sembr{ae}*\sembr{ae_e}*\mathsf{pdf}_{\mathsf{Dist_i}}(\cV(ae), e)\\
    &\mathsf{where}& p = 
     ae_e \Rightarrow_r e\\
    \\
    \sembr{\emptyset} &=& 1\\
    \sembr{\assumed{x = ae};t} &=& \sembr{ae}*\sembr{t}\\
    \sembr{\observed{ae}{e};t} &=& \mathsf{pdf}_{\mathsf{Dist}}(\cV(ae), e)*\sembr{ae}*\sembr{t}\\
\end{array}
\]
\caption{Probabilistic measure over traces} 
    \label{fig:probmeas}
\end{figure}

\noindent{\bf Reversible Subproblem Selection Strategy:}
Let $p$ be a probabilistic program, $t$ and $t'$
be valid traces from program $p$ (i.e. $t, t' \in \mathsf{Traces}(p)$), and 
$\mathsf{SS}$ be a subproblem selection strategy that 
returns a valid subproblem over $t$. 
\begin{definition}[\bf Reversible subproblem selection strategy]
    A subproblem selection strategy $\mathsf{SS}$ is reversible if 
    given any natural number  $n$ and $n$ valid traces $t_1, t_2, \ldots t_n \in \mathsf{Traces}(p)$,
    \[
        \big(\forall~i \in \{1, 2, \ldots n - 1\}.\mathsf{SS}(t_i) \vdash t_i \equiv t_{i+1}\big)
        \implies \mathsf{SS}(t_n) \vdash t_n \equiv t_1
    \]
\end{definition}
    \noindent{i.e.,}
    given any $n$ traces $t_1, t_2 \ldots t_n$, if we can transform the trace $t_i$ to get trace $t_{i+1}$ by
    only modifying parts of the trace $t_i$ selected by the subproblem selection strategy $\mathsf{SS}$, then 
    it is possible to transform the trace $t_n$ to get trace $t_1$ by modifying parts of the trace
    $t_n$ selected by the subproblem selection strategy $\mathsf{SS}$.

In essence, reversible subproblem selection strategies always allow a subproblem based 
inference algorithm to reverse a countable number of changes it has made to a trace (given that the 
trace was not modified under a different subproblem selection strategy).

We use the shorthand $\mathsf{SS} \vdash t \equiv t'$ to denote
$\mathsf{SS}(t) \vdash t \equiv t'$.

\begin{theorem}
    A reversible subproblem selection strategy $\mathsf{SS}$ divides the trace space of program $p$
    into equivalence classes.
\end{theorem}
We present the proof of this theorem in Appendix~\ref{append:meta} (Theorem~\ref{thm:appendequi}).

A reversible subproblem selection strategy $\mathsf{SS}$ divides the trace space $T_p$
into a countable number of equivalence classes where each 
equivalence class contains traces which can be modified into any other trace 
in that class under the subproblem selection strategy.
A trace from one equivalence class cannot be modified by any subproblem based
inference algorithm to a trace from a different equivalence class under the 
given subproblem selection strategy.

Given a reversible subproblem selection strategy $\mathsf{SS}$, let 
$\cC_{SS} = \{c_1, \ldots c_n, \ldots \}$ be the countable set of 
equivalence classes created by $\mathsf{SS}$ over the trace space $T_p$ and 
$\{ T_{c_1}, \ldots T_{c_n}, \ldots \}$ be the equivalence partitions 
created by $\mathsf{SS}$ over $T_p$. Note that for all $c \in \cC_{SS}$ and $t, t' \in T_c$,  
\[
    \mathsf{SS} \vdash t \equiv t' ~\wedge~ \mathsf{SS} \vdash t' \equiv t
\]
and for all $c_i, c_j \in \cC_{SS}, t \in T_{c_i}$ and $t' \in T_{c_j}$, where $c_i \neq c_j$
\[
    (\mathsf{SS} \vdash t \equiv t' ~\vee~ \mathsf{SS} \vdash t' \equiv t) = \mathsf{false}
\]

In practice, subproblems are often specified by associating labels
with stochastic choices, then specifying the labels whose stochastic
choices should be included in the subproblem~\cite{mansinghka2018probabilistic}. 
A standard strategy is to have a fixed set of labels, with the labels
partitioning the choices into classes.  Any strategy that always 
specifies the subproblem via a fixed subset of labels is reversible. 
Because of this property, all of the subproblem selection strategies
presented in~\cite{mansinghka2018probabilistic} are reversible. 

Any subproblem selection strategy that always selects a fixed set of variables
is also reversible. This property ensures that the subproblem
selection strategy in Block Gibbs sampling, for example, is reversible. 
Hence if two traces differ only in the choice of their id's, all reversible
subproblem selection strategies will assign them to the same equivalence class.

\noindent{\bf Class functions given a subproblem selection strategy:}
Consider a reversible subproblem selection strategy $\mathsf{SS}$ which 
creates equivalence classes $\cC_{SS} = \{c_1, \ldots c_n \ldots \}$. 
and equivalence partitions $\{T_{c_1}, \ldots T_{c_n}, \ldots \}$.
We create a generalized product space and class functions using the given subproblem 
selection strategy.
Consider the countable set of disjoint measurable spaces
$\big\{ (C_1, \cC_1), \ldots  (C_n, \cC_n), \ldots  \big\}$
where $C_k = \{c_k\}$ and $\cC_k = \{ \emptyset, C_k \}$. Also
consider disjoint measurable spaces $\big\{ (T_{c_1}, \Sigma_{c_1}), \ldots  (T_{c_n}, \Sigma_{c_n}), \ldots \big\}$.
where 
$\Sigma_{c_k} = \{A \cap T_{c_k} \vert A \in \Sigma_p \}$. We then construct the generalized product space 
\[
    (C, \cC) = (\bigcup\limits_{c_i \in \cC_{SS}} C_i \times T_{c_i}, \sigma(\bigcup\limits_{c_k \in \cC_{SS}} \cC_k \otimes \Sigma_{c_k} ) )
\]
Given a probability measure $\pi$ on $(C, \cC)$, we compute the conditional 
distribution $\pi_i$ on $(C_i \times T_{c_i}, \cC_{SS} \otimes \Sigma_{c_k})$, when $\pi(C_i \times T_{c_i}) > 0$
where
\[
    \pi_i(A) = \frac{\pi(A)}{\pi(C_i \times T_{c_i})}
\]
We then define the regular conditional probability measure 
$v_i : C_i \times \Sigma_{c_i} \rightarrow [0, 1]$, where 
    $v_i(c_i, A) = \pi_i(C_i \times A)$.
We create the class function $f_{SS} : T_p \rightarrow C$, 
where 
$ f_{SS}(t) = \tup{c, t}$ where $c$ is the equivalence class of trace $t$.
Since $f_{SS}$ is a one-to-one function, it is straightforward to 
prove that $f_{SS}$ is a two-way measurable function.

\noindent{\bf Probability of the subtraces:}
For all traces $t, t' \in T_p$ such that $\mathsf{SS} \vdash t \equiv t'$,
subtraces $t_s = \mathsf{ExtractTraces}(t, \mathsf{SS}(t))$ and $t'_s = \mathsf{ExtractTraces}(t', \mathsf{SS}(t'))$
are traces from the same program, i.e., $t_s, t'_s \in \mathsf{Traces}(p_s)$, where 
$p_s$ is the subprogram (Soundness). 
Similarly, for all traces $t \in T_p$ and subtraces $t_s = \mathsf{ExtractTrace}(t, \mathsf{SS}(t))$,
for all subtraces $t'_s \in \mathsf{Traces}(p_s)$ (where $p_s = \mathsf{Program}(t_s)$) and 
$t' = \mathsf{StitchTrace}(t, t'_s, \mathsf{SS}(t))$ then 
$\mathsf{SS} \vdash t \equiv t'$ (Completeness). 

Hence given a equivalence class $c_i$ and partitioned trace space $T_{c_i}$ created 
by subprogram selection strategy $\mathsf{SS}$, there exists a subprogram $p_s$ such that 
for all traces $t \in T_{c_i}$, subtraces $t_s = \mathsf{ExtractTrace}(t, \mathsf{SS}(t))$
are valid traces of $p_s$, i.e., $t_s \in T_{p_s}$. We can therefore associate
traces from an equivalence class to valid subtraces of a subprogram.

\Comment{
\begin{theorem}
    Given a trace $t$ and a valid subproblem $\cS$ on $t$, 
    then for subtrace $t_s = \mathsf{ExtractTrace}(t, \cS)$, 
    \[
        \sembr{t} = \sembr{t_s}
    \]
    i.e., the unnormalized densities of $t$ and $t_s$ are equal.
\end{theorem}
View theorem~\ref{thm:appdensity} in Appendix for proof.
}

\begin{theorem}
    \label{thm:appdensity}
    Given a trace $t$ and a valid subproblem $\cS$ on $t$, 
    then for subtrace $t_s = \mathsf{ExtractTrace}(t, \cS)$, 
    \[
        \sembr{t} = \sembr{t_s}
    \]
    i.e., the unnormalized densities of $t$ and $t_s$ are equal.
\end{theorem}
We present the proof of this theorem in the Appendix~\ref{append:meta} (Theorem~\ref{appendthm:appdensity}).

\noindent Consider an equivalence class $c_i$ and partitioned trace space $T_{c_i}$. Let $p_s$ 
be the subprogram such that all subtraces of traces in $T_{c_i}$ are valid traces of $p_s$.
$(T_{p_s}, \Sigma_{p_s})$ is the measurable space over traces of subprogram $p_s$.
The normalized probability distribution $\mu_{p_s} : \Sigma_{p_s} \rightarrow [0, 1]$ for the 
subprogram $p_s$ is 
\[
    \mu_{p_s}(A) = \frac{\sum\limits_{t_s \in A} \sembr{t_s}}{\sum\limits_{t_s \in T_{p_s} } \sembr{t_s} } 
    = \frac{\sum\limits_{t \in A'} \sembr{t} }{\sum\limits_{t \in T_{c_i}} \sembr{t}} = (\mu_p)_i(A') = \frac{\mu_p(A')}{\mu_p(T_{c_i})} = v_i(c_i, A')
\]
where $A' = \big\{t' \big\vert t_s \in A, t' = \mathsf{StitchTrace}(t, t_s, \mathsf{SS}(t))  \big\}$ for any trace $t \in T_{c_i}$.

Hence, sampling/inference over subprogram $p_s$ is equivalent to sampling/inference 
over the original program $p$ with the constraint that all traces belong to the equivalence class
$c_i$.

\begin{theorem}
    Consider equivalence class $c_i$ and partitioned trace space $T_{c_i}$. Let $p_s$
    be the subprogram such that all subtraces of traces in $T_{c_i}$ are valid traces of $p_s$.
    Then given a  Markov Kernel $K : T_{p_s} \times \Sigma_{p_s} \rightarrow [0, 1]$ which is 
    $\mu_{p_s}$-irreducible, aperiodic, and with stationary distribution $\mu_{p_s}$, the 
    Markov kernel $K(c_i) : T_{c_i} \times \Sigma_{c_i} \rightarrow [0, 1]$ defined as 
    \[
        K(c_i)(t, A) = K(t_s, A')
    \]
    where $t_s = \mathsf{ExtractTrace}(t, \mathsf{SS}(t))$ and $A' = \{t_s \vert t \in A, t_s = \mathsf{ExtractTrace}(t, \mathsf{SS}(t)) \}$,
    is $v_i(c_i, .)$-irreducible, aperiodic, and with stationary distribution $v_i(c_i, .)$.
\end{theorem}
\begin{proof}
    $\mathsf{ExtractTrace}$ is one-to-one function from $T_{c_i}$ to $T_{p_s}$
    and the push forward measure of  $v_i(c_i, .)$ is $\mu_{p_s}$.
\end{proof}

\begin{definition}[\bf Generalized Markov Kernel]
    A Generalized Markov Kernel $K$ is a parameterized Markov kernel which, when parameterized with a 
     probabilistic program $p$, defines the probability space $(T_p, \Sigma_p, \mu_p)$ (as defined above),
and returns a Markov Kernel $K(p) : T_p \times \Sigma_p \rightarrow [0, 1]$, which is
        $\mu_p$-irreducible, aperiodic, and has the stationary distribution $\mu_p$.
\end{definition}
A generalized Markov Kernel formalizes the concept of Markov chain inference algorithms. 
The inference algorithms used within the probabilistic programming framework are 
generally coded to work with any input probabilistic program $p$ and still 
provide convergence guarantees. For example,
Venture \cite{mansinghka2018probabilistic} allows the programmer to use a variety of inference algorithms
which in general work on all probabilistic programs which can be written in that language.

\begin{definition}[\bf Generalized Class Kernels]
    Given a generalized Markov Kernel $K$ and a subproblem selection strategy 
    $\mathsf{SS}$, a generalized class Kernel $K_{f_{\mathsf{SS}}}$ is parameterized 
    with a probabilistic program $p$ (which defines space $(T_p, \Sigma_p, \mu_p)$), where 
    \[
        K_{f_{\mathsf{SS}}}(p)(t, A) = K'(c_i)(t, A) = K(p_s)(t_s, A')
    \]
    and 
    $t_s = \mathsf{ExtractTrace}(t, \mathsf{SS}(t))$, $p_s = \mathsf{Program}(t_s)$, $f_\mathsf{SS}(t) = \tup{c_i, t}$ 
    and $A' = \{t_s \vert t \in A, 
    t_s = \mathsf{ExtractTrace}(t, \mathsf{SS}(t)) \}$.
    \label{def:genclass}
\end{definition}

\subsection{Inference Metaprogramming}
\label{sec:convmeta}
Using the concept of independent subproblem inference (Section~\ref{sec:indepsubinfer}) and generalized Markov Kernels, we define 
an Inference Metaprogramming Language (Figure~\ref{fig:inferlang}). An inference metaprogram 
is either one of the black box Generalized Markov Kernel inference algorithms $b_i : T_p \rightarrow T_p$ in our framework, which takes
a trace from an arbitrary program $p$ as an input and returns another trace from the same program,
or a finite set $S = \{p_1~ic_1, p_2~ic_2, \ldots, p_k~ic_k\}$ of  inference 
statements with an attached probability value $p_i \in (0, 1)$, such that $\sum\limits_{i=0}^k p_k = 1$. 
These probability values are used to randomly select a subproblem inference statement to execute.
Each $\mathsf{infer}$ statement is parameterized with a subproblem selection strategy $\mathsf{SS}$, 
which returns a valid subproblem over input trace $t$ and an inference metaprogram that is executed over the subtrace.
Figure~\ref{fig:inferlangexec} presents the execution semantics of our inference metaprogramming language. 
In comparison with entangled subproblem inference, one benefit of the approach is 
that it is straightforward to apply independent subproblem inference recursively.

\begin{figure}
    \[
        \begin{array}{rcl}
         ic \in IC &:=& \mathsf{infer}(SS, ip)\\
            ip \in IP &:=& b_i \vert \{p_1~ic_1, p_2~ic_2 \ldots p_k~ic_k  \} ~~~ \mathsf{where}~~~\sum\limits_{i=0}^kp_k = 1\\
        \end{array}
    \]
    \caption{Inference Metaprogramming language}
    \label{fig:inferlang}
\end{figure}

\begin{figure}
    \[
        \begin{array}{cc}
        \infra{t' = b(t) }
        {b, t \Rightarrow_i t'}
            &
            \infra{
                n~\sim~\mathsf{multinomial}(p_1, p_2 \ldots p_k)
                ~~~~~ ic_n = \mathsf{infer}(SS_n, ip_n)
                ~~~~~ SS_n(t) = \cS\\
                t_s = \mathsf{ExtractTrace}(t, \cS)
                ~~~~~
                ip_n, t_s \Rightarrow_i t'_s
                ~~~~~
                t' = \mathsf{StitchTrace}(t, t'_s, \cS)\\
                \mathsf{where}~~t'_s \in \mathsf{Traces}(\mathsf{Program}(t_s))
            }
            {\{p_1~ic_1, p_2~ic_2 \ldots p_k~ic_k\}, t \Rightarrow_i t'}
        \end{array}
    \]
    \caption{Execution Semantics for Inference Metaprograms}
    \label{fig:inferlangexec}
\end{figure}

\begin{theorem}
    If all the subproblems used in our inference metaprograms are reversible and connect the space of their respective input probabilistic programs, then 
    all inference metaprograms in our inference metaprogramming language implement
    a generalized Markov kernel.
    \label{thm:convsub}
\end{theorem}
\begin{proof}
    Proof by induction over structure of inference metaprograms. 

    \noindent{\bf Base Case:}
    All black box inference algorithms in our inference metaprogramming language are 
    generalized Markov kernels. Hence given traces of program $p$ (which define the probability space $(T_p, \Sigma_p, \mu_p)$)
    the black box inference algorithm is $\mu_p$-irreducible, aperiodic, and has the stationary distribution $\mu_p$.

    \noindent{\bf Induction Case:}
    Consider the inference metaprogram $ip = \{p_1~ic_1, p_2, ic_2, \ldots, p_k~ic_k \}$,
    where $ic_i = \mathsf{infer}(\mathsf{SS}_i, ip_i)$ and $\sum\limits_{i=1}^k p_i = 1$.

    Using the induction hypothesis, we assume, for all $i \in \{1, 2, \ldots k\}$, all $ip_i$ implement 
    a generalized Markov kernel $K^{ip_i}$. Since our subproblem $\mathsf{SS}$ is reversible, we 
    lift the generalized Markov kernel to generalized class kernel~(Definition \ref{def:genclass}) $K_{f_{SS_i}}$, which
    for any program $p$ (which defines space $(T_p, \Sigma_p, \mu_p)$), is a class kernel.

    Given an probabilistic program $p$, the inference metaprogram $ip$ implements the Generalized Markov Kernel $K$,
    where
    \[
        K(p)(t, A) = \sum\limits_{i = 1}^k p_iK_{f_{SS_i}}(t, A)
    \]
    Using Theorems \ref{thm:convirr}, \ref{thm:convst}, and \ref{thm:convap}, 
    if the class functions $f_{SS_1}, f_{SS_2}, \ldots, f_{SS_k}$ connect the space $(T_p, \Sigma_p, \mu_p)$, 
    $K(p)$ is $\mu_p$-irreducible, aperiodic, and has $\mu_p$ as its stationary distribution.
\end{proof}

\begin{corollary}
    Given a probabilistic program $p$ (defining trace space $(T_p, \Sigma_p, \mu_p)$), inference 
    metaprograms which use reversible subproblem selection strategies 
    which connect the space of their respective probabilistic programs asymptotically 
    converge to $\mu_p$.
\end{corollary}

\section{Related Work}

\noindent{\bf Probabilistic Programming Languages:} Over the last several decades
researchers have developed  a range of probabilistic programming languages. 
With current practice each language typically comes paired with one/a few black box inference strategies.
Example language/inference pairs include Stan~\cite{carpenter2016stan} with 
Hamiltonian Monte Carlo inference~\cite{andrieu2003introduction}; Anglician~\cite{tolpin2015probabilistic} with 
particle Gibbs, etc. 
Languages like LibBi~\cite{murray2013bayesian}, Edward~\cite{tran2017deep} and Pyro~\cite{goodman2017pyro} provide 
inference customization mechanisms, but 
without subproblems or asymptotic convergence guarantees.
 
\noindent{\bf Compilation Strategies for Probabilistic Programs:} 
Techniques for efficiently executing probabilistic programs are a prerequisite for their widespread adoption. 
The Swift~\cite{wu2016swift} and Augur~\cite{huang2017compiling, tristan2014augur} compilers generate
efficient compiled implementations of inference algorithms that operate over probabilistic programs. 
We anticipate that applying these compilation techniques to subproblems can deliver significant performance
improvements for the hybrid inference algorithms we study in this paper. 

\noindent{\bf Subproblem Inference:} 
Both Turing~\cite{ge2018turing} and Venture~\cite{mansinghka2018probabilistic} provide inference metaprogramming
constructs with subproblems and different inference algorithms that operate on these
subproblems. In Venture subproblem inference is performed over full program traces,
with subproblems entangled with the full trace. The inference algorithms in Venture
must therefore operate over the entire trace while ensuring that the inference
effects do not escape the specified subproblem. Our extraction and stitching technique
eliminates this entanglement and enables the use of standard inference algorithms
that operate over complete traces while still supporting subproblem identification
and inference. Turing only provides  mechanisms that target specific stochastic choices in the context of the 
complete probabilistic computation. 

There is work on extending Gibbs-like algorithms to work on Open Universe Probabilisitic Models~\cite{arora2012gibbs,milch2010extending}.
This research studies algorithms that apply one strategy to choose a single variable at each iteration and rely on empirical
evidence of convergence. Our research, in contrast, supports a wide range of subproblem selection strategies 
and provides formal proofs of asymptotic convergence properties for MCMC algorithms applied to these subproblems. 

\noindent{\bf Asymptotic Convergence:}
There is a vast literature on asymptotic convergence of Markov chain algorithms
in various statistics and probability settings~\cite{meyn2012markov, tierney1994markov}. Our work is unique
in that it provides the first characterization of asymptotic convergence 
for subproblem inference in probabilistic programs. 
Complications that occur in this setting include mixtures of discrete
and continuous variables, stochastic choices with cascading effects
that may change the number of stochastic choices in the computation,
and resulting sample spaces with unbounded numbers of random variables. 
Standard results from computational statistics, computational physics,
and Monte-Carlo methods focus on finite dimensional discrete state spaces,
a context in which linear algebra (i.e., spectral analysis of the transition
matrix of the underlying Markov chain~\cite{diaconis1991geometric} or coupling arguments~\cite{levin2017markov})
is sufficient to prove convergence. State spaces with continuous random
variables are outside the scope of these formal analyses. Measure-theoretic
treatments are more general~\cite{roberts2004general}. Our results show how to apply
the concepts in these treatments to prove asymptotic convergence results for
probabilistic programs with interference metaprogramming. 

\noindent{\bf Proving Properties of Probabilistic Inference:} 
Researchers have recently developed techniques for proving a variety of properties of different 
inference algorithms for probabilistic programs~\cite{scibor2018functional,scibior2017denotational,atkinson2017typesafety,anon2020typesystem}.
Our unique contribution relates to the treatment of convergence in the context of subproblems, 
specifically 1) the identification of subproblem extraction and
stitching to obtain independent subproblems, 2) support for a general class of 
(potentially state-dependent) subproblems, including programmable 
subproblem selection strategies that may depend on the values of stochastic
choices from the current execution trace, 
and 3) the mathematical formulation that enables us to state and prove asymptotic convergence results for hybrid inference metaprograms
that apply (a general class of potentially very different) MCMC algorithms to different parts of the inference problem. 
We see our research and the research cited above as synergistic --- one potential synergy is 
that the research cited above can prove properties of the black-box MCMC algorithms 
that our inference metaprograms deploy.

\nocite{dal2012probabilistic, ramsey2002stochastic,
borgstrom2016lambda, 
staton2016semantics,
kozen1979semantics,
gehr2016psi,
monniaux2007abstract,
sharir1984verification,
atkinson2017typesafety,
huang2017compiling,
roberts1994simple,
matthijs2019domain,
scibior2017denotational,
heunen2017convenient,
scibor2018functional,
wu2016swift,
arora2012gibbs,
milch2010extending,
arora2012gibbs,
wingate2011lightweight,
higdon1998auxiliary,
}

\section{Conclusion}

Inference metaprogramming, subproblem inference, and asymptotic convergence are key issues in 
probabilistic programming. 
Detangling the subproblem from the surrounding program trace allows 
us to cleanly analyze subproblem based inference. 
Our mathematical framework
introduces new concepts which enable us to 
model subproblem based inference and 
prove asymptotic convergence properties of the 
resulting hybrid probabilistic inference algorithms.

\bibliography{citations.bib}

\newpage
\appendix
\section{Appendix}
\subsection{Convergence}
\label{append:conv}

\begin{lemma}
    \label{appendlemma:Kfexpand}
    For all $t \in T$ and $A \in \Sigma$, 
    \[K^n_f(t, A) = K^n_i(x)(y, V)I(x, U)\]
    where $f(t) = \tup{x, y} \in X_i \times Y_i$ and $U \times V = f(A) \cap X_i \times Y_i$. 
\end{lemma}
\begin{proof}
    Proof by induction.

    Base case:
    \[K_f(t, A) = K_i(x)(y, V)I(x, U)\] 
    using definition of $K_f$.

    Induction Hypothesis:
    
    For all $1 \leq n \leq m$, the following statement is true 
    \[K^n_f(t, A) = K^n_i(x)(y, V)I(x, U)\]

    Induction Case:
    \[
        K^{m+1}_f(t, A) = \int_{t' \in T} K^m_f(t', A)K_f(t, dt')
    \]
    \[
        =  \int_{x' \in X_i}\int_{y' \in Y_i} K^m_f(f(x', y'), A) K_i(x)(y, dy')I(x, dx')
    \]
    Note that $f(x', y') \in X_i \times Y_i$
    \[
        = \int_{x' \in X_i}\int_{y' \in Y_i} K^m_i(x')(y', V)I(x', U) K_i(x)(y, dy')I(x, dx')
    \]
    \[ 
        = \int_{y' \in Y_i} K^m_i(x)(y', V)I(x, U) K_i(x)(y, dy')
    \]
    \[ 
        = \int_{y' \in Y_i} K^m_i(x)(y', V)I(x, U) K_i(x)(y, dy')
    \]
    \[
        = K^{m+1}_i(x)(y, V)I(x, U)
    \]
    Hence proved.
\end{proof}

\begin{lemma}
    \[
        \pi(A) = \int_{t \in T} K_f(t, A) \pi(dt)
    \]
    \label{lem:convappend2}
\end{lemma}
\begin{proof}
    Every set $A \in \Sigma$ can be written as $f^{-1}(U \times V)$ for some
    $U \times V \in \cC^f$ as $f$ is a two-way measurable function.
    \[
        \int_{t \in T} K_f(t, A) \pi(dt) =  \int_{t \in T} K_f(t, f^{-1}(U \times V)) \pi(f^{-1}(dt))
    \]
    We can split the integral into sum over the constituent product spaces $X_i \times Y_i$,
    \[
        = \sum\limits_{i \in I}\int_{x \in X_i} \int_{y \in Y_i} 
        K_f(f^{-1}(x, y), f^{-1}(U \times V)) f_*(\pi)_i(dx \times dy)
        f_*(\pi)(X_i \times Y_i)
    \]
    We can rewrite $  K_f(f^{-1}(x, y), f^{-1}(U \times V)) $ as  
    $  K(x)(y, V \cap Y_i) I(x, U \cap X_i)$.
    \[
        =  \sum\limits_{i \in I} \int_{x \in X_i} \int_{ y \in Y_i} 
         K(x)(y, V \cap Y_i) I(x, U \cap X_i) f_*(\pi)_i(dx \times dy)
          f_*(\pi)(X_i \times Y_i)
    \]
    $\int_{x \in X_i} f(x) I(x, U \cap X_i) m(dx) = \int_{x \in U \cap X_i} f(x) m(dx)$ as 
    $I(x, U \cap X_i)$ is zero for all $x \notin U \cap X_i$.
    \[
        =  \sum\limits_{i \in I} \int_{ y \in Y_i} \int_{x \in U \cap X_i}  
         K(x)(y, V \cap Y_i) f_*(\pi)_i(dx \times dy)
          f_*(\pi)(X_i \times Y_i)
    \]
    We can rewrite $f_*(\pi)_i(dx' \times dy)$ as $v_{f_*(\pi)_i)}(x, dy)f_*(\pi)_i(dx' \times Y_i)$ 
    using the definition of regular conditional probability distribution.
    \[
        =  \sum\limits_{i \in I}  \int_{ y \in Y_i} 
        \int_{x \in U \cap X_i} K(x)(y, V \cap Y_i) v(x, dy) f_*(\pi)_i(dx \times Y_i)
          f_*(\pi)(X_i \times Y_i)
    \]
    \[
        =  \sum\limits_{i \in I} 
        \int_{x \in U \cap X_i} \big(\int_{ y \in Y_i} 
    K(x)(y, V \cap Y_i)) v_{f_*(\pi)_i}(x, dy) ~\big) f_*(\pi)_i(dx \times Y_i)
          f_*(\pi)(X_i \times Y_i)
    \]
$ \int_{ y \in Y_i} 
    K_i(x)(y, V \cap Y_i) v_{f_*(\pi)_i}(x, dy) = v_{f_*(\pi)_i}(x, V \cap Y_i) $ as 
    $v_{f_*(\pi)_i}(x, .)$ is the stationary distribution for kernel $K_i(x)$. 
    \[
        =  \sum\limits_{i \in I} 
        \int_{x \in U \cap X_i} v_{f_*(\pi)_i}(x, V \times Y_i) f_*(\pi)_i(dx \times Y_i)
          f_*(\pi)(X_i \times Y_i)
    \]
    \[
        =  \sum\limits_{i \in I} 
         f_*(\pi)_i(U \times V \cap X_i\times Y_i)
          f_*(\pi)(X_i \times Y_i) = f_*(\pi)(U \times V) = \pi(A)
    \]

\end{proof}

\begin{lemma}
    \label{lem:classaper}
 $K_f$ is aperiodic if for at least one $x \in X_i$ for some $i \in I$, 
    $K_i(x) : Y_i \times \cY_i \rightarrow [0, 1]$ 
    is aperiodic.
\end{lemma}
\begin{proof}
     Proof by contradiction. Let us assume $K_f$ is periodic, i.e.,
    there exists an integer $d \geq 2$, and  a sequence 
    $\{E_0, E_1, \ldots E_{d-1}\}$ and $N$ of $d$ non-empty disjoint sets in $\tau$ such 
    that, for all $i = 0, 1, \ldots d-1$ and for all $t \in E_i$,
    \begin{enumerate}
    \item $ (\cup_{i=0}^d E_i) \cup N = T$
    \item $K_f(t, E_j) = 1 \text{ for } j = i + 1 (\mathsf{mod}~d)$
    \item $\pi(N) = 0$
    \end{enumerate}

    For all $t \in E_i$, $K_f(t, E_j) = 1 \text{ for } j = i + 1 (\mathsf{mod}~d)$
    Consider any $k \in I$, any $x \in X_k$ 
    and $U_i \times V_i = f(E_i) \cap X_k \times Y_k$.

    Consider a trace $t \in E_i$, such that $f_x(t) = x$. 
    Since $K_f(t, E_{i+1}) = 1 \leq I(x, U_{i+1})$, for 
    all $i = 0, 1 \ldots d-1$, there exists a trace $t' \in E_i$,
    such that $f(t') = \tup{x, y}$.

    $K_f(t, E_{i+1})  = 1 \leq K(x)(y, V_{i+1})$, Hence if 
    $K_f$ is aperiodic, then for all $x$, $K(x)$ is aperiodic.

\end{proof}

\begin{lemma}
    For any $f \in \cF$, any element $t' \in B_t^\infty$ such that 
    $f(t') = \tup{x, y} \in X_i^f \times Y_i^f$, and any 
    set $U \times V \in \sigma(\cX_i^f \otimes \cY_i^f)$ such that $A = f^{-1}(U \times V)$, 
    the following condition holds true
    \[
        v_{f_*(\pi)_i}(x, V) > 0 \wedge x \in U \implies A \in R_t^\infty
    \]
    \label{lem:appendconv3}
\end{lemma}
\begin{proof}
    Consider class Kernel $K_f(t', A) = K_i(x)(y, V)I(x, U)$.
    Since $v_{f_*(\pi)_i}(x, V) > 0$ and $K_i$ is $v_{f_*(\pi)_i}$-irreducible, 
    there exists an $n$ such that
    \[
        K^n_i(x)(y, V) > 0
    \]
    Since $x \in U$, using Lemma \ref{lemma:Kfexpand}
    \[
        K^n_f(t', A) = K^n_i(x)(y, V)I(x, U) > 0
    \]
    \[ 
        K^n_f(t', A) > 0 \implies K^n(t', A) > 0
    \]
    Since $t' \in B_t^\infty$, 
    there exists an $n'$ such that,
    for all sets $B \in \Sigma$ with $t' \in B$, 
    $K^{n'}(t, B) > 0$. Hence 
    \[
        K^{n + n'}(t, A) \geq \int_B K^n(t', A)K^{n'}(t, dt') > 0 \implies A \in R^\infty_t
    \]
\end{proof}

\begin{lemma}
    \label{appendlemma:fxA}
    For any positive probability set $A$ and any function $f \in \cF$, 
    if $A \subseteq f^{-1}_x(f_x(B_t^\infty))$ then $A \in R_t^\infty$.
\end{lemma}
\begin{proof}
    Given a set $A$, we can treat $f(A)$ as a union of sets 
    $\{ U_i \times V_i \vert i \in I^f \}$, where $U_i \times V_i$ are 
    elements of set $f(A)$ which are elements of the set $X_i^f \times Y_i^f$ (i.e.
    $U_i \times V_i  = f(A) \cap X_i^f \times Y_i^f$).
    Since $\pi(A) > 0$, $f_*(\pi)(f(A)) > 0$ and for at least for one 
    $i \in I_f$, $f_*(\pi)_i(U_i \times V_i) > 0$.

    Since $A \subseteq f^{-1}_x(f_x(B_t^\infty))$, 
    for each $x \in U_i$
    there exists at least 
    one element $t' \in B_t^\infty$ such that $f_x(t') = x$.
    
    If $f_*(\pi)_i(U_i \times V_i) > 0$, there exists at least one $t' \in B_t^\infty$
    such that $f(t') = \tup{x, y} \in U_i \times Y_i^f$ and $v_{f_*(\pi)_i}(x, V_i) > 0$.
    Hence $f^{-1}(U_i \times V_i) \in R_t^\infty$.
\end{proof}

\begin{lemma}
    If $\cF$ connects the space $(T, \Sigma, \pi)$ then there does not exist a positive probability set $A \in \Sigma$, 
    such that $A \subseteq \bigcap_{f \in \cF} f_x^{-1}((f_x(B^\infty_t))^c)$.
    \label{lem:appendconv4}
\end{lemma}
\begin{proof}
    Proof by Contradiction.

    Let us assume such a set $A$ exists.
    If $A$ is a positive probability set, then $\pi(\bigcap_{f \in \cF} f_x^{-1}((f_x(B^\infty_t))^c)) > 0$.

    For any two functions $f, g \in \cF$, 
    \[
        \pi(f_x^{-1}((f_x(B^\infty_t))^c) \cap g_x^{-1}(g_x(B^\infty_t)) ) = 0
    \]
    The set $ f_x^{-1}((f_x(B^\infty_t))^c)$
    only contains elements which are not in $B^\infty_t$ and  $g_x^{-1}(g_x(B^\infty_t))$
    contains elements $t'$ such that there exists at least one element $t'' \in B^\infty_t$
    with $g_x(t') = g_x(t'')$. 

    Any positive probability set $B \subseteq g_x^{-1}(g_x(B^\infty_t))$ is also a 
    subset of $B^\infty_t$ (using Lemma \ref{lemma:fxA}).
    Hence \[
        \pi(f_x^{-1}((f_x(B^\infty_t))^c) \cap g_x^{-1}(g_x(B^\infty_t)) ) = 0 
    \]
    Similarly 
    \[
        \pi(f_x^{-1}(f_x(B^\infty_t)) \cap g_x^{-1}((g_x(B^\infty_t))^c) ) = 0
    \]

    Since  $\pi(\bigcap_{f \in \cF} f_x^{-1}((f_x(B^\infty_t))^c)) > 0$, 
    \[
        \pi(f_x^{-1}(  f_x^{-1}((f_x(B^\infty_t))^c)) > 0 \text{ and } 
        \pi( g_x^{-1}((g_x(B^\infty_t))^c))  > 0
    \]
    
    But this contradicts the fact the $\cF$ connects the space $(T, \Sigma, \pi)$.
    Hence no such set $A$ exists.
\end{proof}

\begin{theorem} 
    \label{appendthm:convirr}
    If $\cF$ connects the space $(T, \Sigma, \pi)$ then the Markov Transition Kernel 
    $K$ is $\pi$-irreducible.    
\end{theorem}
\begin{proof}
    Proof by contradiction. Let us assume $K$ is not $\pi$-irreducible, then there exists a
    positive probability set $A \in \Sigma$ such that $A \notin R_t^\infty$,  
    then

    If 
       $ \pi(A \cap B_t^\infty) > 0$, there exists a set $B \subseteq B_t^\infty$ and $B \subseteq A$
    which implies $A \in R_t^\infty$. Hence $\pi(A \cap B_t^\infty) = 0$.

    For any $f \in \cF$,
        $\pi(A \cap f_x^{-1}(f_x(B^\infty_t))) > 0$, 
        there exists a set $B \in R_t^\infty$ and $B \subseteq A$
    which implies $A \in R_t^\infty$. Hence  $\pi(A \cap f_x^{-1}(f_x(B^\infty_t))) = 0$. 

    Since $f_x$ is a 2-way measurable function (and one-one function from sets to sets), for any set $B$
    $f_x^{-1}(f_x(B))^c = f_x^{-1}(f_x(B)^c)$.

    Since $\pi(A) > 0$ and 
      For any $f \in \cF$
    $
        \pi(A \cap f_x^{-1}(f_x(B^\infty_t))) = 0
    $
    this means
    \[
        \pi(A \cap \bigcap_{f \in \cF} f_x^{-1}(f_x(B^\infty_t)^c)) > 0
    \]
    which means there exists a positive probability set $B \in \Sigma$
    and $B \subseteq f_x^{-1}(f_x(B^\infty_t)^c))$,
    which is impossible. 

    Hence no such set $A$ exists. $K$ is $\pi$-irreducible.
\end{proof}

\begin{theorem}
    $\pi$ is the stationary distribution of Markov Kernel $K$, i.e.
    \[\int K(t, A)\pi(dt) = \pi(A) \text{ for all } A \in \Sigma\]
    \label{appendthm:convst}
\end{theorem}
\begin{proof}
    \[
        \int_{t \in T} K(t, A) \pi(dt) = \int_{t \in T} \sum\limits_{i = 1}^m p_iK_{f_i}(t, A)\pi(dt)
    \]
   \[
        = \sum\limits_{i=1}^m p_i\int_{t \in T} K_{f_i}(t, A)\pi(dt) = \sum\limits_{i=1}^m p_i\pi(A) = \pi(A) 
    \]

\end{proof}

\begin{theorem}
    \label{appendthm:convap}
    The Markov Transition Kernel $K$ is aperiodic if at least one of the 
    class kernels $K_{f_j}$ is aperiodic.
\end{theorem}
\begin{proof}
    Proof by Contradiction.

    Let us assume $K$ is periodic.
     i.e.,
    there exists an integer $d \geq 2$ and a sequence 
    $\{E_0, E_1, \ldots E_{d-1}\}$ and $N$ of $d$ non-empty disjoint sets in $\tau$ such 
    that, for all $i = 0, 1, \ldots d-1$ and for all $t \in E_i$,
    \begin{enumerate}
    \item $ (\cup_{i=0}^d E_i) \cup N = T$
    \item $K(t, E_j) = 1 \text{ for } j = i + 1 (\mathsf{mod}~d)$
    \item $\pi(N) = 0$
    \end{enumerate}

    If $K(t, E_j) = 1$ then for all $f \in \cF$, $K_f(t, E_j) = 1$.
    Therefore, for all $i = 0, 1, \ldots d-1$ and for all $t \in E_i$,
    $K_f(t, E_j) = 1 \text{ for } j = i + 1 (\mathsf{mod}~d)$.

    Hence if $K$ is periodic, then for all $f \in \cF$, $K_f$ is periodic.
     
     Hence by contradiction, $K$ is aperiodic.
\end{proof}

\newpage

\subsection{Soundness}
\label{append:sound}

\begin{observation}
 Note that whenever a rule in $\Rightarrow_{ex}$ introduces an 
    $\mathsf{assume}$ statement in the subtrace, it creates a new variable name. Therefore
    variable names in the new subtrace do not conflict with any variable names previously introduced
    in another part of the trace. This fact will be used at various points within this paper.
\label{ex:obs}
\end{observation}

\begin{observation}
    Note that whenever $\SAC ae \Rightarrow_{ex} ae_s, t_s$, $\cV(ae) = \cV(ae_s)$ and
    whenever $\SAC ae, ae_s, t_s \Rightarrow_{st} ae'$, $\cV(ae') = \cV(ae_s)$.
    \label{ex:obs2}
\end{observation}

\Comment{
\begin{lemma}
    \label{lem:sound1}
    Given a valid trace $t$ and a valid subproblem $\cS$, for all
    subtraces $t_s$
    \[t_s = \mathsf{ExtractTrace}(t, \cS) \implies \exists~p.~
    t_s \in \mathsf{Traces}(p)\]
\end{lemma}

\begin{proof}
    I prove this using structural induction over structure of traces and augmented
    expressions.
    \begin{lemma}
    Given an augmented expressions $ae$, for any $\sigma_v, \sigma_{id}$, 
    the following holds true:
    \[
        \begin{array}{c}
        \exists~e.~\SI{}{} e \Rightarrow_s \_,\_ ae \wedge \SAC ae \Rightarrow_{ex} ae_s, t_s
            \implies \exists~p.~ \SI{}{} p \Rightarrow_s t_s;\assume{z}{ae_s}
        \end{array}
    \]
    \end{lemma}
    \begin{proof}
    \noindent{Base Case:}
    
    \noindent{Case 1:} 
    $ae = \ndi{x}{x}{id}$.

    \noindent{By assumption}
    \[
        \exists~e.~\SI{}{} e \Rightarrow_s \_,\_ ae 
    \]
    By definition of $\Rightarrow_s$ 
     \[
         \SI{}{} x \Rightarrow_s \_, \_, \ndi{x}{x}{id}
     \]
     Then
     $e = x$, $x \notin \dom \sigma_v$ and $id$ is a unique id.
     
        \noindent{By assumption}
        \[
\SAC ae \Rightarrow_{ex} ae_s, t_s
        \]
        By definition of $\Rightarrow_{ex}$ 
     \[
         \SAC \ndi{x}{x}{id} \Rightarrow_{ex} \ndi{x}{x}{id}, \emptyset
     \] 
    Then $ae_s = \ndi{x}{x}{id}$ and $t_s = \emptyset$.

        \noindent{Let} $p = \assume{z}{ae_s}$. Because $x \notin \dom \sigma_v$,
      definition of $\Rightarrow_s$ implies
        \[
        \SI{}{} p \Rightarrow_s \assume{z}{\ndi{x}{x}{id}}
      \]

        \noindent{Then} $ae = \ndi{x}{x}{id}$
        implies
    \[
        \begin{array}{c}
        \exists~e.~\SI{}{} e \Rightarrow_s \_,\_ ae \wedge \SAC ae \Rightarrow_{ex} ae_s, t_s
            \implies \exists~p.~ \SI{}{} p \Rightarrow_s t_s;\assume{z}{ae_s}
        \end{array}
    \]

    \noindent{Case 2:} 
    $ae = \ndi{x(id')}{v}{id}$

    \noindent{By assumption}
     \[
        \exists~e.~\SI{}{} e \Rightarrow_s \_,\_ ae 
    \]
        By definition of $\Rightarrow_s$
     \[
         \SI{}{} x \Rightarrow_s \_, \_, \ndi{x(id')}{v}{id}
     \]
     Then $v = \sigma_v(x)$, $id' = \sigma_{id}(x)$ and $id$ is a unique id.
     
        \noindent{By assumption}
        \[
\SAC ae \Rightarrow_{ex} ae_s, t_s
        \]
        By definition of $\Rightarrow_{ex}$
     \[
         \SAC \ndi{x(id')}{v}{id} \Rightarrow_{ex} \ndi{x(id')}{v}{id}, \emptyset
     \]
     Then
        $ae_s = \ndi{x(id')}{v}{id}, t_s = \emptyset$.

      \noindent{Let} $p = \assume{z}{ae_s}$, Because $v = \sigma_v(x)$ and $id' = \sigma_{id}(x)$,
      the definition of $\Rightarrow_s$ implies
      \[
          \SI{}{} p \Rightarrow_s \assume{z}{\ndi{x(id')}{v}{id}}
      \]

      \noindent{Then} $ae = \ndi{x(id')}{v}{id}$ implies
    \[
        \begin{array}{c}
        \exists~e.~\SI{}{} e \Rightarrow_s \_,\_ ae \wedge \SAC ae \Rightarrow_{ex} ae_s, t_s
            \implies \exists~p.~ \SI{}{} p \Rightarrow_s t_s;\assume{z}{ae_s}
        \end{array}
    \]

    \noindent{\bf Case 3:}
    $ae = \ndi{\lambda.x~e'}{v}{id}$
   
    \noindent{By assumption}
    \[
        \exists~e.~\SI{}{} e \Rightarrow_s \_,\_ ae 
    \]
    By definition of $\Rightarrow_s$ 
    \[
        \SI{}{} \lambda.x~e' \Rightarrow_s \_, \_, \ndi{\lambda.x~e'}{v}{id}
    \]
    Then $e = \lambda.x~e'$.

    \noindent{By assumption}
      \[
\SAC ae \Rightarrow_{ex} ae_s, t_s
        \]
        By definition of $\Rightarrow_{ex}$
    \[
        \SAC \ndi{\lambda.x~e'}{v}{id} \Rightarrow_{ex} \ndi{\lambda.x~e'}{v}{id}, \emptyset
    \]
    Then  $ae_s = \ndi{\lambda.x~e'}{v}{id}$ and $t_s = \emptyset$.

    \noindent{Let}
    $p = \assume{z}{e}$. Because $
        \SI{}{} \lambda.x~e' \Rightarrow_s \_, \_, \ndi{\lambda.x~e'}{v}{id}
    $, the definition of $\Rightarrow_s$ implies
    \[
        \SI{}{} p \Rightarrow_s \assume{z}{\ndi{\lambda.x~e'}{v}{id}}
    \]

    \noindent{Then} $ae = \ndi{\lambda.x~e'}{v}{id}$ implies
     \[
        \begin{array}{c}
        \exists~e.~\SI{}{} e \Rightarrow_s \_,\_ ae \wedge \SAC ae \Rightarrow_{ex} ae_s, t_s
            \implies \exists~p.~ \SI{}{} p \Rightarrow_s t_s;\assume{z}{ae_s}
        \end{array}
    \]

    \noindent{\bf Induction Step:}

\noindent{\bf Case 1:}
    $ae = \ndi{(ae_1~ae_2)\perp}{v}{id}$
    
    \noindent{By assumption}
    \[
        \exists~e.~\SI{}{} e \Rightarrow_s \_,\_ ae 
    \]
    By definition of $\Rightarrow_s$
    \[
        \SI{}{} (e_1~e_2) \Rightarrow_s \_, \_, \ndi{(ae_1~ae_2)\perp}{v}{id}
    \]
    Then $e = (e_1~e_2)$, $\SI{}{} e_1 \Rightarrow_s \_,\_ ae_1$ 
    and $\SI{}{} e_2 \Rightarrow_s \_,\_ ae_2$. 

    \noindent{By assumption}
      \[
    \SAC ae \Rightarrow_{ex} ae_s, t_s
        \]
        By definition of $\Rightarrow_{ex}$
    \[
        \SAC \ndi{(ae_1~ae_2)\perp}{v}{id} \Rightarrow_{ex} \ndi{(ae'_s~ae''_s)\perp }{v}{id}, t_s
    \]
    Then $ae_s = \ndi{(ae'_s~ae''_s)\perp }{v}{id}$, 
  $\SAC ae_1 \Rightarrow_{ex} ae'_s, t'_s$, 
    $\SAC ae_2 \Rightarrow_{ex} ae''_s, t''_s$, 
    and $t_s = t'_s;t''_s$.

    \noindent{By induction hypothesis}
    \[
        \begin{array}{c}
            \exists~e_1.~\SI{}{} e_1 \Rightarrow_s \_,\_ ae_1 \wedge \SAC ae_1 \Rightarrow_{ex} ae'_s, t'_s
        \\ \implies \exists~p'.~ \SI{}{} p' \Rightarrow_s t'_s;\assume{z}{ae'_s}
        \end{array}
    \]
    and
     \[
        \begin{array}{c}
            \exists~e_2.~\SI{}{} e_2 \Rightarrow_s \_,\_ ae_2 \wedge \SAC ae_2 \Rightarrow_{ex} ae''_s, t''_s
           \\  \implies \exists~p''.~ \SI{}{} p'' \Rightarrow_s t''_s;\assume{z}{ae''_s}
        \end{array}
    \]

    \noindent{Because} $\SI{}{} e_1 \Rightarrow_s \_,\_ ae_1$, $\SI{}{} e_2 \Rightarrow_s \_,\_ ae_2$, 
    $\SAC ae_1 \Rightarrow_{ex} ae'_s, t'_s$, and
    $\SAC ae_2 \Rightarrow_{ex} ae''_s, t''_s$, induction hypothesis implies
    \[
         \exists~p';\assume{z}{e'_s}.~ \SI{}{} p';\assume{z}{e'_s} \Rightarrow_s t'_s;\assume{z}{ae'_s}
        \]
        and
    \[
        \exists~p'';\assume{z}{e''_s}.~ \SI{}{} p'';\assume{z}{e''_s} \Rightarrow_s t''_s;\assume{z}{ae''_s}
    \]

    \noindent{Let} $p = p';p'';\assume{z}{(e'_s~e''_s)}$. Because all the variable names introduced by $t''_s$ 
    do not conflict with variable names in $ae'_s$ (Observation~\ref{ex:obs}), the definition of $\Rightarrow_s$ implies
    \[
        \SI{}{} p \Rightarrow_s t'_s;t''_s;\assume{z}{(ae'_s~ae''_s)\perp}
    \]

    \noindent{Then} by the induction 
    hypothesis, $ae = \ndi{(ae_1~ae_2)\perp}{v}{id}$ implies 
     \[
        \begin{array}{c}
        \exists~e.~\SI{}{} e \Rightarrow_s \_,\_ ae \wedge \SAC ae \Rightarrow_{ex} ae_s, t_s
            \\ \implies \exists~p.~ \SI{}{} p \Rightarrow_s t_s;\assume{z}{ae_s}
        \end{array}
    \]

    \noindent{\bf Case 2:}
     $ae = \ndi{(ae_1~ae_2)y=ae_3}{v}{id}$ and $ID(ae_1) \notin \cS$
 
    \noindent{By assumption}
    \[
        \exists~e.~\SI{}{} e \Rightarrow_s \_,\_ ae 
    \]
    By definition of $\Rightarrow_s$
    \[
        \SI{}{} (e_1~e_2) \Rightarrow_s \_, \_, \ndi{(ae_1~ae_2)y=ae_3}{v}{id}
    \]
    Then $e = (e_1~e_2)$, $\SI{}{} e_1 \Rightarrow_s \_,\_ ae_1$, 
    $\SI{}{} e_2 \Rightarrow_s \_,\_ ae_2$, $\sigma'_v = \sigma''_v[y \rightarrow \cV(ae_1)]$, 
    $\sigma'_{id} = \sigma''_{id}[y \rightarrow ID(ae_1)]$,
    $\SI{'}{'} e_3 \Rightarrow_s \_, \_, ae_3$, and 
    $\cV(ae_1) = \tup{\lambda.x~e_3, \sigma''_v, \sigma''_{id}}$. 

    \noindent{By assumption}
      \[
    \SAC ae \Rightarrow_{ex} ae_s, t_s
        \]
        By definition of $\Rightarrow_{ex}$
    \[
        \SAC \ndi{(ae_1~ae_2)y=ae_3}{v}{id} \Rightarrow_{ex} \ndi{ae'''_s}{v}{id}, t_s
    \]
    Then $ae_s = \ndi{ae'''_s}{v}{id}$, 
  $\SAC ae_1 \Rightarrow_{ex} ae'_s, t'_s$, 
    $\SAC ae_2 \Rightarrow_{ex} ae''_s, t''_s$, 
    $\SAC ae_3 \Rightarrow_{ex} ae'''_s, t'''_s$
    and $t_s = t'_s;\assume{x}{ae'_s};t''_s;\assume{y}{ae''_s};t'''_s$.

    \noindent{By induction hypothesis}
    \[
        \begin{array}{c}
            \exists~e_1.~\SI{}{} e_1 \Rightarrow_s \_,\_ ae_1 \wedge \SAC ae_1 \Rightarrow_{ex} ae'_s, t'_s
        \\ \implies \exists~p'.~ \SI{}{} p' \Rightarrow_s t'_s;\assume{z}{ae'_s},
        \end{array}
    \]
     \[
        \begin{array}{c}
            \exists~e_2.~\SI{}{} e_2 \Rightarrow_s \_,\_ ae_2 \wedge \SAC ae_2 \Rightarrow_{ex} ae''_s, t''_s
           \\  \implies \exists~p''.~ \SI{}{} p'' \Rightarrow_s t''_s;\assume{z}{ae''_s},
        \end{array}
    \]
    and
    \[
        \begin{array}{c}
            \exists~e_3.~\SI{'}{'} e_3 \Rightarrow_s \_,\_ ae_3 \wedge \SAC ae_3 \Rightarrow_{ex} ae'''_s, t'''_s
           \\  \implies \exists~p'''.~ \SI{'}{'} p''' \Rightarrow_s t'''_s;\assume{z}{ae'''_s}
        \end{array}
    \]

    \noindent{Because} $\SI{}{} e_1 \Rightarrow_s \_,\_ ae_1$, $\SI{}{} e_2 \Rightarrow_s \_,\_ ae_2$,
    $\SI{'}{'} e_3 \Rightarrow_s \_, \_, ae_3$,
    $\SAC ae_1 \Rightarrow_{ex} ae'_s, t'_s$,
    $\SAC ae_2 \Rightarrow_{ex} ae''_s, t''_s$, and
    $\SAC ae_3 \Rightarrow_{ex} ae'''_s, t'''_s$,
    induction hypothesis implies
    \[
         \exists~p';\assume{z}{e'_s}.~ \SI{}{} p';\assume{z}{e'_s} \Rightarrow_s t'_s;\assume{z}{ae'_s}
        , \]
    \[
        \exists~p'';\assume{z}{e''_s}.~ \SI{}{} p'';\assume{z}{e''_s} \Rightarrow_s t''_s;\assume{z}{ae''_s}
    , \]
    and
    \[
\exists~p''';\assume{z}{e'''_s}.~ \SI{}{} p''';\assume{z}{e'''_s} \Rightarrow_s t'''_s;\assume{z}{ae'''_s}
    \]

    \noindent{Let} $p = p';\assume{x}{e'_s};p'';\assume{y}{e''_s};p''';\assume{z}{e'''_s}$. Because all the variable names introduced by $t''_s$ 
    do not conflict with variable names in $ae'_s$ (Observation~\ref{ex:obs}), the definition of $\Rightarrow_s$ implies XXX
    \[
        \SI{}{} p \Rightarrow_s t'_s;\assume{x}{ae'_s};t''_s;\assume{y}{ae''_s};t'''_s\assume{z}{ae'''_s}
    \]

    \noindent{Then} by the induction 
    hypothesis, $ae = \ndi{(ae_1~ae_2)y=ae_3}{v}{id}$ implies 
     \[
        \begin{array}{c}
        \exists~e.~\SI{}{} e \Rightarrow_s \_,\_ ae \wedge \SAC ae \Rightarrow_{ex} ae_s, t_s
            \\ \implies \exists~p.~ \SI{}{} p \Rightarrow_s t_s;\assume{z}{ae_s}
        \end{array}
    \]

   \noindent{\bf Case 3:}
    $ae = \ndi{(ae_1~ae_2)y=ae_3}{v}{id}$ and $ID(ae_1) \in \cS$

    \noindent{By assumption}
    \[
        \exists~e.~\SI{}{} e \Rightarrow_s \_,\_ ae 
    \]
    By definition of $\Rightarrow_s$
    \[
        \SI{}{} (e_1~e_2) \Rightarrow_s \_, \_, \ndi{(ae_1~ae_2)y=ae_3}{v}{id}
    \]
    Then $e = (e_1~e_2)$, $\SI{}{} e_1 \Rightarrow_s \_,\_ ae_1$, 
    $\SI{}{} e_2 \Rightarrow_s \_,\_ ae_2$, $ae_1 = \tup{\lambda.x~e_3, \sigma'_v, \sigma'_{id}}$, and 
    $\SI{'[y \rightarrow \cV(ae_2)]}{'[y \rightarrow ID(ae_2)} e_3 \Rightarrow_s \_, \_, ae_3$. 

    \noindent{By assumption}
      \[
    \SAC ae \Rightarrow_{ex} ae_s, t_s
        \]
        By definition of $\Rightarrow_{ex}$
    \[
        \SAC \ndi{(ae_1~ae_2)y=ae_3}{v}{id} \Rightarrow_{ex} \ndi{(ae'_s~ae''_s)y=ae_3 }{v}{id}, t_s
    \]
    Then $ae_s = \ndi{(ae'_s~ae''_s)y=ae_3}{v}{id}$, 
  $\SAC ae_1 \Rightarrow_{ex} ae'_s, t'_s$, 
    $\SAC ae_2 \Rightarrow_{ex} ae''_s, t''_s$, 
    and $t_s = t'_s;t''_s$.

    \noindent{By induction hypothesis}
    \[
        \begin{array}{c}
            \exists~e_1.~\SI{}{} e_1 \Rightarrow_s \_,\_ ae_1 \wedge \SAC ae_1 \Rightarrow_{ex} ae'_s, t'_s
        \\ \implies \exists~p'.~ \SI{}{} p' \Rightarrow_s t'_s;\assume{z}{ae'_s}
        \end{array}
    \]
    and
     \[
        \begin{array}{c}
            \exists~e_2.~\SI{}{} e_2 \Rightarrow_s \_,\_ ae_2 \wedge \SAC ae_2 \Rightarrow_{ex} ae''_s, t''_s
           \\  \implies \exists~p''.~ \SI{}{} p'' \Rightarrow_s t''_s;\assume{z}{ae''_s}
        \end{array}
    \]

    \noindent{Because} $\SI{}{} e_1 \Rightarrow_s \_,\_ ae_1$, $\SI{}{} e_2 \Rightarrow_s \_,\_ ae_2$, 
    $\SAC ae_1 \Rightarrow_{ex} ae'_s, t'_s$, and
    $\SAC ae_2 \Rightarrow_{ex} ae''_s, t''_s$, induction hypothesis implies
    \[
         \exists~p';\assume{z}{e'_s}.~ \SI{}{} p';\assume{z}{e'_s} \Rightarrow_s t'_s;\assume{z}{ae'_s}
        \]
        and
    \[
        \exists~p'';\assume{z}{e''_s}.~ \SI{}{} p'';\assume{z}{e''_s} \Rightarrow_s t''_s;\assume{z}{ae''_s}
    \]

    \noindent{Let} $p = p';p'';\assume{z}{(e'_s~e''_s)}$. Because all the variable names introduced by $t''_s$ 
    do not conflict with variable names in $ae'_s$ (Observation~\ref{ex:obs}), $\cV(ae_1) = \cV(ae'_s)$, and
   $\SI{'[y \rightarrow \cV(ae_2)]}{'[y \rightarrow ID(ae_2)]} e_3 \Rightarrow_s \_, \_, ae_3$, the definition of $\Rightarrow_s$ implies
    \[
        \SI{}{} p \Rightarrow_s t'_s;t''_s;\assume{z}{(ae'_s~ae''_s)y=ae_3}
    \]

    \noindent{Then} by the induction 
    hypothesis, $ae = \ndi{(ae_1~ae_2)y=ae_3}{v}{id}$ implies 
     \[
        \begin{array}{c}
        \exists~e.~\SI{}{} e \Rightarrow_s \_,\_ ae \wedge \SAC ae \Rightarrow_{ex} ae_s, t_s
           \\ \implies \exists~p.~ \SI{}{} p \Rightarrow_s t_s;\assume{z}{ae_s}
        \end{array}
    \]

    \noindent{\bf Case 4:}
    $ae = \ndi{\mathsf{Dist}(ae' \#id_e) = ae_e }{v}{id}$ and $ id_e \in \cS$
   
    Assuming
    \[
        \begin{array}{c}
            \exists~e'.~\SI{}{} e' \Rightarrow_s \_,\_ ae' \wedge \SAC ae' \Rightarrow_{ex} ae'_s, t'_s
            \implies \exists~p'.~ \SI{}{} p' \Rightarrow_s t'_s;\assume{z}{ae'_s}
        \end{array}
    \]
    from induction hypothesis.
    
    And since
    \[
        \SI{}{} \mathsf{Dist}(e') \Rightarrow_s \_, \_, \ndi{\mathsf{Dist}(ae' \#id_e) = ae_e }{v}{id}
    \]
    \[
        \SAC \ndi{\mathsf{Dist}(ae' \#id_e) = ae_e }{v}{id} \Rightarrow_{ex} \ndi{\mathsf{Dist}(ae'_s \#id_e) = ae_e }{v}{id}, t_s
    \]
    The following is true
    \[
    \exists~e'.~\SI{}{} e' \Rightarrow_s \_,\_ ae' \wedge \SAC ae' \Rightarrow_{ex} ae'_s, t'_s
    \]
    Where $t_s = t'_s;t''_s$. Using induction hypothesis, the following is true.
    \[
         \exists~p';\assume{z}{e'_s}.~ \SI{}{} p';\assume{z}{e'_s} \Rightarrow_s t'_s;\assume{z}{ae'_s}
    \]

    Given $p = p';\assume{z}{\mathsf{Dist}(e')}$, the following is true
    \[
        \SI{}{} p \Rightarrow_s t'_s;\assume{z}{\ndi{\mathsf{Dist}(ae'_s \#id_e) = ae_e }{v}{id}}
    \]

    Hence, given induction 
    hypothesis, if    $ae = \ndi{\mathsf{Dist}(ae' \#id_e) = ae_e }{v}{id}$ and $ id_e \in \cS$
   then the theorem is true.

    \noindent{\bf Case 5:}
  $ae = \ndi{\mathsf{Dist}(ae' \#id_e) = ae_e }{v}{id}$ and $ id_e \notin \cS$

    Assuming
    \[
        \begin{array}{c}
            \exists~e'.~\SI{}{} e' \Rightarrow_s \_,\_, ae' \wedge \SAC ae' \Rightarrow_{ex} ae'_s, t'_s
            \implies \exists~p'.~ \SI{}{} p' \Rightarrow_s t'_s;\assume{z}{ae'_s}
        \end{array}
    \]
     and
     \[
        \begin{array}{c}
            \exists~e_e.~\SI{}{} e_e \Rightarrow_s \_,\_ ae_e \wedge \SAC ae_e \Rightarrow_{ex} ae''_s, t''_s
            \implies \exists~p''.~ \SI{}{} p'' \Rightarrow_s t''_s;\assume{z}{ae''_s}
        \end{array}
    \]
    from induction hypothesis.
    
    And since
    \[
        \SI{}{} \mathsf{Dist}(e') \Rightarrow_s \_, \_, \ndi{\mathsf{Dist}(ae'\#id_e)=ae_e}{v}{id}
    \]
    \[
        \SAC  \ndi{\mathsf{Dist}(ae'\#id_e)=ae_e}{v}{id}
 \Rightarrow_{ex} \ndi{ae''_s }{v}{id}, t_s
    \]
    The following is true
    \[
    \exists~e'.~\SI{}{} e' \Rightarrow_s \_,\_, ae' \wedge \SAC ae' \Rightarrow_{ex} ae'_s, t'_s
    \]
    \[
 \exists~e_e.~\SI{}{} e_e \Rightarrow_s \_,\_, ae_e \wedge \SAC ae_e \Rightarrow_{ex} ae''_s, t''_s
    \]
    Where $t_s = t'_s;\observed{ae'_s\#id_e}{e_e};t''_s$. Using induction hypothesis, the following is true.
    \[
         \exists~p';\assume{z}{e'_s}.~ \SI{}{} p';\assume{z}{e'_s} \Rightarrow_s t'_s;\assume{z}{ae'_s}
    \]
    \[
        \exists~p'';\assume{z}{e''_s}.~ \SI{}{} p'';\assume{z}{e''_s} \Rightarrow_s t''_s;\assume{z}{ae''_s}
    \]

    Given $p = p';\observed{e'_s}{e_e};p'';\assume{z}{e''_s}$, the following is true
    \[
        \SI{}{} p \Rightarrow_s t'_s;\observed{ae'_s\#id_e}{e_e};t''_s;\assume{z}{ae''_s}
    \]

    Hence, given induction 
    hypothesis, if   $ae = \ndi{\mathsf{Dist}(ae' \#id_e) = ae_e }{v}{id}$ and $ id_e \notin \cS$
   then the theorem is true.

    Since I have covered all the bases cases and induction cases, by using induction 
    my induction statement is correct,
      \[
        \begin{array}{c}
        \exists~e.~\SI{}{} e \Rightarrow_s \_,\_ ae \wedge \SAC ae \Rightarrow_{ex} ae_s, t_s
            \implies \exists~p.~ \SI{}{} p \Rightarrow_s t_s;\assume{z}{ae_s}
        \end{array}
    \]
    \end{proof}

\begin{lemma}
    Given environments $\sigma_v$ and $\sigma_{id}$
    such that $ \dom \sigma_v = \dom \sigma_{id}$,
    given a partial trace $t$ of trace $t_p$ ($t$ is the suffix trace for trace $t_p$),
    a valid subproblem $\cS$ on trace $t_p$  and a subtrace $t_s$,   
    \[
        \exists.p~\SI{}{} p \Rightarrow_s t \wedge \SAC t \Rightarrow_{ex} t_s
        \implies \exists.p_s~\SI{}{} p_s \Rightarrow_s t_s
    \]
\end{lemma}
   \begin{proof}
   I will prove this using induction over traces.

    \noindent{\bf Base Case:}
    $t = \emptyset$, $p = \emptyset$ 
    Since
    \[
        \SI{}{} p \Rightarrow_s t
    \]
    and 
    \[
        \SAC \emptyset \Rightarrow_{ex} \emptyset
    \]
    
    Hence if $p_s = \emptyset$ and $t_s = \emptyset$, the
    lemma is true.
     
    \noindent{\bf Induction Case:}
    
    \noindent{\bf Case 1:}$t = \assume{x}{ae};t'$
    
    Assuming from Induction Hypothesis
    \[
        \exists.p'~\SI{'}{'} p' \Rightarrow_s t' \wedge \SAC t' \Rightarrow_{ex} t'_s
        \implies \exists.p'_s~\SI{'}{'} p'_s \Rightarrow_s t'_s
    \]
    and from previous lemma
      \[
    \begin{array}{c}
    \exists~e.~\SI{}{} e \Rightarrow_s v,id, ae \wedge \SAC ae \Rightarrow_{ex} ae_s, t''_s
        \implies \exists~p'';\assume{z}{e_s}.~ \SI{}{} p'';\assume{z}{e_s} \Rightarrow_s t''_s;\assume{z}{ae_s}
    \end{array}
    \]
    where $\sigma'_v = \sigma_v[x \rightarrow v], \sigma'_{id} = \sigma_{id}[x \rightarrow v']$.

    Given
    \[
        \SI{}{} \assume{x}{e};p' \Rightarrow_s \assume{x}{ae};t' \wedge 
        \SAC \assume{x}{ae};t' \Rightarrow_{ex} t''_s;\assume{x}{ae_s};t'_s
    \]

    Hence 
    \[
        \SI{}{} p'';\assume{x}{e_s};p'_s \Rightarrow_s t''_s;\assume{x}{ae_s};t'_s
    \]
    since $t''_s$ only introduces new variables which do not conflict with 
    variables in $\sigma_v$.
    
    \noindent{\bf Case 2:} $t = \observed{ae \# id}{e_v};t'$

    Assuming from Induction Hypothesis
    \[
        \exists.p'~\SI{}{} p' \Rightarrow_s t' \wedge \SAC t' \Rightarrow_{ex} t'_s
        \implies \exists.p'_s~\SI{}{} p'_s \Rightarrow_s t'_s
    \]
    and from previous lemma
      \[
    \begin{array}{c}
    \exists~e.~\SI{}{} e \Rightarrow_s v,id, ae \wedge \SAC ae \Rightarrow_{ex} ae_s, t''_s
        \implies \exists~p'';\assume{z}{e_s}.~ \SI{}{} p'';\assume{z}{e_s} \Rightarrow_s t''_s;\assume{z}{ae_s}
    \end{array}
    \]

    Given
    \[
        \SI{}{} \assume{x}{e};p' \Rightarrow_s \assume{x}{ae};t' \wedge 
        \SAC \assume{x}{ae};t' \Rightarrow_{ex} t''_s;\assume{x}{ae_s};t'_s
    \]

    Hence 
    \[
        \SI{}{} p'';\observed{e_s \# id}{e_v};p'_s \Rightarrow_s t''_s;\observed{ae_s \# id}{e_v};t'_s
    \]
    since $t''_s$ only introduces new variables which do not conflict with 
    variables in $\sigma_v$.

    Hence, since I have covered all the bases cases and induction cases,
    the following statement is true.
    \[
        \exists.p~\SI{}{} p \Rightarrow_s t \wedge \SAC t \Rightarrow_{ex} t_s
        \implies \exists.p_s~\SI{}{} p_s \Rightarrow_s t_s
    \]
    \end{proof}

Hence 
given a valid trace (i.e. $\exists.p~\emptyset, \emptyset \vdash p \Rightarrow_s t$)
and a valid subproblem $\cS$, for all subtraces $t_s$,

\[
    t_s = \mathsf{ExtractTrace}(t, \cS) \implies \exists~p_s.t_s \in \mathsf{Traces}(p_s)
\]

\end{proof}

}
    \noindent{To prove} soundness of our interface we start by proving 
    that for a given trace $t$ and a valid subproblem $\cS$ on trace $t$, 
    if for any subtrace $t_s$ the stitching process succeeds, the output trace 
    $t'$ differs from trace $t$ only in parts which are within the subproblem (i.e. $\SAC t \equiv t'$).

    \noindent{Formally}, for any trace $t$ and subproblem $\cS$,
    \[
        t' = \mathsf{StitchTrace}(t, t_s, \cS) \implies \SAC t \equiv t'
    \]

    \noindent{Using} the definition of $\mathsf{StitchTrace}$,
    the above lemma can be rewritten as 
    \[
        \SAC t, \_ \Rightarrow_{st} t' \implies \SAC t \equiv t'
    \]
    One will note that the stucture of $t_s$ does not play a significant role in 
    proving the above condition.

    \noindent{To prove the above statement we require a similar condition  over 
    augmented expressions embedded within traces. The lemma over augmented expressions is given below:}
   
   \begin{lemma}
    \noindent{For all} augmented expressions $ae, ae'$,
       \[
        \SAC ae, \_, \_ \Rightarrow_{st} ae' \implies \SAC ae \equiv ae'    
        \]
       \label{lem:sound2exp}
    \end{lemma}
    \begin{proof}
    \noindent{Proof using induction.}

    \noindent{\bf Base Case:}

    \noindent{Case 1:} $ae = \ndi{x}{x}{id}$,
   
    \noindent{By assumption}
    \[
        \SAC ae, \_, \_ \Rightarrow_{st} ae'
    \]
        By definition of $\Rightarrow_{st}$
        \[
            \SAC \ndi{x}{x}{id'}, \_, \_ \Rightarrow_{st} \ndi{x}{x}{id'}
        \]
        Then
        $ae' = \ndi{x}{x}{id'}$.

     \noindent{By definition} of $\equiv$
    \[
    \SAC \ndi{x}{x}{id} \equiv \ndi{x}{x}{id'}
    \]
    Therefore
        \[
            \SAC ae \equiv ae'
        \]

        \noindent{Therefore} when $ae = \ndi{x}{x}{id}$,
        \[
            \SAC ae, \_, \_ \Rightarrow_{st} ae' \implies \SAC ae \equiv ae'
        \]

    \noindent{Case 2:} $ae = \ndi{x(id_v)}{v}{id}$

    \noindent{By assumption}
        \[
            \SAC ae, \_, \_ \Rightarrow_{st} ae'
        \]
    By definition of $\Rightarrow_{st}$
        \[
            \SAC \ndi{x(id_v)}{v}{id}, \_, \_ \Rightarrow_{st} \ndi{x(id'_v)}{v'}{id'}
        \]
     Then $ae' = \ndi{x(id'_v)}{v'}{id'}$.

        \noindent{By definition} of $\equiv$
        \[
            \SAC \ndi{x(id_v)}{v}{id} \equiv \ndi{x(id'_v)}{v'}{id'}
        \]
        Therefore
        \[
            \SAC ae \equiv ae'
        \]

        \noindent{Therefore} when $ae = \ndi{x(id_v)}{v}{id}$,
        \[
            \SAC ae, \_, \_ \Rightarrow_{st} ae' \implies \SAC ae \equiv ae'
        \]

    \noindent{\bf Case 3:} $ae = \ndi{\lambda.x~e}{v}{id}$

    \noindent{By assumption}
        \[
            \SAC ae, \_, \_ \Rightarrow_{st} ae'
        \]
    By definition of $\Rightarrow_{st}$
        \[
            \SAC \ndi{\lambda.x~e}{v}{id}, \_, \_ \Rightarrow_{st} \ndi{\lambda.x~e}{v'}{id'}
        \]
     Then $ae' = \ndi{\lambda.x~e}{v'}{id'}$.

        \noindent{By definition} of $\equiv$
        \[
            \SAC \ndi{\lambda.x~e}{v}{id} \equiv \ndi{\lambda.x~e}{v'}{id'}
        \]
        Therefore
        \[
            \SAC ae \equiv ae'
        \]

        \noindent{Therefore} when $ae = \ndi{\lambda.x~e}{v}{id}$,
        \[
            \SAC ae, \_, \_ \Rightarrow_{st} ae' \implies \SAC ae \equiv ae'
        \]

    \noindent{\bf Induction Cases:}

    \noindent{\bf Case 1:}
    $ae = \ndi{(ae_1~ae_2)\perp}{v}{id}$ and $ID(ae_1) \notin \cS$

    \noindent{By assumption}
        \[
            \SAC ae, \_, \_ \Rightarrow_{st} ae'
        \]
    By definition of $\Rightarrow_{st}$
        \[
            \SAC \ndi{(ae_1~ae_2)\perp}{v}{id}, \_, \_ \Rightarrow_{st} \ndi{(ae'_1~ae'_2)\perp}{v'}{id'}
        \]
     Then $ae' = \ndi{(ae_1~ae_2)\perp}{v'}{id'}$, $\SAC ae_1, \_, \_ \Rightarrow_{st} ae'_1$, and
     $\SAC ae_2, \_, \_ \Rightarrow_{st} ae'_2$.

        \noindent{By induction hypothesis}
        \[
            \SAC ae_1, \_, \_ \Rightarrow_{st} ae'_1 \implies \SAC ae_1 \equiv ae'_1
        \]
        Because $\SAC ae_1, \_, \_ \Rightarrow_{st} ae'_1$
        \[
            \SAC ae_1 \equiv ae'_1
        \]

        \noindent{By induction hypothesis}
        \[
            \SAC ae_2, \_, \_ \Rightarrow_{st} ae'_2 \implies \SAC ae_2 \equiv ae'_2
        \]
        Because $\SAC ae_2, \_, \_ \Rightarrow_{st} ae'_2$
        \[
            \SAC ae_2 \equiv ae'_2
        \]

        \noindent{Because} $\SAC ae_1 \equiv ae'_1$ and $\SAC ae_2 \equiv ae'_2$, by definition of $\equiv$
        \[
            \SAC \ndi{(ae_1~ae_2)\perp}{v}{id} \equiv \ndi{(ae'_1~ae'_2)\perp}{v'}{id'}
        \]
        Therefore
        \[
            \SAC ae \equiv ae'
        \]

        \noindent{Therefore} when $ae = \ndi{(ae_1~ae_2)\perp}{v}{id}$,
        \[
            \SAC ae, \_, \_ \Rightarrow_{st} ae' \implies \SAC ae \equiv ae'
        \]

  \noindent{\bf Case 2:}
     $ae = \ndi{(ae_1~ae_2)aa}{v}{id}$ and $ID(ae_1) \in \cS$

    \noindent{By assumption}
        \[
            \SAC ae, \_, \_ \Rightarrow_{st} ae'
        \]
    By definition of $\Rightarrow_{st}$
        \[
            \SAC \ndi{(ae_1~ae_2)aa}{v}{id}, \_, \_ \Rightarrow_{st} \ndi{(ae'_1~ae'_2)aa'}{v'}{id'}
        \]
     Then $ae' = \ndi{(ae_1~ae_2)aa'}{v'}{id'}$, $\SAC ae_1, \_, \_ \Rightarrow_{st} ae'_1$, and
     $\SAC ae_2, \_, \_ \Rightarrow_{st} ae'_2$.

        \noindent{By induction hypothesis}
        \[
            \SAC ae_1, \_, \_ \Rightarrow_{st} ae'_1 \implies \SAC ae_1 \equiv ae'_1
        \]
        Because $\SAC ae_1, \_, \_ \Rightarrow_{st} ae'_1$
        \[
            \SAC ae_1 \equiv ae'_1
        \]

        \noindent{By induction hypothesis}
        \[
            \SAC ae_2, \_, \_ \Rightarrow_{st} ae'_2 \implies \SAC ae_2 \equiv ae'_2
        \]
        Because $\SAC ae_2, \_, \_ \Rightarrow_{st} ae'_2$
        \[
            \SAC ae_2 \equiv ae'_2
        \]

        \noindent{Because} $\SAC ae_1 \equiv ae'_1$ and $\SAC ae_2 \equiv ae'_2$, by definition of $\equiv$
        \[
            \SAC \ndi{(ae_1~ae_2)aa}{v}{id} \equiv \ndi{(ae'_1~ae'_2)aa'}{v'}{id'}
        \]
        Therefore
        \[
            \SAC ae \equiv ae'
        \]

        \noindent{Therefore} when $ae = \ndi{(ae_1~ae_2)aa}{v}{id}$,
        \[
            \SAC ae, \_, \_ \Rightarrow_{st} ae' \implies \SAC ae \equiv ae'
        \]

    \noindent{\bf Case 3:}
     $ae = \ndi{(ae_1~ae_2)x=ae_3}{v}{id}$ and $ID(ae_1) \notin \cS$

    \noindent{By assumption}
        \[
            \SAC ae, \_, \_ \Rightarrow_{st} ae'
        \]
    By definition of $\Rightarrow_{st}$
        \[
            \SAC \ndi{(ae_1~ae_2)x=ae_3}{v}{id}, \_, \_ \Rightarrow_{st} \ndi{(ae'_1~ae'_2)y=ae'_3}{v'}{id'}
        \]
     Then $ae' = \ndi{(ae_1~ae_2)y=ae_3'}{v'}{id'}$, $\SAC ae_1, \_, \_ \Rightarrow_{st} ae'_1$, and
     $\SAC ae_2, \_, \_ \Rightarrow_{st} ae'_2$.

        \noindent{By induction hypothesis}
        \[
            \SAC ae_1, \_, \_ \Rightarrow_{st} ae'_1 \implies \SAC ae_1 \equiv ae'_1
        \]
        Because $\SAC ae_1, \_, \_ \Rightarrow_{st} ae'_1$
        \[
            \SAC ae_1 \equiv ae'_1
        \]

        \noindent{By induction hypothesis}
        \[
            \SAC ae_2, \_, \_ \Rightarrow_{st} ae'_2 \implies \SAC ae_2 \equiv ae'_2
        \]
        Because $\SAC ae_2, \_, \_ \Rightarrow_{st} ae'_2$
        \[
            \SAC ae_2 \equiv ae'_2
        \]

        \noindent{By induction hypothesis}
        \[
            \SAC ae_3, \_, \_ \Rightarrow_{st} ae'_3 \implies \SAC ae_3 \equiv ae'_3
        \]
        Because $\SAC ae_3, \_, \_ \Rightarrow_{st} ae'_3$
        \[
            \SAC ae_3 \equiv ae'_3
        \]

        \noindent{Because} $\SAC ae_1 \equiv ae'_1$, $\SAC ae_2 \equiv ae'_2$ and $\SAC ae_3 \equiv ae'_3$, by definition of $\equiv$
        \[
            \SAC \ndi{(ae_1~ae_2)x=ae_3}{v}{id} \equiv \ndi{(ae'_1~ae'_2)y=ae'_3}{v'}{id'}
        \]
        Therefore
        \[
            \SAC ae \equiv ae'
        \]

        \noindent{Therefore} when $ae = \ndi{(ae_1~ae_2)x=ae_3}{v}{id}$,
        \[
            \SAC ae, \_, \_ \Rightarrow_{st} ae' \implies \SAC ae \equiv ae'
        \]

    \noindent{\bf Case 4:}
    $ae = \ndi{\mathsf{Dist}(ae_1 \# id_e) = ae_2)}{v}{id}$ and $id_e \notin \cS$.

    \noindent{By assumption}
        \[
            \SAC ae, \_, \_ \Rightarrow_{st} ae'
        \]
    By definition of $\Rightarrow_{st}$
        \[
            \SAC \ndi{\mathsf{Dist}(ae_1 \# id_e) = ae_2)}{v}{id}, \_, \_ \Rightarrow_{st} \ndi{\mathsf{Dist}(ae'_1 \# id_e) = ae'_2)}{v'}{id'}
        \]
     Then $ae' = \ndi{\mathsf{Dist}(ae'_1 \# id_e) = ae'_2)}{v'}{id'}$, $\SAC ae_1, \_, \_ \Rightarrow_{st} ae'_1$, and
     $\SAC ae_2, \_, \_ \Rightarrow_{st} ae'_2$.

        \noindent{By induction hypothesis}
        \[
            \SAC ae_1, \_, \_ \Rightarrow_{st} ae'_1 \implies \SAC ae_1 \equiv ae'_1
        \]
        Because $\SAC ae_1, \_, \_ \Rightarrow_{st} ae'_1$
        \[
            \SAC ae_1 \equiv ae'_1
        \]

        \noindent{By induction hypothesis}
        \[
            \SAC ae_2, \_, \_ \Rightarrow_{st} ae'_2 \implies \SAC ae_2 \equiv ae'_2
        \]
        Because $\SAC ae_2, \_, \_ \Rightarrow_{st} ae'_2$
        \[
            \SAC ae_2 \equiv ae'_2
        \]

        \noindent{Because} $\SAC ae_1 \equiv ae'_1$ and $\SAC ae_2 \equiv ae'_2$, by definition of $\equiv$
        \[
            \SAC  \ndi{\mathsf{Dist}(ae_1 \# id_e) = ae_2)}{v}{id} \equiv \ndi{\mathsf{Dist}(ae'_1 \# id_e) = ae'_2)}{v'}{id'}
        \]
        Therefore
        \[
            \SAC ae \equiv ae'
        \]

        \noindent{Therefore} when $ae = \ndi{\mathsf{Dist}(ae_1 \# id_e) = ae_2)}{v}{id}$ and $id_e \notin \cS$,
        \[
            \SAC ae, \_, \_ \Rightarrow_{st} ae'\implies \SAC ae \equiv ae'
        \]

    \noindent{\bf Case 5:}
    $ae = \ndi{\mathsf{Dist}(ae_1 \# id_e) = ae_2)}{v}{id}$ and $id_e \in \cS$.

    \noindent{By assumption}
        \[
            \SAC ae, \_, \_ \Rightarrow_{st} ae'
        \]
    By definition of $\Rightarrow_{st}$
        \[
            \SAC \ndi{\mathsf{Dist}(ae_1 \# id_e) = ae_2)}{v}{id}, \_, \_ \Rightarrow_{st} \ndi{\mathsf{Dist}(ae'_1 \# id_e) = ae'_2)}{v'}{id'}
        \]
     Then $ae' = \ndi{\mathsf{Dist}(ae'_1 \# id_e) = ae'_2)}{v'}{id'}$ and $\SAC ae_1, \_, \_ \Rightarrow_{st} ae'_1$.

        \noindent{By induction hypothesis}
        \[
            \SAC ae_1, \_, \_ \Rightarrow_{st} ae'_1 \implies \SAC ae_1 \equiv ae'_1
        \]
        Because $\SAC ae_1, \_, \_ \Rightarrow_{st} ae'_1$
        \[
            \SAC ae_1 \equiv ae'_1
        \]

        \noindent{Because} $\SAC ae_1 \equiv ae'_1$, by definition of $\equiv$
        \[
            \SAC  \ndi{\mathsf{Dist}(ae_1 \# id_e) = ae_2)}{v}{id} \equiv \ndi{\mathsf{Dist}(ae'_1 \# id_e) = ae'_2)}{v'}{id'}
        \]
        Therefore
        \[
            \SAC ae \equiv ae'
        \]

        \noindent{Therefore} when $ae = \ndi{\mathsf{Dist}(ae_1 \# id_e) = ae_2)}{v}{id}$ and $id_e \in \cS$,
        \[
            \SAC ae, \_, \_ \Rightarrow_{st} ae'\implies \SAC ae \equiv ae'
        \]

    \noindent{Because} all cases are covered, using induction, the following statement is true for 
    all augmented expressions $ae, ae'$ and subproblems $\cS$.
    \[
        \SAC ae, \_, \_ \Rightarrow_{st} ae' \implies \SAC ae \equiv ae'
    \]
\end{proof}

    \noindent{Next we use Lemma~\ref{lem:sound2exp} to prove the lemma below::}

\begin{lemma}
    \noindent{Given} traces $t$ and $t'$  and a subproblem $\cS$,
    \[
        \SAC t, \_ \Rightarrow_{st} t' \implies \SAC t \equiv t'
    \]
\end{lemma}
   \begin{proof}
    \noindent{Proof by Induction}

    \noindent{\bf Base Case:}
    $t = \emptyset$

    \noindent{By assumption}
            \[
                \SAC t, \_ \Rightarrow_{st} t'
            \]
            By definition of $\Rightarrow_{st}$
            \[
                \SAC \emptyset, \_ \Rightarrow_{st} \emptyset
            \]
            Then $t' = \emptyset$.

      \noindent{By definition} of $\equiv$
            \[
                \SAC \emptyset \equiv \emptyset
            \]
            Therefore
            \[
                \SAC t \equiv t'
            \]

    \noindent{Therefore} when $t = \emptyset$
            \[
                \SAC t, \_ \Rightarrow_{st} t' \implies \SAC t \equiv t'
            \]

            \noindent{\bf Induction Case:}

    \noindent{\bf Case 1:}
        $t = t_s;\assume{x}{ae}$
 
    \noindent{By assumption}
            \[
                \SAC t, \_ \Rightarrow_{st} t'
            \]
            By definition of $\Rightarrow_{st}$
            \[
                \SAC t_s;\assume{x}{ae}, \_ \Rightarrow_{st} t'_s;\assume{x}{ae'}
            \]
            Then $t' = t'_s;\assume{x}{ae'}$, $\SAC ae, \_, \_ \Rightarrow_{st} ae'$, and $\SAC t_s, \_ \Rightarrow_{st} t'_s$.

    \noindent{By induction hypothesis}
            \[
                \SAC t_s, \_ \Rightarrow_{st} t'_s \implies \SAC t_s \equiv t'_s
            \]
            Because $\SAC t_s, \_ \Rightarrow_{st} t'_s$
            \[
                \SAC t_s \equiv t'_s
            \]

      \noindent{From Lemma}~\ref{lem:sound2exp}
            \[
                \SAC ae, \_, \_ \Rightarrow_{st} ae' \implies \SAC ae \equiv ae'
            \]
            Because $\SAC ae, \_, \_ \Rightarrow_{st} ae'$
            \[
                \SAC ae \equiv ae'
            \]

      \noindent{Because} $\SAC t'_s \equiv t_s$, and $\SAC ae \equiv ae'$, by definition of $\equiv$
            \[
                \SAC t_s;\assume{x}{ae} \equiv t'_s;\assume{x}{ae'}
            \]
            Therefore
            \[
                \SAC t \equiv t'
            \]

    \noindent{Therefore} when $t = t_s;\assume{x}{ae}$
            \[
                \SAC t, \_ \Rightarrow_{st} t' \implies \SAC t \equiv t'
            \]
    
            \noindent{\bf Case 2:}
        $t = t_s;\observed{ae}{e_v}$
 
    \noindent{By assumption}
            \[
                \SAC t, \_ \Rightarrow_{st} t'
            \]
            By definition of $\Rightarrow_{st}$
            \[
                \SAC t_s;\observed{ae}{e_v}, \_ \Rightarrow_{st} t'_s;\observed{ae'}{e_v}
            \]
            Then $t' = t'_s;\observed{ae'}{e_v}$, $\SAC ae, \_, \_ \Rightarrow_{st} ae' $, and $\SAC t_s, \_ \Rightarrow_{st} t'_s$.

    \noindent{By induction hypothesis}
            \[
                \SAC t_s, \_ \Rightarrow_{st} t'_s \implies \SAC t_s \equiv t'_s
            \]
            Because $\SAC t_s, \_ \Rightarrow_{st} t'_s$
            \[
                \SAC t_s \equiv t'_s
            \]

      \noindent{From Lemma}~\ref{lem:sound2exp}
            \[
                \SAC ae, \_, \_ \Rightarrow_{st} ae' \implies \SAC ae \equiv ae'
            \]
            Because $\SAC ae, \_, \_ \Rightarrow_{st} ae'$
            \[
                \SAC ae \equiv ae'
            \]

      \noindent{Because} $\SAC t'_s \equiv t_s$ and $\SAC ae \equiv ae'$, by definition of $\equiv$
            \[
                \SAC t_s;\observed{ae}{e_v} \equiv t'_s;\observed{ae'}{e_v}
            \]
            Therefore
            \[
                \SAC t \equiv t'
            \]

    \noindent{Therefore} when $t = t_s;\observed{ae'}{e_v}$
            \[
                \SAC t, \_ \Rightarrow_{st} t' \implies \SAC t \equiv t'
            \]

    \noindent{Because} all cases have been covered, using induction, 
    for all traces $t, t'$, subproblems $\cS$  and  subtrace $t_s$, 
    \[\SAC t, t_s \Rightarrow_{st} t' \implies \SAC t \equiv t'\]
    \end{proof}

\begin{corollary}
    \label{lem:sound2}
Given a valid trace $t$, a valid subproblem $\cS$, 
a valid subtrace $t_s$, for all traces $t'$ :
    \[t' = \mathsf{StitchTrace}(t, t_s, \cS) \implies \SAC t \equiv t'\]
\end{corollary}
Therefore, given a trace $t$ and a valid subproblem $\cS$ on $t$, if the 
stitching process succeeds, then the output trace $t'$ will only differ from trace $t$ 
with parts which are within the subproblem $\cS$.

\noindent{Next, we} prove that given a valid trace $t$, a valid subproblem $\cS$ on trace $t$,
    and a subtrace $t_s = \mathsf{ExtractTrace}(t, \cS)$, for any subtrace $t'_s \in \mathsf{Traces}(\mathsf{Program}(t_s))$,
    the stitched trace $t' = \mathsf{StitchTrace}(t, t'_s, \cS)$ is a valid trace from the program of trace $t$ 
    (i.e. $t' \in \mathsf{Traces}(\mathsf{Program}(t))$).

    \noindent{Formally}, given a valid trace $t$, a valid subproblem $\cS$ on $t$, and 
    a subtrace $t_s = \mathsf{ExtractTrace}(t, \cS)$, 
    \[
        \forall~t'_s \in \mathsf{Traces}(\mathsf{Program}(t_s)).~t' = \mathsf{StitchTrace}(t, t'_s, \cS) \wedge t' \in \mathsf{Traces}(\mathsf{Program}(t))
    \]

    \noindent{To prove the above statement, we require a few lemmas first which prove a similar condition for augmented expressions 
    and traces under non-empty environment}

    \begin{lemma}
    \noindent{Given} environements $\sigma_v, \sigma_{id}, \sigma'_v$ and $\sigma'_{id}$
    such that $\dom \sigma_v = \dom \sigma_{id} = \dom \sigma'_v = 
    \dom \sigma'_{id}$ an augmented expression $ae$ within a trace $t$, 
    and a valid subproblem $\cS$ over trace $t$
    \[
        \exists~e, p_s. \SI{}{} e \Rightarrow_s \_, \_, ae
    \wedge \SAC ae \Rightarrow_{ex} ae_s, t_s\]
    \[
    \wedge t_s;\assume{z}{ae_s} \Rightarrow_r p_s
    \wedge \SI{'}{'} p_s \Rightarrow_s  t'_s;\assume{z}{ae'_s}
    \]
        \[
        \implies \exists~ae'.
        \SAC ae, ae'_s, t'_s \Rightarrow_{st} ae', \_
        \wedge \SI{'}{'} e \Rightarrow_s \_,\_, ae' 
        \]    \label{lem:sound3exp}
    \end{lemma}
    \begin{proof}
        \noindent{Proof by induction}

    \noindent{\bf Base Case:}
    
    \noindent{Case 1:}
    $ae = \ndi{x}{x}{id}$
   
    \noindent{By assumption}
    \[
        \SI{}{} e \Rightarrow_s \_, \_, ae
    \]
    By definition of $\Rightarrow_s$
    \[
        \SI{}{} x \Rightarrow_s \_, \_, \ndi{x}{x}{id}
    \]
    Then $e = x$ and $x \notin \dom \sigma_v$.

    \noindent{By assumption}
    \[
        \SAC ae \Rightarrow_{ex} ae_s, t_s
    \]
    By definition of $\Rightarrow_{ex}$
    \[
        \SAC \ndi{x}{x}{id} \Rightarrow_{ex} \ndi{x}{x}{id}, \emptyset
    \]
    Then $ae_s = \ndi{x}{x}{id}$ and $t_s = \emptyset$.

    \noindent{By assumption}
    \[
        t_s;\assume{z}{ae_s} \Rightarrow_r p_s
    \]
    By definition of $\Rightarrow_r$
    \[
        \assume{z}{\ndi{x}{x}{id}} \Rightarrow_r \assume{z}{x}
    \]
    Then $p_s = \assume{z}{x}$.

    \noindent{By assumption}
    \[
        \SI{'}{'} p_s \Rightarrow_s t'_s;\assume{z}{ae'_s}
    \]
    By definition of $\Rightarrow_s$, $\dom \sigma'_v = \dom \sigma_v$
    \[
        \SI{'}{'} \assume{z}{x} \Rightarrow_s \assume{z}{\ndi{x}{x}{id'}}
    \]
    Then $t'_s = \emptyset$ and $ae'_s = \ndi{x}{x}{id'}$.

    \noindent{Consider} $ae' = \ndi{x}{x}{id'}$.

    \noindent{By definition} of $\Rightarrow_{st}$
    \[
        \SAC \ndi{x}{x}{id}, \ndi{x}{x}{id'}, \emptyset \Rightarrow_{st} \ndi{x}{x}{id'}
    \]
    Therefore
    \[
        \SAC ae, ae'_s, t'_s \Rightarrow_{st} ae'
    \]

    \noindent{By definition} of $\Rightarrow_s$ and $x \in \dom \sigma'_v$
    \[
        \SI{'}{'} x \Rightarrow_s \_, \_, \ndi{x}{x}{id'} 
    \]
    Therefore
    \[
        \SI{'}{'} e \Rightarrow_s \_, \_, ae'
    \]
    
    \noindent{Therefore} when $ae = \ndi{x}{x}{id}$
    \[
        \exists~e, p_s. \SI{}{} e \Rightarrow_s \_, \_, ae
    \wedge \SAC ae \Rightarrow_{ex} ae_s, t_s\]
    \[
    \wedge t_s;\assume{z}{ae_s} \Rightarrow_r p_s
    \wedge \SI{'}{'} p_s \Rightarrow_s  t'_s;\assume{z}{ae'_s}
    \]
    \[
        \implies \exists~ae'.
        \SAC ae, ae'_s, t'_s \Rightarrow_{st} ae', \_
        \wedge \SI{'}{'} e \Rightarrow_s \_,\_, ae'
    \]

    \noindent{Case 2:}
    $ae = \ndi{x(id_v)}{v}{id}$
   
    \noindent{By assumption}
    \[
        \SI{}{} e \Rightarrow_s \_, \_, ae
    \]
    By definition of $\Rightarrow_s$
    \[
        \SI{}{} x \Rightarrow_s \_, \_, \ndi{x(id_v)}{v}{id}
    \]
    Then $e = x$, $v = \sigma_v(x)$, and $id_v = \sigma_{id}(x)$.

    \noindent{By assumption}
    \[
        \SAC ae \Rightarrow_{ex} ae_s, t_s
    \]
    By definition of $\Rightarrow_{ex}$
    \[
        \SAC \ndi{x(id_v)}{v}{id} \Rightarrow_{ex} \ndi{x(id_v)}{v}{id}, \emptyset
    \]
    Then $ae_s = \ndi{x(id_v)}{v}{id}$ and $t_s = \emptyset$.

    \noindent{By assumption}
    \[
        t_s;\assume{z}{ae_s} \Rightarrow_r p_s
    \]
    By definition of $\Rightarrow_r$
    \[
        \assume{z}{\ndi{x(id_v)}{v}{id}} \Rightarrow_r \assume{z}{x}
    \]
    Then $p_s = \assume{z}{x}$.

    \noindent{By assumption}
    \[
        \SI{'}{'} p_s \Rightarrow_s t'_s;\assume{z}{ae'_s}
    \]
    By definition of $\Rightarrow_s$, $\dom \sigma'_v = \dom \sigma_v$
    \[
        \SI{'}{'} \assume{z}{x} \Rightarrow_s \assume{z}{\ndi{x(id'_v)}{v'}{id'}}
    \]
    Then $t'_s = \emptyset$, $ae'_s = \ndi{x(id'_v)}{v'}{id'}$, $v' = \sigma'_v(x)$, and $id'_v = \sigma'_{id}(x)$.

    \noindent{Consider} $ae' = \ndi{x(id'_v)}{v'}{id'}$.

    \noindent{By definition} of $\Rightarrow_{st}$
    \[
        \SAC \ndi{x(id_v)}{v}{id}, \ndi{x(id'_v)}{v'}{id'}, \emptyset \Rightarrow_{st} \ndi{x(id'_v)}{v'}{id'}
    \]
    Therefore
    \[
        \SAC ae, ae'_s, t'_s \Rightarrow_{st} ae'
    \]

    \noindent{By definition} of $\Rightarrow_s$,$ v' = \sigma'_v(x)$, and $id'_v = \sigma'_{id}(x)$
    \[
        \SI{'}{'} x \Rightarrow_s \_, \_, \ndi{x(id'_v)}{v'}{id'} 
    \]
    Therefore
    \[
        \SI{'}{'} e \Rightarrow_s \_, \_, ae'
    \]
    
    \noindent{Therefore} when $ae = \ndi{x(id_v)}{v}{id}$
    \[
        \exists~e, p_s. \SI{}{} e \Rightarrow_s \_, \_, ae
    \wedge \SAC ae \Rightarrow_{ex} ae_s, t_s\]
    \[
    \wedge t_s;\assume{z}{ae_s} \Rightarrow_r p_s
    \wedge \SI{'}{'} p_s \Rightarrow_s  t'_s;\assume{z}{ae'_s}
    \]
    \[
        \implies \exists~ae'.
        \SAC ae, ae'_s, t'_s \Rightarrow_{st} ae', \_
        \wedge \SI{'}{'} e \Rightarrow_s \_,\_, ae'
    \]

    \noindent{Case 3:}
    $ae = \ndi{\lambda.x~e'}{v}{id}$
    
    \noindent{By assumption}
    \[
        \SI{}{} e \Rightarrow_s \_, \_, ae
    \]
    By definition of $\Rightarrow_s$
    \[
        \SI{}{} \lambda.x~e' \Rightarrow_s \_, \_, \ndi{\lambda.x~e'}{v}{id}
    \]
    Then $e = \lambda.x~e'$.

    \noindent{By assumption}
    \[
        \SAC ae \Rightarrow_{ex} ae_s, t_s
    \]
    By definition of $\Rightarrow_{ex}$
    \[
        \SAC \ndi{\lambda.x~e'}{v}{id} \Rightarrow_{ex} \ndi{\lambda.x~e'}{v}{id}, \emptyset
    \]
    Then $ae_s = \ndi{\lambda.x~e'}{v}{id}$ and $t_s = \emptyset$.

    \noindent{By assumption}
    \[
        t_s;\assume{z}{ae_s} \Rightarrow_r p_s
    \]
    By definition of $\Rightarrow_r$
    \[
        \assume{z}{\ndi{\lambda.x~e'}{v}{id}} \Rightarrow_r \assume{z}{\lambda.x~e'}
    \]
    Then $p_s = \assume{z}{\lambda.x~e'}$.

    \noindent{By assumption}
    \[
        \SI{'}{'} p_s \Rightarrow_s t'_s;\assume{z}{ae'_s}
    \]
    By definition of $\Rightarrow_s$, $\dom \sigma'_v = \dom \sigma_v$
    \[
        \SI{'}{'} \assume{z}{\lambda.x~e'} \Rightarrow_s \assume{z}{\ndi{\lambda.x~e'}{v'}{id'}}
    \]
    Then $t'_s = \emptyset$ and $ae'_s = \ndi{\lambda.x~e'}{v'}{id'}$.

    \noindent{Consider} $ae' = \ndi{\lambda.x~e'}{v'}{id'}$.

    \noindent{By definition} of $\Rightarrow_{st}$
    \[
        \SAC \ndi{\lambda.x~e'}{v}{id}, \ndi{\lambda.x~e'}{v'}{id'}, \emptyset \Rightarrow_{st}\ndi{\lambda.x~e'}{v'}{id'}
    \]
    Therefore
    \[
        \SAC ae, ae'_s, t'_s \Rightarrow_{st} ae'
    \]

    \noindent{By definition} of $\Rightarrow_s$ and $\SI{'}{'} \lambda.x~e' \Rightarrow_s \_, \_, \ndi{\lambda.x~e'}{v'}{id'}$
    \[
        \SI{'}{'} \lambda.x~e' \Rightarrow_s \_, \_, \ndi{\lambda.x~e'}{v'}{id'} 
    \]
    Therefore
    \[
        \SI{'}{'} e \Rightarrow_s \_, \_, ae'
    \]
    
    \noindent{Therefore} when $ae = \ndi{\lambda.x~e'}{v}{id}$
    \[
        \exists~e, p_s. \SI{}{} e \Rightarrow_s \_, \_, ae
    \wedge \SAC ae \Rightarrow_{ex} ae_s, t_s\]
    \[
    \wedge t_s;\assume{z}{ae_s} \Rightarrow_r p_s
    \wedge \SI{'}{'} p_s \Rightarrow_s  t'_s;\assume{z}{ae'_s}
    \]
    \[
        \implies \exists~ae'.
        \SAC ae, ae'_s, t'_s \Rightarrow_{st} ae', \_
        \wedge \SI{'}{'} e \Rightarrow_s \_,\_, ae'
    \]

    \noindent{\bf Induction Cases:}
    
    \noindent{Case 1:} 
    $ae = \ndi{\mathsf{Dist}(ae_1 \#id_e)=ae_v}{v}{id}$ and $id_e \in \cS$

    \noindent{By assumption}
    \[
        \SI{}{} e \Rightarrow_s \_, \_, ae
    \]
    By definition of $\Rightarrow_s$
    \[
        \SI{}{} \mathsf{Dist}(e_1) \Rightarrow_s \_, \_, \ndi{\mathsf{Dist}(ae_1 \#id_e)=ae_v}{v}{id}
    \]
    Then $e = \mathsf{Dist}(e_1)$, $\SI{}{} e_1 \Rightarrow_s \_, \_, ae_1$.

    \noindent{By assumption}
    \[
        \SAC ae \Rightarrow_{ex} ae_s, t_s
    \]
    By definition of $\Rightarrow_{ex}$
    \[
        \SAC \ndi{\mathsf{Dist}(ae_1 \#id_e)=ae_v}{v}{id} \Rightarrow_{ex} \ndi{\mathsf{Dist}(ae^1_s \#id_e)=ae_v}{v}{id}, t^1_s
    \]
    Then $ae_s = \ndi{\mathsf{Dist}(ae_1 \#id_e)=ae_v}{v}{id}$, $t_s = t^1_s$, 
    and $\SAC ae_1 \Rightarrow_{ex} ae^1_s, t^1_s$.

    \noindent{By assumption}
    \[
        t_s;\assume{z}{ae_s} \Rightarrow_r p_s
    \]
    By definition of $\Rightarrow_r$
    \[
        \assume{z}{\ndi{\mathsf{Dist}(ae^1_s \#id_e)=ae_v}{v}{id}} \Rightarrow_r \assume{z}{\mathsf{Dist}(e^1_s)}
    \]
    Then $p_s = p^1_s;\assume{z}{\mathsf{Dist}(e^1_s)}$ and $ae^1_s \Rightarrow_r e^1_s$.

    \noindent{By assumption}
    \[
        \SI{'}{'} p_s \Rightarrow_s t'_s;\assume{z}{ae'_s}
    \]
    By definition of $\Rightarrow_s$, 
    \[
        \SI{'}{'} p^1_s;\assume{z}{\mathsf{Dist}(e_1) } \Rightarrow_s t^2_s;\assume{z}{\ndi{\mathsf{Dist}(ae^2_s \#id'_e)=ae'_v}{v'}{id'}}
    \]
    Then $t'_s = t^2_s$, $ae'_s = \ndi{\mathsf{Dist}(ae^2_s \#id_e)=ae'_v}{v'}{id'}$, and $\SI{'}{'} p^1_s;\assume{z}{e_1} \Rightarrow_s t^2_s;\assume{z}{ae^2_s}$.

   \noindent{By induction hypothesis}
    \[
        \exists~e_1, p^1_s, e^1_s. \SI{}{} e_1 \Rightarrow_s \_, \_, ae_1
    \wedge \SAC ae_1 \Rightarrow_{ex} ae^1_s, t^1_s\]
    \[
        \wedge t^1_s;\assume{z}{ae^1_s} \Rightarrow_r p^1_s;\assume{z}{e^1_s}
    \]
            \[
    \wedge \SI{'}{'} p^1_s;\assume{z}{e^1_s} \Rightarrow_s  t^2_s;\assume{z}{ae^2_s}
    \]
    \[
        \implies \exists~ae'_1.
        \SAC ae_1, ae^2_s, t^2_s \Rightarrow_{st} ae'_1
        \wedge \SI{'}{'} e_1 \Rightarrow_s \_,\_, ae'_1
    \]
    Because $\SI{}{} e_1 \Rightarrow_s \_, \_, ae_1$, $\SAC ae_1 \Rightarrow_{ex} ae^1_s, t^1_s$, 
    $t^1_s;\assume{z}{ae^1_s} \Rightarrow_r p^1_s;\assume{z}{e^1_s}$, and 
    $\SI{'}{'} p^1_s;\assume{z}{e^1_s} \Rightarrow_s t^2_s;\assume{z}{ae^2_s}$.
    \[
\SAC ae_1, ae^2_s, t^2_s \Rightarrow_{st} ae'_1
        \wedge \SI{'}{'} e_1 \Rightarrow_s \_,\_, ae'_1
    \]

    \noindent{Consider} $ae' = \ndi{\mathsf{Dist}(ae'_1 \#id'_e)=ae'_v}{v'}{id'}$.

    \noindent{By definition} of $\Rightarrow_{st}$, $\SAC ae_1, ae^2_s, t^2_s \Rightarrow_{st} ae'_1$
    \[
        \begin{array}{c}
        \SAC \ndi{\mathsf{Dist}(ae^1_s \#id_e)=ae_v}{v}{id}, 
        \ndi{\mathsf{Dist}(ae^2_s \#id'_e)=ae'_v}{v'}{id'}, t'_s \\ \Rightarrow_{st} \ndi{\mathsf{Dist}(ae'_1 \#id'_e)=ae'_v}{v'}{id'}
        \end{array}
    \]
    Therefore
    \[
        \SAC ae, ae'_s, t'_s \Rightarrow_{st} ae'
    \]

    \noindent{By definition} of $\Rightarrow_s$ and $\SI{'}{'} e_1 \Rightarrow_s \_, \_, ae_1$,
    \[
        \SI{'}{'} e \Rightarrow_s \_, \_,  \ndi{\mathsf{Dist}(ae'_1 \#id'_e)=ae'_v}{v'}{id'} 
    \]
    Therefore
    \[
        \SI{'}{'} e \Rightarrow_s \_, \_, ae'
    \]
    
    \noindent{Therefore} when $ae = \ndi{\mathsf{Dist}(ae_1 \#id_e)=ae_v}{v}{id} $
    \[
        \exists~e, p_s. \SI{}{} e \Rightarrow_s \_, \_, ae
    \wedge \SAC ae \Rightarrow_{ex} ae_s, t_s\]
    \[
    \wedge t_s;\assume{z}{ae_s} \Rightarrow_r p_s
    \wedge \SI{'}{'} p_s \Rightarrow_s  t'_s;\assume{z}{ae'_s}
    \]
    \[
        \implies \exists~ae'.
        \SAC ae, ae'_s, t'_s \Rightarrow_{st} ae', \_
        \wedge \SI{'}{'} e \Rightarrow_s \_,\_, ae'
    \]

    \noindent{Case 2:} 
    $ae = \ndi{\mathsf{Dist}(ae_1 \#id_e)=ae_2}{v}{id}$ and $id_e \notin \cS$

    \noindent{By assumption}
    \[
        \SI{}{} e \Rightarrow_s \_, \_, ae
    \]
    By definition of $\Rightarrow_s$
    \[
        \SI{}{} \mathsf{Dist}(e_1) \Rightarrow_s \_, \_, \ndi{\mathsf{Dist}(ae_1 \#id_e)=ae_2}{v}{id}
    \]
    Then $e = \mathsf{Dist}(e_1)$, $\SI{}{} e_1 \Rightarrow_s \_, \_, ae_1$, $ae_2 \Rightarrow_r e_2$
    and $\SI{}{} e_2 \Rightarrow_s \_, \_, ae_2$.

    \noindent{By assumption}
    \[
        \SAC ae \Rightarrow_{ex} ae_s, t_s
    \]
    By definition of $\Rightarrow_{ex}$
    \[
        \SAC \ndi{\mathsf{Dist}(ae_1 \#id_e)=ae_2}{v}{id} \Rightarrow_{ex} ae^3_s, t^1_s;\observed{ae^1_s}{e_2};t^3_s
    \]
    Then $ae_s = ae^3_s$, $t_s = t^1_s;\observed{ae^1_s}{e_2};t^3_s$, 
    $\SAC ae_1 \Rightarrow_{ex} ae^1_s, t^1_s$, and $\SAC ae_2 \Rightarrow_{ex} ae^3_s, t^3_s$.

    \noindent{By assumption}
    \[
        t_s;\assume{z}{ae_s} \Rightarrow_r p_s
    \]
    By definition of $\Rightarrow_r$
    \[
        \begin{array}{c}
        t^1_s;\observed{ae^1_s}{e_2};t^3_s;\assume{z}{ae^3_s} \Rightarrow_r \\ p^1_s;\observed{e^1_s}{e_2};p^2_s\assume{z}{\mathsf{Dist}(e^1_s)}
        \end{array}
    \]
    Then $p_s = p^1_s;\observed{e^1_s}{e_2};p^2_s;\assume{z}{e^2_s}$, $ae^1_s \Rightarrow_r e^1_s$, 
    $t^1_s \Rightarrow_r p^1_s$, $ae^3_s \Rightarrow_r e^2_s$, and $t^3_s \Rightarrow_r p^2_s$.

    \noindent{By assumption}
    \[
        \SI{'}{'} p_s \Rightarrow_s t'_s;\assume{z}{ae'_s}
    \]
    By definition of $\Rightarrow_s$, 
    \[
        \begin{array}{c}
        \SI{'}{'} p^1_s;\observed{e^1_s}{e_2};p^2_s;\assume{z}{e^2_s} \Rightarrow_s \\ t^2_s;\observed{ae^2_s}{e_2};t^4_s;
        \assume{z}{ae^4_s}
        \end{array}
    \]
    Then $t'_s = t^2_s;\observed{ae^2_s}{e_2};t^4_s$, $ae'_s = ae^4_s$, $\SI{'}{'} p^1_s;\assume{z}{e^1_s} \Rightarrow_s t^2_s;\assume{z}{ae^2_s}$,
    and $\SI{'}{'} p^2_s;\assume{z}{e^2_s} \Rightarrow_s t^4_s;\assume{z}{ae^4_s}$.

   \noindent{By induction hypothesis}
    \[
        \exists~e_1, p^1_s, e^1_s. \SI{}{} e_1 \Rightarrow_s \_, \_, ae_1
    \wedge \SAC ae_1 \Rightarrow_{ex} ae^1_s, t^1_s\]
    \[
        \wedge t^1_s;\assume{z}{ae^1_s} \Rightarrow_r p^1_s;\assume{z}{e^1_s}
    \]
            \[
    \wedge \SI{'}{'} p^1_s;\assume{z}{e^1_s} \Rightarrow_s  t^2_s;\assume{z}{ae^2_s}
    \]
    \[
        \implies \exists~ae'_1.
        \SAC ae_1, ae^2_s, t^2_s \Rightarrow_{st} ae'_1
        \wedge \SI{'}{'} e_1 \Rightarrow_s \_,\_, ae'_1
    \]
    Because $\SI{}{} e_1 \Rightarrow_s \_, \_, ae_1$, $\SAC ae_1 \Rightarrow_{ex} ae^1_s, t^1_s$, 
    $t^1_s;\assume{z}{ae^1_s} \Rightarrow_r p^1_s;\assume{z}{e^1_s}$, and 
    $\SI{'}{'} p^1_s;\assume{z}{e^1_s} \Rightarrow_s t^2_s;\assume{z}{ae^2_s}$.
    \[
\SAC ae_1, ae^2_s, t^2_s \Rightarrow_{st} ae'_1
        \wedge \SI{'}{'} e_1 \Rightarrow_s \_,\_, ae'_1
    \]

   \noindent{By induction hypothesis}
    \[
        \exists~e_2, p^2_s, e^2_s. \SI{}{} e_2 \Rightarrow_s \_, \_, ae_2
    \wedge \SAC ae_2 \Rightarrow_{ex} ae^3_s, t^3_s\]
    \[
        \wedge t^3_s;\assume{z}{ae^3_s} \Rightarrow_r p^2_s;\assume{z}{e^2_s}
    \]
            \[
    \wedge \SI{'}{'} p^2_s;\assume{z}{e^2_s} \Rightarrow_s  t^4_s;\assume{z}{ae^4_s}
    \]
    \[
        \implies \exists~ae'_2.
        \SAC ae_2, ae^4_s, t^4_s \Rightarrow_{st} ae'_2
        \wedge \SI{'}{'} e_2 \Rightarrow_s \_,\_, ae'_2
    \]
    Because $\SI{}{} e_2 \Rightarrow_s \_, \_, ae_2$, $\SAC ae_2 \Rightarrow_{ex} ae^3_s, t^3_s$, 
    $t^3_s;\assume{z}{ae^3_s} \Rightarrow_r p^2_s;\assume{z}{e^2_s}$, and 
    $\SI{'}{'} p^2_s;\assume{z}{e^2_s} \Rightarrow_s t^4_s;\assume{z}{ae^4_s}$.
    \[
\SAC ae_2, ae^4_s, t^4_s \Rightarrow_{st} ae'_2
        \wedge \SI{'}{'} e_2 \Rightarrow_s \_,\_, ae'_2
    \]

    \noindent{Consider} $ae' = \ndi{\mathsf{Dist}(ae'_1 \#id_e)=ae'_2}{v'}{id'}$.

    \noindent{By definition} of $\Rightarrow_{st}$, $\SAC ae_1, ae^2_s, t^2_s \Rightarrow_{st} ae'_1$, and
    $\SAC ae_2, ae^4_s, t^4_s \Rightarrow_{st} ae'_2$,
    \[
        \begin{array}{c}
        \SAC \ndi{\mathsf{Dist}(ae_1 \#id_e)=ae_2}{v}{id}, 
        ae^4_s, t'_s \\ \Rightarrow_{st} \ndi{\mathsf{Dist}(ae'_1 \#id'_e)=ae'_2}{v'}{id'}
        \end{array}
    \]
    Therefore
    \[
        \SAC ae, ae'_s, t'_s \Rightarrow_{st} ae'
    \]

    \noindent{By definition} of $\Rightarrow_s$, $\SI{'}{'} e_1 \Rightarrow_s \_, \_, ae'_1$, $\SI{'}{'} e_2 \Rightarrow_s, \_, \_, ae'_2$,
    \[
        \SI{'}{'} e \Rightarrow_s \_, \_,  \ndi{\mathsf{Dist}(ae'_1 \#id_e)=ae'_2}{v'}{id'} 
    \]
    Therefore
    \[
        \SI{'}{'} e \Rightarrow_s \_, \_, ae'
    \]
    
    \noindent{Therefore} when $ae = \ndi{\mathsf{Dist}(ae_1 \#id_e)=ae_2}{v}{id} $
    \[
        \exists~e, p_s. \SI{}{} e \Rightarrow_s \_, \_, ae
    \wedge \SAC ae \Rightarrow_{ex} ae_s, t_s\]
    \[
    \wedge t_s;\assume{z}{ae_s} \Rightarrow_r p_s
    \wedge \SI{'}{'} p_s \Rightarrow_s  t'_s;\assume{z}{ae'_s}
    \]
    \[
        \implies \exists~ae'.
        \SAC ae, ae'_s, t'_s \Rightarrow_{st} ae', \_
        \wedge \SI{'}{'} e \Rightarrow_s \_,\_, ae'
    \]

    \noindent{Case 3:} 
    $ae = \ndi{(ae_1~ae_2)\perp}{v}{id}$ and $ID(ae_1) \notin \cS$

    \noindent{By assumption}
    \[
        \SI{}{} e \Rightarrow_s \_, \_, ae
    \]
    By definition of $\Rightarrow_s$
    \[
        \SI{}{} (e_1~e_2) \Rightarrow_s \_, \_, \ndi{(ae_1~ae_2)\perp}{v}{id}
    \]
    Then $e = (e_1~e_2)$, $\SI{}{} e_1 \Rightarrow_s \_, \_, ae_1$, 
    and $\SI{}{} e_2 \Rightarrow_s \_, \_, ae_2$.

    \noindent{By assumption}
    \[
        \SAC ae \Rightarrow_{ex} ae_s, t_s
    \]
    By definition of $\Rightarrow_{ex}$
    \[
        \SAC \ndi{(ae_1~ae_2)\perp}{v}{id} \Rightarrow_{ex} \ndi{(ae^1_s~ae^3_s)\perp}{v}{id}, t^1_s;t^3_s
    \]
    Then $ae_s = \ndi{(ae^1_s~ae^3_s)\perp}{v}{id}$, $t_s = t^1_s;t^3_s$, 
    $\SAC ae_1 \Rightarrow_{ex} ae^1_s, t^1_s$, and $\SAC ae_2 \Rightarrow_{ex} ae^3_s, t^3_s$.

    \noindent{By assumption}
    \[
        t_s;\assume{z}{ae_s} \Rightarrow_r p_s
    \]
    By definition of $\Rightarrow_r$
    \[
        t^1_s;t^3_s;\assume{z}{\ndi{(ae^1_s~ae^3_s)\perp}{v}{id}} \Rightarrow_r p^1_s;p^2_s\assume{z}{(e^1_s~e^2_s)}
    \]
    Then $p_s = p^1_s;p^2_s;\assume{z}{(e^1_s~e^2_s)}$, $ae^1_s \Rightarrow_r e^1_s$, 
    $t^1_s \Rightarrow_r p^1_s$, $ae^3_s \Rightarrow_r e^2_s$, and $t^3_s \Rightarrow_r p^2_s$.

    \noindent{By assumption}
    \[
        \SI{'}{'} p_s \Rightarrow_s t'_s;\assume{z}{ae'_s}
    \]
    By definition of $\Rightarrow_s$, 
    \[
        \SI{'}{'} p^1_s;p^2_s;\assume{z}{(e^1_s~e^2_s)} \Rightarrow_s t^2_s;t^4_s;
        \assume{z}{\ndi{(ae^2_s~ae^4_s)\perp}{v}{id}}
    \]
    Then $t'_s = t^2_s;t^4_s$, $ae'_s = \ndi{(ae^2_s~ae^4_s)\perp}{v}{id}$, $\SI{'}{'} p^1_s;\assume{z}{e^1_s} \Rightarrow_s t^2_s;\assume{z}{ae^2_s}$,
    and $\SI{'}{'} p^2_s;\assume{z}{e^2_s} \Rightarrow_s t^4_s;\assume{z}{ae^4_s}$.

   \noindent{By induction hypothesis}
    \[
        \exists~e_1, p^1_s, e^1_s. \SI{}{} e_1 \Rightarrow_s \_, \_, ae_1
    \wedge \SAC ae_1 \Rightarrow_{ex} ae^1_s, t^1_s\]
    \[
        \wedge t^1_s;\assume{z}{ae^1_s} \Rightarrow_r p^1_s;\assume{z}{e^1_s}
    \]
            \[
    \wedge \SI{'}{'} p^1_s;\assume{z}{e^1_s} \Rightarrow_s  t^2_s;\assume{z}{ae^2_s}
    \]
    \[
        \implies \exists~ae'_1.
        \SAC ae_1, ae^2_s, t^2_s \Rightarrow_{st} ae'_1
        \wedge \SI{'}{'} e_1 \Rightarrow_s \_,\_, ae'_1
    \]
    Because $\SI{}{} e_1 \Rightarrow_s \_, \_, ae_1$, $\SAC ae_1 \Rightarrow_{ex} ae^1_s, t^1_s$, 
    $t^1_s;\assume{z}{ae^1_s} \Rightarrow_r p^1_s;\assume{z}{e^1_s}$, and 
    $\SI{'}{'} p^1_s;\assume{z}{e^1_s} \Rightarrow_s t^2_s;\assume{z}{ae^2_s}$.
    \[
\SAC ae_1, ae^2_s, t^2_s \Rightarrow_{st} ae'_1
        \wedge \SI{'}{'} e_1 \Rightarrow_s \_,\_, ae'_1
    \]

   \noindent{By induction hypothesis}
    \[
        \exists~e_2, p^2_s, e^2_s. \SI{}{} e_2 \Rightarrow_s \_, \_, ae_2
    \wedge \SAC ae_2 \Rightarrow_{ex} ae^3_s, t^3_s\]
    \[
        \wedge t^3_s;\assume{z}{ae^3_s} \Rightarrow_r p^2_s;\assume{z}{e^2_s}
    \]
            \[
    \wedge \SI{'}{'} p^2_s;\assume{z}{e^2_s} \Rightarrow_s  t^4_s;\assume{z}{ae^4_s}
    \]
    \[
        \implies \exists~ae'_2.
        \SAC ae_2, ae^4_s, t^4_s \Rightarrow_{st} ae'_2
        \wedge \SI{'}{'} e_2 \Rightarrow_s \_,\_, ae'_2
    \]
    Because $\SI{}{} e_2 \Rightarrow_s \_, \_, ae_2$, $\SAC ae_2 \Rightarrow_{ex} ae^3_s, t^3_s$, 
    $t^3_s;\assume{z}{ae^3_s} \Rightarrow_r p^2_s;\assume{z}{e^2_s}$, and 
    $\SI{'}{'} p^2_s;\assume{z}{e^2_s} \Rightarrow_s t^4_s;\assume{z}{ae^4_s}$.
    \[
\SAC ae_2, ae^4_s, t^4_s \Rightarrow_{st} ae'_2
        \wedge \SI{'}{'} e_2 \Rightarrow_s \_,\_, ae'_2
    \]

    \noindent{Consider} $ae' = \ndi{(ae'_1~ae'_2)\perp}{v'}{id'}$.

    \noindent{By definition} of $\Rightarrow_{st}$, $\SAC ae_1, ae^2_s, t^2_s \Rightarrow_{st} ae'_1$, and
    $\SAC ae_2, ae^4_s, t^4_s \Rightarrow_{st} ae'_2$,
    \[
        \begin{array}{c}
        \SAC \ndi{(ae_1~ae_2)\perp}{v}{id}, \ndi{(ae^2_s~ae^4_s)\perp}{v'}{id'}, 
            t'_s \\ \Rightarrow_{st} \ndi{(ae'_1~ae'_2)\perp}{v'}{id'}
        \end{array}
    \]
    Therefore
    \[
        \SAC ae, ae'_s, t'_s \Rightarrow_{st} ae'
    \]

    \noindent{By definition} of $\Rightarrow_s$, $\SI{'}{'} e_1 \Rightarrow_s \_, \_, ae'_1$, $\SI{'}{'} e_2 \Rightarrow_s, \_, \_, ae'_2$,
    and because $ae_1$ is not in the subproblem $\cS$, therefore its value will not change
    \[
        \SI{'}{'} e \Rightarrow_s \_, \_,  \ndi{(ae'_1~ae'_2)\perp}{v'}{id'}
    \]
    Therefore
    \[
        \SI{'}{'} e \Rightarrow_s \_, \_, ae'
    \]
    
    \noindent{Therefore} when $ae = \ndi{(ae_1~ae_2)\perp}{v}{id}$, $ID(ae_1) \notin \cS$
    \[
        \exists~e, p_s. \SI{}{} e \Rightarrow_s \_, \_, ae
    \wedge \SAC ae \Rightarrow_{ex} ae_s, t_s\]
    \[
    \wedge t_s;\assume{z}{ae_s} \Rightarrow_r p_s
    \wedge \SI{'}{'} p_s \Rightarrow_s  t'_s;\assume{z}{ae'_s}
    \]
    \[
        \implies \exists~ae'.
        \SAC ae, ae'_s, t'_s \Rightarrow_{st} ae', \_
        \wedge \SI{'}{'} e \Rightarrow_s \_,\_, ae'
    \]

    \noindent{Case 4:} 
    $ae = \ndi{(ae_1~ae_2)aa}{v}{id}$ and $ID(ae_1) \in \cS$
 
    \noindent{By assumption}
    \[
        \SI{}{} e \Rightarrow_s \_, \_, ae
    \]
    By definition of $\Rightarrow_s$
    \[
        \SI{}{} (e_1~e_2) \Rightarrow_s \_, \_, \ndi{(ae_1~ae_2)aa}{v}{id}
    \]
    Then $e = (e_1~e_2)$, $\SI{}{} e_1 \Rightarrow_s \_, \_, ae_1$, 
    and $\SI{}{} e_2 \Rightarrow_s \_, \_, ae_2$.

    \noindent{By assumption}
    \[
        \SAC ae \Rightarrow_{ex} ae_s, t_s
    \]
    By definition of $\Rightarrow_{ex}$
    \[
        \SAC \ndi{(ae_1~ae_2)aa}{v}{id} \Rightarrow_{ex} \ndi{(ae^1_s~ae^3_s)aa}{v}{id}, t^1_s;t^3_s
    \]
    Then $ae_s = \ndi{(ae^1_s~ae^3_s)aa}{v}{id}$, $t_s = t^1_s;t^3_s$, 
    $\SAC ae_1 \Rightarrow_{ex} ae^1_s, t^1_s$, and $\SAC ae_2 \Rightarrow_{ex} ae^3_s, t^3_s$.

    \noindent{By assumption}
    \[
        t_s;\assume{z}{ae_s} \Rightarrow_r p_s
    \]
    By definition of $\Rightarrow_r$
    \[
        t^1_s;t^3_s;\assume{z}{\ndi{(ae^1_s~ae^3_s)aa}{v}{id}} \Rightarrow_r p^1_s;p^2_s\assume{z}{(e^1_s~e^2_s)}
    \]
    Then $p_s = p^1_s;p^2_s;\assume{z}{(e^1_s~e^2_s)}$, $ae^1_s \Rightarrow_r e^1_s$, 
    $t^1_s \Rightarrow_r p^1_s$, $ae^3_s \Rightarrow_r e^2_s$, and $t^3_s \Rightarrow_r p^2_s$.

    \noindent{By assumption}
    \[
        \SI{'}{'} p_s \Rightarrow_s t'_s;\assume{z}{ae'_s}
    \]
    By definition of $\Rightarrow_s$, 
    \[
        \SI{'}{'} p^1_s;p^2_s;\assume{z}{(e^1_s~e^2_s)} \Rightarrow_s t^2_s;t^4_s;
        \assume{z}{\ndi{(ae^2_s~ae^4_s)aa'}{v}{id}}
    \]
    Then $t'_s = t^2_s;t^4_s$, $ae'_s = \ndi{(ae^2_s~ae^4_s)aa'}{v}{id}$, $\SI{'}{'} p^1_s;\assume{z}{e^1_s} \Rightarrow_s t^2_s;\assume{z}{ae^2_s}$,
    and $\SI{'}{'} p^2_s;\assume{z}{e^2_s} \Rightarrow_s t^4_s;\assume{z}{ae^4_s}$.

   \noindent{By induction hypothesis}
    \[
        \exists~e_1, p^1_s, e^1_s. \SI{}{} e_1 \Rightarrow_s \_, \_, ae_1
    \wedge \SAC ae_1 \Rightarrow_{ex} ae^1_s, t^1_s\]
    \[
        \wedge t^1_s;\assume{z}{ae^1_s} \Rightarrow_r p^1_s;\assume{z}{e^1_s}
    \]
            \[
    \wedge \SI{'}{'} p^1_s;\assume{z}{e^1_s} \Rightarrow_s  t^2_s;\assume{z}{ae^2_s}
    \]
    \[
        \implies \exists~ae'_1.
        \SAC ae_1, ae^2_s, t^2_s \Rightarrow_{st} ae'_1
        \wedge \SI{'}{'} e_1 \Rightarrow_s \_,\_, ae'_1
    \]
    Because $\SI{}{} e_1 \Rightarrow_s \_, \_, ae_1$, $\SAC ae_1 \Rightarrow_{ex} ae^1_s, t^1_s$, 
    $t^1_s;\assume{z}{ae^1_s} \Rightarrow_r p^1_s;\assume{z}{e^1_s}$, and 
    $\SI{'}{'} p^1_s;\assume{z}{e^1_s} \Rightarrow_s t^2_s;\assume{z}{ae^2_s}$.
    \[
\SAC ae_1, ae^2_s, t^2_s \Rightarrow_{st} ae'_1
        \wedge \SI{'}{'} e_1 \Rightarrow_s \_,\_, ae'_1
    \]

   \noindent{By induction hypothesis}
    \[
        \exists~e_2, p^2_s, e^2_s. \SI{}{} e_2 \Rightarrow_s \_, \_, ae_2
    \wedge \SAC ae_2 \Rightarrow_{ex} ae^3_s, t^3_s\]
    \[
        \wedge t^3_s;\assume{z}{ae^3_s} \Rightarrow_r p^2_s;\assume{z}{e^2_s}
    \]
            \[
    \wedge \SI{'}{'} p^2_s;\assume{z}{e^2_s} \Rightarrow_s  t^4_s;\assume{z}{ae^4_s}
    \]
    \[
        \implies \exists~ae'_2.
        \SAC ae_2, ae^4_s, t^4_s \Rightarrow_{st} ae'_2
        \wedge \SI{'}{'} e_2 \Rightarrow_s \_,\_, ae'_2
    \]
    Because $\SI{}{} e_2 \Rightarrow_s \_, \_, ae_2$, $\SAC ae_2 \Rightarrow_{ex} ae^3_s, t^3_s$, 
    $t^3_s;\assume{z}{ae^3_s} \Rightarrow_r p^2_s;\assume{z}{e^2_s}$, and 
    $\SI{'}{'} p^2_s;\assume{z}{e^2_s} \Rightarrow_s t^4_s;\assume{z}{ae^4_s}$.
    \[
\SAC ae_2, ae^4_s, t^4_s \Rightarrow_{st} ae'_2
        \wedge \SI{'}{'} e_2 \Rightarrow_s \_,\_, ae'_2
    \]

    \noindent{Consider} $ae' = \ndi{(ae'_1~ae'_2)aa'}{v'}{id'}$.

    \noindent{By definition} of $\Rightarrow_{st}$, $\SAC ae_1, ae^2_s, t^2_s \Rightarrow_{st} ae'_1$, and
    $\SAC ae_2, ae^4_s, t^4_s \Rightarrow_{st} ae'_2$,
    \[
        \begin{array}{c}
        \SAC \ndi{(ae_1~ae_2)aa}{v}{id}, \ndi{(ae^2_s~ae^4_s)aa'}{v'}{id'}, 
            t'_s \\ \Rightarrow_{st} \ndi{(ae'_1~ae'_2)aa'}{v'}{id'}
        \end{array}
    \]
    Therefore
    \[
        \SAC ae, ae'_s, t'_s \Rightarrow_{st} ae'
    \]

    \noindent{By definition} of $\Rightarrow_s$, $\SI{'}{'} e_1 \Rightarrow_s \_, \_, ae'_1$, $\SI{'}{'} e_2 \Rightarrow_s, \_, \_, ae'_2$,
    and because $\cV(ae'_1) = \cV(ae^2_s)$, $\cV(ae'_2) = \cV(ae^4_s)$ (Observation~\ref{ex:obs2})
    \[
        \SI{'}{'} e \Rightarrow_s \_, \_,  \ndi{(ae'_1~ae'_2)aa'}{v'}{id'}
    \]
    Therefore
    \[
        \SI{'}{'} e \Rightarrow_s \_, \_, ae'
    \]
    
    \noindent{Therefore} when $ae = \ndi{(ae_1~ae_2)aa}{v}{id}$, $ID(ae_1) \in \cS$
    \[
        \exists~e, p_s. \SI{}{} e \Rightarrow_s \_, \_, ae
    \wedge \SAC ae \Rightarrow_{ex} ae_s, t_s\]
    \[
    \wedge t_s;\assume{z}{ae_s} \Rightarrow_r p_s
    \wedge \SI{'}{'} p_s \Rightarrow_s  t'_s;\assume{z}{ae'_s}
    \]
    \[
        \implies \exists~ae'.
        \SAC ae, ae'_s, t'_s \Rightarrow_{st} ae', \_
        \wedge \SI{'}{'} e \Rightarrow_s \_,\_, ae'
    \]

    \noindent{\bf Case 5:} 
    $ae = \ndi{(ae_1~ae_2)y=ae_3}{v}{id}$ and $ID(ae_1) \notin \cS$

    \noindent{By assumption}
    \[
        \SI{}{} e \Rightarrow_s \_, \_, ae
    \]
    By definition of $\Rightarrow_s$
    \[
        \SI{}{} (e_1~e_2) \Rightarrow_s \_, \_, \ndi{(ae_1~ae_2)y=ae_3}{v}{id}
    \]
    Then $e = (e_1~e_2)$, $\SI{}{} e_1 \Rightarrow_s \_, \_, ae_1$, 
    $\SI{}{} e_2 \Rightarrow_s \_, \_, ae_2$, $\SI{''}{''} e_3 \Rightarrow_s \_, \_, ae_3$, and $\cV(ae_1) = \tup{\lambda.x~e_3, \sigma''_v, \sigma''_{id}}$.

    \noindent{By assumption}
    \[
        \SAC ae \Rightarrow_{ex} ae_s, t_s
    \]
    By definition of $\Rightarrow_{ex}$
    \[
        \SAC \ndi{(ae_1~ae_2)y=ae_3}{v}{id} \Rightarrow_{ex} ae^5_s, t^1_s;\assume{x}{ae^1_s};t^3_s;\assume{y}{ae^3_s};t^5_s
    \]
    Then $ae_s = ae^5_s$, $t_s = t^1_s;\assume{x}{ae^1_s};t^3_s;\assume{y}{ae^3_s};t^5_s$, 
    $\SAC ae_1 \Rightarrow_{ex} ae^1_s, t^1_s$, $\SAC ae_2 \Rightarrow_{ex} ae^3_s, t^3_s$, and
    $\SAC ae_3 \Rightarrow_{ex} ae^5_s, t^5_s$.

    \noindent{By assumption}
    \[
        t_s;\assume{z}{ae_s} \Rightarrow_r p_s
    \]
    By definition of $\Rightarrow_r$
    \[
        \begin{array}{c}
        t^1_s;\assume{x}{ae^1_s};t^3_s;\assume{y}{ae^3_s};t^5_s\assume{z}{ae^5_s} \Rightarrow_r \\
        p^1_s;\assume{x}{e^1_s};p^2_s;\assume{y}{e^2_s};p^3_s;\assume{z}{e^3_s}
        \end{array}
    \]
    Then $p_s =  p^1_s;\assume{x}{e^1_s};p^2_s;\assume{y}{e^2_s};p^3_s;\assume{z}{e^3_s}$, $ae^1_s \Rightarrow_r e^1_s$, 
    $t^1_s \Rightarrow_r p^1_s$, $ae^3_s \Rightarrow_r e^2_s$, $t^3_s \Rightarrow_r p^2_s$, $ae^5_s \Rightarrow_r e^3_s$, and $t^5_s \Rightarrow_r p^3_s$.

    \noindent{By assumption}
    \[
        \SI{'}{'} p_s \Rightarrow_s t'_s;\assume{z}{ae'_s}
    \]
    By definition of $\Rightarrow_s$, 
    \[
        \begin{array}{c}
        \SI{'}{'}  p^1_s;\assume{x}{e^1_s};p^2_s;\assume{y}{e^2_s};p^3_s;\assume{z}{e^3_s}
 \Rightarrow_s \\ t^2_s;\assume{x}{ae^2_s};t^4_s;\assume{y}{ae^4_s};t^6_s
        \assume{z}{ae^6_s}
        \end{array}
    \]
    Then $t'_s = t^2_s;\assume{x}{ae^2_s};t^4_s;\assume{y}{ae^4_s};t^6_s$,
    $ae'_s = ae^6_s$, $\SI{'}{'} p^1_s;\assume{z}{e^1_s} \Rightarrow_s t^2_s;\assume{z}{ae^2_s}$,
    $\SI{'}{'} p^2_s;\assume{z}{e^2_s} \Rightarrow_s t^4_s;\assume{z}{ae^4_s}$, $\SI{'''}{'''} p^3_s;\assume{z}{e^3_s} \Rightarrow_s t^6_s;\assume{z}{ae^6_s}$.

   \noindent{By induction hypothesis}
    \[
        \exists~e_1, p^1_s, e^1_s. \SI{}{} e_1 \Rightarrow_s \_, \_, ae_1
    \wedge \SAC ae_1 \Rightarrow_{ex} ae^1_s, t^1_s\]
    \[
        \wedge t^1_s;\assume{z}{ae^1_s} \Rightarrow_r p^1_s;\assume{z}{e^1_s}
    \]
            \[
    \wedge \SI{'}{'} p^1_s;\assume{z}{e^1_s} \Rightarrow_s  t^2_s;\assume{z}{ae^2_s}
    \]
    \[
        \implies \exists~ae'_1.
        \SAC ae_1, ae^2_s, t^2_s \Rightarrow_{st} ae'_1
        \wedge \SI{'}{'} e_1 \Rightarrow_s \_,\_, ae'_1
    \]
    Because $\SI{}{} e_1 \Rightarrow_s \_, \_, ae_1$, $\SAC ae_1 \Rightarrow_{ex} ae^1_s, t^1_s$, 
    $t^1_s;\assume{z}{ae^1_s} \Rightarrow_r p^1_s;\assume{z}{e^1_s}$, and 
    $\SI{'}{'} p^1_s;\assume{z}{e^1_s} \Rightarrow_s t^2_s;\assume{z}{ae^2_s}$.
    \[
\SAC ae_1, ae^2_s, t^2_s \Rightarrow_{st} ae'_1
        \wedge \SI{'}{'} e_1 \Rightarrow_s \_,\_, ae'_1
    \]

   \noindent{By induction hypothesis}
    \[
        \exists~e_2, p^2_s, e^2_s. \SI{}{} e_2 \Rightarrow_s \_, \_, ae_2
    \wedge \SAC ae_2 \Rightarrow_{ex} ae^3_s, t^3_s\]
    \[
        \wedge t^3_s;\assume{z}{ae^3_s} \Rightarrow_r p^2_s;\assume{z}{e^2_s}
    \]
            \[
    \wedge \SI{'}{'} p^2_s;\assume{z}{e^2_s} \Rightarrow_s  t^4_s;\assume{z}{ae^4_s}
    \]
    \[
        \implies \exists~ae'_2.
        \SAC ae_2, ae^4_s, t^4_s \Rightarrow_{st} ae'_2
        \wedge \SI{'}{'} e_2 \Rightarrow_s \_,\_, ae'_2
    \]
    Because $\SI{}{} e_2 \Rightarrow_s \_, \_, ae_2$, $\SAC ae_2 \Rightarrow_{ex} ae^3_s, t^3_s$, 
    $t^3_s;\assume{z}{ae^3_s} \Rightarrow_r p^2_s;\assume{z}{e^2_s}$, and 
    $\SI{'}{'} p^2_s;\assume{z}{e^2_s} \Rightarrow_s t^4_s;\assume{z}{ae^4_s}$.
    \[
\SAC ae_2, ae^4_s, t^4_s \Rightarrow_{st} ae'_2
        \wedge \SI{'}{'} e_2 \Rightarrow_s \_,\_, ae'_2
    \]

   \noindent{By induction hypothesis}
    \[
        \exists~e_3, p^3_s, e^3_s. \SI{''}{''} e_3 \Rightarrow_s \_, \_, ae_3
    \wedge \SAC ae_3 \Rightarrow_{ex} ae^5_s, t^5_s\]
    \[
        \wedge t^5_s;\assume{z}{ae^5_s} \Rightarrow_r p^3_s;\assume{z}{e^3_s}
    \]
            \[
    \wedge \SI{'''}{'''} p^3_s;\assume{z}{e^3_s} \Rightarrow_s  t^6_s;\assume{z}{ae^6_s}
    \]
    \[
        \implies \exists~ae'_3.
        \SAC ae_3, ae^6_s, t^6_s \Rightarrow_{st} ae'_3
        \wedge \SI{'}{'} e_3 \Rightarrow_s \_,\_, ae'_3
    \]
    Because $\SI{''}{''} e_3 \Rightarrow_s \_, \_, ae_3$, $\SAC ae_3 \Rightarrow_{ex} ae^5_s, t^5_s$, 
    $t^5_s;\assume{z}{ae^5_s} \Rightarrow_r p^3_s;\assume{z}{e^3_s}$, and 
    $\SI{'''}{'''} p^3_s;\assume{z}{e^3_s} \Rightarrow_s t^6_s;\assume{z}{ae^6_s}$.
    \[
\SAC ae_3, ae^6_s, t^6_s \Rightarrow_{st} ae'_3
        \wedge \SI{'''}{'''} e_3 \Rightarrow_s \_,\_, ae'_3
    \]

    \noindent{Consider} $ae' = \ndi{(ae'_1~ae'_2)y=ae'_3}{v'}{id'}$.

    \noindent{By definition} of $\Rightarrow_{st}$, $\SAC ae_1, ae^2_s, t^2_s \Rightarrow_{st} ae'_1$,
    $\SAC ae_2, ae^4_s, t^4_s \Rightarrow_{st} ae'_2$, and $\SAC ae_3, ae^6_s, t^6_s \Rightarrow_{st} ae'_3$
    \[
        \begin{array}{c}
        \SAC \ndi{(ae_1~ae_2)y=ae_3}{v}{id}, \ndi{(ae^2_s~ae^4_s)y=ae^6_s}{v'}{id'}, 
            t'_s \\ \Rightarrow_{st} \ndi{(ae'_1~ae'_2)y=ae'_3}{v'}{id'}
        \end{array}
    \]
    Therefore
    \[
        \SAC ae, ae'_s, t'_s \Rightarrow_{st} ae'
    \]

    \noindent{By definition} of $\Rightarrow_s$, $\SI{'}{'} e_1 \Rightarrow_s \_, \_, ae'_1$, $\SI{'}{'} e_2 \Rightarrow_s, \_, \_, ae'_2$,
    and because $\cV(ae'_1) = \cV(ae^2_s)$, $\cV(ae'_1) = \cV(ae^2_s)$ (Observation~\ref{ex:obs2}), $\SI{'''}{'''} e_3 \Rightarrow_s \_, \_, ae'_3$
    \[
        \SI{'}{'} e \Rightarrow_s \_, \_,  \ndi{(ae'_1~ae'_2)y=ae'_3}{v'}{id'}
    \]
    Therefore
    \[
        \SI{'}{'} e \Rightarrow_s \_, \_, ae'
    \]
    
    \noindent{Therefore} when $ae = \ndi{(ae_1~ae_2)y=ae_3}{v}{id}$, $ID(ae_1) \notin \cS$
    \[
        \exists~e, p_s. \SI{}{} e \Rightarrow_s \_, \_, ae
    \wedge \SAC ae \Rightarrow_{ex} ae_s, t_s\]
    \[
    \wedge t_s;\assume{z}{ae_s} \Rightarrow_r p_s
    \wedge \SI{'}{'} p_s \Rightarrow_s  t'_s;\assume{z}{ae'_s}
    \]
    \[
        \implies \exists~ae'.
        \SAC ae, ae'_s, t'_s \Rightarrow_{st} ae', \_
        \wedge \SI{'}{'} e \Rightarrow_s \_,\_, ae'
    \]
 
    \noindent{Because} we have covered all cases, using induction, the lemma is true.
    \end{proof}

    \begin{lemma}
    \noindent{Given} environements $\sigma_v, \sigma_{id}, \sigma'_v$ and $\sigma'_{id}$
    such that $\dom \sigma_v = \dom \sigma_{id} = \dom \sigma'_v = 
    \dom \sigma'_{id}$ a partial trace $t$ within a trace $t_p$($t$ is a suffix of trace $t_p$), 
    and a valid subproblem $\cS$ over trace $t_p$
    \[
        \exists.p, p_s~\SI{}{} p \Rightarrow_s t\\
        \wedge
        \SAC t \Rightarrow_{ex} t_s\\
        \wedge
        p_s \Rightarrow_r t_s\\
        \wedge 
        \SI{'}{'} p_s \Rightarrow_s t'_s\\
    \]
        \[\implies \exists~t'.
        \SAC t, t'_s \Rightarrow_{st} t'
        \wedge \SI{'}{'} p \Rightarrow_s t'
        \]
        \label{lem:sound3tr}
    \end{lemma}
    \begin{proof}
   \noindent{Proof by induction}

    \noindent{\bf Base Case:}
    $t = \emptyset$ 
    
    \noindent{By assumption}
        \[
            \SI{}{} p \Rightarrow_s t
        \]
        By definition of $\Rightarrow_s$
        \[
            \SI{}{} \emptyset \Rightarrow_s \emptyset
        \]
        Then $p = \emptyset$.

        \noindent{By assumption}
        \[
            \SAC t \Rightarrow_{ex} t_s
        \]
        By definition of $\Rightarrow_{ex}$
        \[
            \SAC \emptyset \Rightarrow_{ex} \emptyset
        \]
        Then $t_s = \emptyset$.

        \noindent{By assumption}
        \[
            t_s \Rightarrow_r p_s
        \]
        By definition of $\Rightarrow_r$
        \[
            \emptyset \Rightarrow_r \emptyset
        \]
        Then $p_s = \emptyset$.

        \noindent{By assumption}
        \[
            \SI{'}{'} p_s \Rightarrow_s t'_s
        \]
        By definition of $\Rightarrow_s$
        \[
            \SI{'}{'} \emptyset \Rightarrow_s \emptyset
        \]
        Then $t'_s = \emptyset$.

    \noindent{Consider} $t' = \emptyset$.

        \noindent{By definition} of $\Rightarrow_{st}$
        \[
            \SAC \emptyset, \emptyset \Rightarrow_{st} \emptyset
        \]
        Therefore
        \[
            \SAC t, t'_s \Rightarrow_{st} t'
        \]

        \noindent{By definition} of $\Rightarrow_s$
        \[
            \SI{'}{'} \emptyset \Rightarrow_s \emptyset
        \]
        Therefore
        \[
            \SI{'}{'} p \Rightarrow_s t'
        \]

    \noindent{Therefore} when $t = \emptyset$
    \[
        \exists.p, p_s~\SI{}{} p \Rightarrow_s t\\
        \wedge
        \SAC t \Rightarrow_{ex} t_s\\
        \wedge
        p_s \Rightarrow_r t_s\\
        \wedge 
        \SI{'}{'} p_s \Rightarrow_s t'_s\\
    \]
    \[
        \implies
        \SAC t, t'_s \Rightarrow_{st} t'
        \wedge \SI{'}{'} p \Rightarrow_s t'
    \]

    \noindent{\bf Induction Case:}

    \noindent{ Case 1:}
    $t = \assume{x}{ae};t_1$
    
    \noindent{By assumption}
        \[
            \SI{}{} p \Rightarrow_s t
        \]
        By definition of $\Rightarrow_s$
        \[
            \SI{}{} \assume{x}{e};p_1 \Rightarrow_s \assume{x}{ae};t_1
        \]
        Then $p = \assume{x}{e};p_1$, $\SI{}{} e \Rightarrow_s v, id, ae$, $\sigma''_v = \sigma_v[x \rightarrow v]$, 
        $\sigma''_{id} = \sigma_{id}[x \rightarrow id]$, and
        $\SI{''}{''} p_1 \Rightarrow_s t_1$.

        \noindent{By assumption}
        \[
            \SAC t \Rightarrow_{ex} t_s
        \]
        By definition of $\Rightarrow_{ex}$
        \[
            \SAC \assume{x}{ae};t_1 \Rightarrow_{ex} t^1_s;\assume{x}{ae_s};t^3_s
        \]
        Then $t_s = t^1_s;\assume{x}{ae_s};t^3_s$, $\SAC ae \Rightarrow_{ex} ae_s, t^1_s$, and
        $\SAC t_1 \Rightarrow_{ex} t^3_s$.

        \noindent{By assumption}
        \[
            t_s \Rightarrow_r p_s
        \]
        By definition of $\Rightarrow_r$
        \[
            t^1_s;\assume{x}{ae_s};t^3_s \Rightarrow_r p^1_s;\assume{x}{e_s};p^2_s
        \]
        Then $p_s = p^1_s;\assume{x}{e_s};p^2_s$, $t^1_s \Rightarrow_r p^1_s$, $ae_s \Rightarrow_r e_s$, and $t^3_s \Rightarrow_r p^2_s$.

        \noindent{By assumption}
        \[
            \SI{'}{'} p_s \Rightarrow_s t'_s
        \]
        By definition of $\Rightarrow_s$
        \[
            \SI{'}{'} p^1_s;\assume{x}{e_s};p^2_s \Rightarrow_s t^2_s;\assume{x}{ae'_s};t^4_s
        \]
        Then $t'_s =  t^2_s;\assume{x}{ae'_s};t^4_s$, $\SI{'}{'} p^1_s;\assume{x}{e_s} \Rightarrow_s t^2_s;\assume{x}{ae'_s}$, 
        $\sigma'''_v = \sigma'_v[x \rightarrow \cV(ae'_s)]$, $\sigma'''_{id} = \sigma'_{id}[x \rightarrow ID(ae'_s)]$, 
        and $\SI{'''}{'''} p^2_s \Rightarrow_s t^4_s$ (Observation~\ref{ex:obs}, variable names in $t^2_s$ do not collide with 
        variable names in $t^4_s$).

        \noindent{By induction hypothesis}
        \[
        \SI{''}{''} p_1 \Rightarrow_s t_1\\
        \wedge
        \SAC t_1 \Rightarrow_{ex} t^3_s\\
        \wedge
        t^3_s \Rightarrow_r p^2_s\\
        \wedge 
        \SI{'''}{'''} p^2_s \Rightarrow_s t^4_s\\
    \]
    \[
        \implies
        \SAC t_1, t^4_s \Rightarrow_{st} t'_1
        \wedge \SI{'}{'} p_1 \Rightarrow_s t'_1
    \]
    Because $\SI{''}{''} p_1 \Rightarrow_s t_1$, $\SAC t_1 \Rightarrow_{ex} t^3_s$,
    $\SAC t_1 \Rightarrow_{ex} t^3_s$, $t^3_s \Rightarrow_r p^2_s$, and $\SI{'''}{'''} p^2_s \Rightarrow_s t^4_s$,
        \[
        \SAC t_1, t^4_s \Rightarrow_{st} t'_1
        \wedge \SI{'}{'} p_1 \Rightarrow_s t'_1
    \]

    \noindent{By statement~\ref{lem:sound3exp}}
    \[
       \SI{}{} e \Rightarrow_s \_, \_, ae
    \wedge \SAC ae \Rightarrow_{ex} ae_s, t^1_s\]
    \[
        \wedge t^1_s;\assume{z}{ae_s} \Rightarrow_r p^1_s;\assume{z}{e_s}
    \]\[
    \wedge \SI{'}{'} p^1_s;\assume{z}{e_s} \Rightarrow_s  t^2_s;\assume{z}{ae'_s}
    \]
    \[
        \implies
        \SAC ae, ae'_s, t^2_s \Rightarrow_{st} ae'
        \wedge \SI{'}{'} e \Rightarrow_s \_,\_, ae'
    \]
    Because $\SI{}{} e \Rightarrow_s \_, \_, ae$, $\SAC ae \Rightarrow_{ex} ae_s, t^1_s$, 
    $t^1_s;\assume{z}{ae_s} \Rightarrow_r p^1_s;\assume{z}{e_s}$ and $\SI{'}{'} p^1_s;\assume{z}{e_s} 
    \Rightarrow_s t^2_s;\assume{z}{ae'_s}$,
    \[
        \SAC ae, ae'_s, t^2_s \Rightarrow_{st} ae'
        \wedge \SI{'}{'} e \Rightarrow_s \_,\_, ae'
    \]

    \noindent{Consider} $t' = \assume{x}{ae'};t'_1$.

        \noindent{By definition} of $\Rightarrow_{st}$, $\SAC t_1, t^4_s \Rightarrow_{st} t'_1$, and 
        $\SAC ae, ae'_s, t^2_s \Rightarrow_{st} ae'$,
        \[
            \SAC \assume{x}{ae}, t^2_s;\assume{x}{ae'_s};t^4_s \Rightarrow_{st} \assume{x}{ae'};t'_1
        \]
        Therefore
        \[
            \SAC t, t'_s \Rightarrow_{st} t'
        \]

        \noindent{By definition} of $\Rightarrow_s$, $\SI{'}{'} e \Rightarrow_s \_, \_, ae'$, and 
        $\SI{'''}{'''} p^2_s \Rightarrow_s t^4_s$ ($\cV(ae') = \cV(ae'_s)$ and $ID(ae')= ID(ae'_s)$ Observation~\ref{ex:obs2})
        \[
            \SI{'}{'} \assume{x}{e};p_1 \Rightarrow_s \assume{x}{ae'};t'_1
        \]
        Therefore
        \[
            \SI{'}{'} p \Rightarrow_s t'
        \]

    \noindent{Therefore} when $t = \assume{x}{ae};t_1$
    \[
        \exists.p, p_s~\SI{}{} p \Rightarrow_s t\\
        \wedge
        \SAC t \Rightarrow_{ex} t_s\\
        \wedge
        p_s \Rightarrow_r t_s\\
        \wedge 
        \SI{'}{'} p_s \Rightarrow_s t'_s\\
    \]
    \[
        \implies
        \SAC t, t'_s \Rightarrow_{st} t'
        \wedge \SI{'}{'} p \Rightarrow_s t'
    \]

    \noindent{\bf Case 2:}
    $t = \observed{ae}{e_v};t_1$

    \noindent{By assumption}
        \[
            \SI{}{} p \Rightarrow_s t
        \]
        By definition of $\Rightarrow_s$
        \[
            \SI{}{} \observed{e}{e_v};p_1 \Rightarrow_s \observed{ae}{e_v};t_1
        \]
        Then $p = \observed{e}{e_v};p_1$, $\SI{}{} e \Rightarrow_s v, id, ae$, and
        $\SI{}{} p_1 \Rightarrow_s t_1$.

        \noindent{By assumption}
        \[
            \SAC t \Rightarrow_{ex} t_s
        \]
        By definition of $\Rightarrow_{ex}$
        \[
            \SAC \observed{ae}{e_v};t_1 \Rightarrow_{ex} t^1_s;\observed{ae_s}{e_v};t^3_s
        \]
        Then $t_s = t^1_s;\observed{ae_s}{e_v};t^3_s$, $\SAC ae \Rightarrow_{ex} ae_s, t^1_s$, and
        $\SAC t_1 \Rightarrow_{ex} t^3_s$.

        \noindent{By assumption}
        \[
            t_s \Rightarrow_r p_s
        \]
        By definition of $\Rightarrow_r$
        \[
            t^1_s;\observed{ae_s}{e_v};t^3_s \Rightarrow_r p^1_s;\observed{e_s}{e_v};p^2_s
        \]
        Then $p_s = p^1_s;\observed{ae_s}{e_v};p^2_s$, $t^1_s \Rightarrow_r p^1_s$, $ae_s \Rightarrow_r e_s$, and $t^3_s \Rightarrow_r p^2_s$.

        \noindent{By assumption}
        \[
            \SI{'}{'} p_s \Rightarrow_s t'_s
        \]
        By definition of $\Rightarrow_s$
        \[
            \SI{'}{'} p^1_s;\observed{e_s}{e_v};p^2_s \Rightarrow_s t^2_s;\observed{ae'_s}{e_v};t^4_s
        \]
        Then $t'_s =  t^2_s;\observed{ae'_s}{e_v};t^4_s$, $\SI{'}{'} p^1_s;\observed{e_s}{e_v} \Rightarrow_s t^2_s;\observed{ae'_s}{e_v}$, 
        and $\SI{'}{'} p^2_s \Rightarrow_s t^4_s$ (Observation~\ref{ex:obs}, variable names in $t^2_s$ do not collide with 
        variable names in $t^4_s$).

        \noindent{By induction hypothesis}
        \[
        \SI{}{} p_1 \Rightarrow_s t_1\\
        \wedge
        \SAC t_1 \Rightarrow_{ex} t^3_s\\
        \wedge
        t^3_s \Rightarrow_r p^2_s\\
        \wedge 
        \SI{'}{'} p^2_s \Rightarrow_s t^4_s\\
    \]
    \[
        \implies
        \SAC t_1, t^4_s \Rightarrow_{st} t'_1
        \wedge \SI{'}{'} p_1 \Rightarrow_s t'_1
    \]
    Because $\SI{}{} p_1 \Rightarrow_s t_1$, $\SAC t_1 \Rightarrow_{ex} t^3_s$,
    $\SAC t_1 \Rightarrow_{ex} t^3_s$, $t^3_s \Rightarrow_r p^2_s$, and $\SI{'}{'} p^2_s \Rightarrow_s t^4_s$,
        \[
        \SAC t_1, t^4_s \Rightarrow_{st} t'_1
        \wedge \SI{'}{'} p_1 \Rightarrow_s t'_1
    \]

    \noindent{By statement~\ref{lem:sound3exp}}
    \[
       \SI{}{} e \Rightarrow_s \_, \_, ae
    \wedge \SAC ae \Rightarrow_{ex} ae_s, t^1_s\]
    \[
        \wedge t^1_s;\observed{ae_s}{e_v} \Rightarrow_r p^1_s;\observed{e_s}{e_v}
    \]\[
    \wedge \SI{'}{'} p^1_s;\observed{e_s}{e_v} \Rightarrow_s  t^2_s;\observed{e_s}{ae'_v}
    \]
    \[
        \implies
        \SAC ae, ae'_s, t^2_s \Rightarrow_{st} ae'
        \wedge \SI{'}{'} e \Rightarrow_s \_,\_, ae'
    \]
    Because $\SI{}{} e \Rightarrow_s \_, \_, ae$, $\SAC ae \Rightarrow_{ex} ae_s, t^1_s$, 
    $t^1_s;\observed{ae_s}{e_v} \Rightarrow_r p^1_s;\observed{e_s}{e_v}$ and $\SI{'}{'} p^1_s;\observed{e_s}{e_v} 
    \Rightarrow_s t^2_s;\observed{ae'_s}{e_v}$,
    \[
        \SAC ae, ae'_s, t^2_s \Rightarrow_{st} ae'
        \wedge \SI{'}{'} e \Rightarrow_s \_,\_, ae'
    \]

    \noindent{Consider} $t' = \observed{ae'}{e_v};t'_1$.

        \noindent{By definition} of $\Rightarrow_{st}$, $\SAC t_1, t^4_s \Rightarrow_{st} t'_1$, and 
        $\SAC ae, ae'_s, t^2_s \Rightarrow_{st} ae'$,
        \[
            \SAC \observed{ae}{e_v}, t^2_s;\observed{ae'_s}{e_v};t^4_s \Rightarrow_{st} \observed{ae'}{e_v};t'_1
        \]
        Therefore
        \[
            \SAC t, t'_s \Rightarrow_{st} t'
        \]

        \noindent{By definition} of $\Rightarrow_s$, $\SI{'}{'} e \Rightarrow_s \_, \_, ae'$, and 
        $\SI{'}{'} p^2_s \Rightarrow_s t^4_s$
        \[
            \SI{'}{'}\observed{e}{e_v};p_1 \Rightarrow_s \observed{ae'}{e_v};t'_1
        \]
        Therefore
        \[
            \SI{'}{'} p \Rightarrow_s t'
        \]

    \noindent{Therefore} when $t = \observed{ae}{e_v};t_1$
    \[
        \exists.p, p_s~\SI{}{} p \Rightarrow_s t\\
        \wedge
        \SAC t \Rightarrow_{ex} t_s\\
        \wedge
        p_s \Rightarrow_r t_s\\
        \wedge 
        \SI{'}{'} p_s \Rightarrow_s t'_s\\
    \]
    \[
        \implies
        \SAC t, t'_s \Rightarrow_{st} t'
        \wedge \SI{'}{'} p \Rightarrow_s t'
    \]

    \noindent{Because} we have covered all cases, using induction, the  below statement 
    is true.
    \[
        \exists.p, p_s~\SI{}{} p \Rightarrow_s t\\
        \wedge
        \SAC t \Rightarrow_{ex} t_s\\
        \wedge
        p_s \Rightarrow_r t_s\\
        \wedge 
        \SI{'}{'} p_s \Rightarrow_s t'_s\\
    \]
    \[
        \implies
        \SAC t, t'_s \Rightarrow_{st} t'
        \wedge \SI{'}{'} p \Rightarrow_s t'
    \]
    \end{proof}
\begin{lemma}
    \label{lem:sound3}
    Given a valid trace $t$ and a valid subproblem $\cS$ 
    and subtrace $t_s = \mathsf{ExtractTrace}(t, \cS)$ 
    for all possible subtraces $t'_s$ :
    \[t'_s \in
    \mathsf{Traces}(\mathsf{Program}(t_s))\implies \exists~t'.t' = \mathsf{StitchTrace}(t, t'_s, \cS) 
    ~\wedge~
    t' \in
    \mathsf{Traces}(\mathsf{Program}(t))\]
\end{lemma}
\begin{proof}
\noindent{Since} statement \ref{lem:sound3tr} is true for all $\sigma_v, \sigma_{id}, \sigma'_v$, and $\sigma'_{id}$,
    given $\dom \sigma_v = \dom \sigma_{id} = \dom \sigma'_v = \dom \sigma'_{id}$, 
    replacing $\sigma_v, \sigma_{id}, \sigma'_v$, and $\sigma'_{id}$ with $\emptyset$ (empty environment) will result in 
    \[
        \exists.p, p_s~t \in \mathsf{Traces}(p)
        \wedge
       \SAC t \Rightarrow_{ex} t_s
        \wedge 
        t_s, t'_s  \in \mathsf{Traces}(p_s)
        \implies \SAC t, t'_s \Rightarrow_{st} t'
        \wedge t' \in \mathsf{Traces}(p)
    \]

\end{proof}

\begin{theorem}
    \label{thm:sound}
    Given a valid trace $t$, a valid subproblem $\cS$, subtrace $t_s 
    = \mathsf{ExtractTrace}(t, \cS)$. 
    For all possible subtraces $t_s'$, 
    $t_s' \in \mathsf{Traces}(\mathsf{Program}(t_s))$ implies
    there exist a trace $t'$ such that:
    \begin{itemize}
        \item $t' = \mathsf{StitchTrace}(t, t'_s, \cS)$
        \item $t' \in \mathsf{Traces}(\mathsf{Program}(t))$
        \item $\SAC t \equiv t'$
    \end{itemize}
\end{theorem}
\begin{proof}
    From Corollary \ref{lem:sound2} and Lemma \ref{lem:sound3}.
\end{proof}

\subsection{Completeness}
\label{append:comp}
Within this section we prove that our interface is complete, i.e.
given a valid trace $t$, a valid subproblem $\cS$ on $t$, a subtrace $t_s = \mathsf{ExtractTrace}(t, t_s, \cS)$, 
for any trace $t'$ which can be achived from entangled subproblem interface (i.e. $t' \in \mathsf{Traces}(\mathsf{Program}(t))$ and $\SAC t \equiv t'$),
there exists a subtrace $t'_s \in \mathsf{Traces}(\mathsf{Program}(t_s))$ such that $t' = \mathsf{StitchTrace}(t, t'_s, \cS)$.

\noindent{Formally,} given valid trace $t$, a valid subproblem $\cS$ on $t$, a subtrace $t_s$, and any trace $t'$
\[
    t_s = \mathsf{ExtractTrace}(t, \cS)~\wedge~t' \in \mathsf{Traces}(\mathsf{Program}(t))~\wedge~\SAC t \equiv t'
\]
\[
    \implies \exists~t'_s.~t'_s \in \mathsf{Traces}(\mathsf{Program}(t_s))~\wedge~t' = \mathsf{StitchTrace}(t, t'_s, \cS)
\]

\noindent{We need to prove a few lemmas which will aid us in proving the above statement.}
   
   \begin{lemma}
        \noindent{Given an augmented} expression $ae$ and a subproblem $\cS$, a subtrace $t_s$, subaugmented expression $ae_s$
        and an augmented expressions $ae'$ such that
        \[
            \SAC ae \Rightarrow_{ex} ae_s, t_s
        \wedge
            \SAC ae \equiv ae'
        \wedge
           \exists~e.~ae \Rightarrow_r e
        \wedge
            \SI{}{} e \Rightarrow_s \_, \_, ae'
        \]
        then, there exists an $ae'_s$, $t'_s$, $p_s$ and $e_s$ such that
        \[
            \SAC ae, ae'_s, t'_s \Rightarrow_{st} ae'
        \wedge
            t_s \Rightarrow_r p_s~\wedge~ae_s \Rightarrow_r e_s
        \]
            \[
                \wedge~
            \SI{}{} p_s;\assume{z}{e_s} \Rightarrow_s t'_s;\assume{z}{ae'_s}
        \]
   \end{lemma}
    \begin{proof}
        \noindent{Proof by Induction}

        \noindent{\bf Base Case:}

        \noindent{Case 1:} $ae = \ndi{x}{x}{id}$
        
        \noindent{By assumption}
        \[
            \SAC ae \Rightarrow_{ex} ae_s, t_s
        \]
        By definition of $\Rightarrow_{ex}$ 
        \[
            \SAC \ndi{x}{x}{id} \Rightarrow_{ex} \ndi{x}{x}{id}, \emptyset
        \]
        Then $ae_s = \ndi{x}{x}{id}$ and $t_s = \emptyset$.

        \noindent{By assumption}
        \[
            \SAC ae \equiv ae'
        \]
        By definition of $\equiv$
        \[
            \SAC \ndi{x}{x}{id} \equiv \ndi{x}{x}{id'}
        \]
        Then $ae' = \ndi{x}{x}{id'}$.

        \noindent{By assumption}
        \[
            ae \Rightarrow_r e
        \]
        By definition of $\Rightarrow_r$
        \[
            \ndi{x}{x}{id} \Rightarrow_r x
        \]
        Then $e = x$.

        \noindent{By assumption}
        \[
            \SI{}{} e \Rightarrow_s \_, \_, ae'
        \]
        By definition of $\Rightarrow_s$
        \[
            \SI{}{} x \Rightarrow_s \_, \_, \ndi{x}{x}{id'}
        \]
        Then $x \notin \dom \sigma_v$ and $id'$ is a unique id.

        \noindent{Consider} $ae'_s = \ndi{x}{x}{id'}$, $t'_s = \emptyset$, 
        $p_s = \emptyset$ and $e_s = x$. By definition of $\Rightarrow_{st}$
        \[
            \SAC \ndi{x}{x}{id}, \ndi{x}{x}{id'}, \emptyset \Rightarrow_{st} \ndi{x}{x}{id'}
        \]
        Therefore 
        \[
            \SAC ae, ae'_s, t'_s \Rightarrow_{st} ae'
        \]
        
        \noindent{By} definition of $\Rightarrow_r$
        \[
            \emptyset \Rightarrow_r \emptyset
        \]
        Therefore
        \[
            t_s \Rightarrow_r p_s
        \]

        \noindent{By} definition of $\Rightarrow_r$
        \[
            \ndi{x}{x}{id} \Rightarrow_r x
        \]
        Therefore
        \[
            ae_s \Rightarrow_r e_s
        \]

        \noindent{Because} $x \notin \dom \sigma_v$, the definition of $\Rightarrow_s$ implies
        \[
            \SI{}{} \assume{z}{x} \Rightarrow_s \assume{z}{\ndi{x}{x}{id'}}
        \]
        Therefore
        \[
            \SI{}{} p_s;\assume{z}{e_s} \Rightarrow_s t_s;\assume{z}{ae'_s}
        \]

        \noindent{Therefore} when $ae = \ndi{x}{x}{id}$,
        \[
            \forall t_s, ae_s, ae'~\SAC ae \Rightarrow_{ex} ae_s, t_s~\wedge~\SAC ae \equiv ae'~\exists e.~ae \Rightarrow_r e~\wedge~\SI{}{} e \Rightarrow_s \_, \_, ae'
        \]
        implies
        \[
            \exists ae'_s, t'_s, p_s, e_s.~\SAC ae, ae'_s, t'_s \Rightarrow_{st} ae'~\wedge~t_s \Rightarrow_r p_s~\wedge~ae_s \Rightarrow_r e_s
        \]
        \[
            \wedge~\SI{}{} p_s;\assume{z}{e_s} \Rightarrow_s t'_s;\assume{z}{ae'_s}
        \]

        \noindent{Case 2:} $ae = \ndi{x(id_v)}{v}{id}$
       
         \noindent{By assumption}
        \[
            \SAC ae \Rightarrow_{ex} ae_s, t_s
        \]
        By definition of $\Rightarrow_{ex}$ 
        \[
            \SAC \ndi{x(id_v)}{v}{id} \Rightarrow_{ex} \ndi{x(id_v)}{v}{id}, \emptyset
        \]
        Then $ae_s = \ndi{x(id_v)}{v}{id}$ and $t_s = \emptyset$.

        \noindent{By assumption}
        \[
            \SAC ae \equiv ae'
        \]
        By definition of $\equiv$
        \[
            \SAC \ndi{x(id_v)}{v}{id} \equiv \ndi{x(id'_v)}{v'}{id'}
        \]
        Then $ae' = \ndi{x(id'_v)}{v}{id'}$.

        \noindent{By assumption}
        \[
            ae \Rightarrow_r e
        \]
        By definition of $\Rightarrow_r$
        \[
            \ndi{x(id_v)}{v}{id} \Rightarrow_r x
        \]
        Then $e = x$.

        \noindent{By assumption}
        \[
            \SI{}{} e \Rightarrow_s \_, \_, ae'
        \]
        By definition of $\Rightarrow_s$
        \[
            \SI{}{} x \Rightarrow_s \_, \_, \ndi{x(id'_v)}{v'}{id'}
        \]
        Then $\sigma_v(x) = v', \sigma_{id}(x) = id'_v$, and $id'$ is a unique id.

        \noindent{Consider} $ae'_s = \ndi{x(id'_v)}{v'}{id'}$, $t'_s = \emptyset$, 
        $p_s = \emptyset$ and $e_s = x$. By definition of $\Rightarrow_{st}$
        \[
            \SAC \ndi{x(id_v)}{v}{id}, \ndi{x(id'_v)}{v'}{id'}, \emptyset \Rightarrow_{st} \ndi{x(id'_v)}{v'}{id'}
        \]
        Therefore 
        \[
            \SAC ae, ae'_s, t'_s \Rightarrow_{st} ae'
        \]
        
        \noindent{By} definition of $\Rightarrow_r$
        \[
            \emptyset \Rightarrow_r \emptyset
        \]
        Therefore
        \[
            t_s \Rightarrow_r p_s
        \]

        \noindent{By} definition of $\Rightarrow_r$
        \[
            \ndi{x(id_v)}{v}{id} \Rightarrow_r x
        \]
        Therefore
        \[
            ae_s \Rightarrow_r e_s
        \]

        \noindent{Because} $\sigma_v(x) = v', \sigma_{id}(x) = id'_v$, the definition of $\Rightarrow_s$ implies
        \[
            \SI{}{} \assume{z}{x} \Rightarrow_s \assume{z}{\ndi{x(id'_v)}{v'}{id'}}
        \]
        Therefore
        \[
            \SI{}{} p_s;\assume{z}{e_s} \Rightarrow_s t_s;\assume{z}{ae'_s}
        \]

        \noindent{Therefore when} $ae = \ndi{x(id_v)}{v}{id}$,
        \[
            \forall t_s, ae_s, ae'~\SAC ae \Rightarrow_{ex} ae_s, t_s~\wedge~\SAC ae \equiv ae'~\exists e.~ae \Rightarrow_r e~\wedge~\SI{}{} e \Rightarrow_s \_, \_, ae'
        \]
        implies
        \[
            \exists ae'_s, t'_s, p_s, e_s.~\SAC ae, ae'_s, t'_s \Rightarrow_{st} ae'~\wedge~t_s \Rightarrow_r p_s~\wedge~ae_s \Rightarrow_r e_s
        \]
        \[
            \wedge~\SI{}{} p_s;\assume{z}{e_s} \Rightarrow_s t'_s;\assume{z}{ae'_s}
        \]

         \noindent{Case 3:} $ae = \ndi{\lambda.x~e'}{v}{id}$

        \noindent{By assumption}
        \[
            \SAC ae \Rightarrow_{ex} ae_s, t_s
        \]
        By definition of $\Rightarrow_{ex}$ 
        \[
            \SAC \ndi{\lambda.x~e'}{v}{id} \Rightarrow_{ex} \ndi{\lambda.x~e'}{v}{id}, \emptyset
        \]
        Then $ae_s = \ndi{\lambda.x~e'}{v}{id}$ and $t_s = \emptyset$.

        \noindent{By assumption}
        \[
            \SAC ae \equiv ae'
        \]
        By definition of $\equiv$
        \[
            \SAC \ndi{\lambda.x~e'}{v}{id} \equiv \ndi{\lambda.x~e'}{v'}{id'}
        \]
        Then $ae' = \ndi{\lambda.x~e'}{v}{id'}$.

        \noindent{By assumption}
        \[
            ae \Rightarrow_r e
        \]
        By definition of $\Rightarrow_r$
        \[
            \ndi{\lambda.x~e'}{v}{id} \Rightarrow_r \lambda.x~e'
        \]
        Then $e = \lambda.x~e'$.

        \noindent{By assumption}
        \[
            \SI{}{} e \Rightarrow_s \_, \_, ae'
        \]
        By definition of $\Rightarrow_s$
        \[
            \SI{}{} \lambda.x~e' \Rightarrow_s \_, \_, \ndi{\lambda.x~e'}{v'}{id'}
        \]
        Then $id'$ is a unique id.

        \noindent{Consider} $ae'_s = \ndi{\lambda.x~e'}{v'}{id'}$, $t'_s = \emptyset$, 
        $p_s = \emptyset$ and $e_s = \lambda.x~e'$. By definition of $\Rightarrow_{st}$
        \[
            \SAC \ndi{\lambda.x~e'}{v}{id}, \ndi{\lambda.x~e'}{v'}{id'}, \emptyset \Rightarrow_{st} \ndi{\lambda.x~e'}{v'}{id'}
        \]
        Therefore 
        \[
            \SAC ae, ae'_s, t'_s \Rightarrow_{st} ae'
        \]
        
        \noindent{By} definition of $\Rightarrow_r$
        \[
            \emptyset \Rightarrow_r \emptyset
        \]
        Therefore
        \[
            t_s \Rightarrow_r p_s
        \]

        \noindent{By} definition of $\Rightarrow_r$
        \[
            \ndi{\lambda.x~e'}{v}{id} \Rightarrow_r \lambda.x~e'
        \]
        Therefore
        \[
            ae_s \Rightarrow_r e_s
        \]

        \noindent{Because} $\SI{}{} \lambda.x~e' \Rightarrow_s \_, \_, \ndi{\lambda.x~e'}{v'}{id'}$, the definition of $\Rightarrow_s$ implies
        \[
            \SI{}{} \assume{z}{x} \Rightarrow_s \assume{z}{\ndi{x(id'_v)}{v'}{id'}}
        \]
        Therefore
        \[
            \SI{}{} p_s;\assume{z}{e_s} \Rightarrow_s t_s;\assume{z}{ae'_s}
        \]

        \noindent{Therefore when} $ae = \ndi{\lambda.x~e'}{v}{id}$,
        \[
            \forall t_s, ae_s, ae'~\SAC ae \Rightarrow_{ex} ae_s, t_s~\wedge~\SAC ae \equiv ae'~\exists e.~ae \Rightarrow_r e~\wedge~\SI{}{} e \Rightarrow_s \_, \_, ae'
        \]
        implies
        \[
            \exists ae'_s, t'_s, p_s, e_s.~\SAC ae, ae'_s, t'_s \Rightarrow_{st} ae'~\wedge~t_s \Rightarrow_r p_s~\wedge~ae_s \Rightarrow_r e_s
        \]
        \[
            \wedge~\SI{}{} p_s;\assume{z}{e_s} \Rightarrow_s t'_s;\assume{z}{ae'_s}
        \]

        \noindent{\bf Induction Case:}

        \noindent{Case 1:} $ae = \ndi{\mathsf{Dist}(ae_1\#id_e) = ae_2}{v}{id}$ and $id_e \in \cS$

        \noindent{By assumption}
        \[
            \SAC ae \Rightarrow_{ex} ae_s, t_s
        \]
        By definition of $\Rightarrow_{ex}$
        \[
            \SAC \ndi{\mathsf{Dist}(ae_1\#id_e) = ae_2}{v}{id} \Rightarrow_{ex} \ndi{\mathsf{Dist}(ae^1_s\#id_e) = ae_2}{v}{id}, t_s
        \]
        Then $ae_s = \ndi{\mathsf{Dist}(ae^1_s\#id_e) = ae_2}{v}{id}$, $\SAC ae_1 \Rightarrow_{ex} ae^1_s, t^1_s$ and $t_s = t^1_s$.

        \noindent{By assumption}
        \[
            \SAC ae \equiv ae'
        \]
        By definition of $\equiv$
        \[
            \SAC \ndi{\mathsf{Dist}(ae_1\#id_e) = ae_2}{v}{id} \equiv \ndi{\mathsf{Dist}(ae'_1\#id'_e) = ae'_2}{v'}{id'}
        \]
        Then $ae' = \ndi{\mathsf{Dist}(ae'_1\#id'_e) = ae'_2}{v'}{id'}$ and $\SAC ae_1 \equiv ae'_1$.

        \noindent{By assumption}
        \[
            \exists e.~ae \Rightarrow_r e
        \]
        By definition of $\Rightarrow_r$
        \[
            \ndi{\mathsf{Dist}(ae_1\#id_e) = ae_2}{v}{id} \Rightarrow_r \mathsf{Dist}(e_1)
        \]
        Then $e = \mathsf{Dist}(e_1)$ and $ae_1 \Rightarrow_r e_1$.

        \noindent{By assumption}
        \[
            \SI{}{} e \Rightarrow_s \_, \_, ae'
        \]
        By definition of $\Rightarrow_s$
        \[
            \SI{}{} \mathsf{Dist}(e_1) \Rightarrow_s \_, \_, \ndi{\mathsf{Dist}(ae'_1\#id'_e)=ae'_2}{v'}{id'}
        \]
        Then $\SI{}{} e_1 \Rightarrow_s \_, \_, ae'_1$, $e_2 \in \dom \mathsf{Dist}(v')$, and $\SI{}{} e_2 \Rightarrow_s \_, \_, ae'_2$.

        \noindent{By induction hypothesis}
        \[
            \SAC ae_1 \Rightarrow_{ex} ae^1_s, t^1_s~\wedge~\SAC ae_1 \equiv ae'_1~
            \wedge~ae_1 \Rightarrow_r e_1~\wedge \SI{}{} e_1 \Rightarrow_s \_, \_, ae'_1
        \]
        \[\implies
            \SAC ae_1, ae^2_s, t^2_s \Rightarrow_{st} ae'_1~\wedge~t^1_s \Rightarrow_r p^1_s~\wedge~ae^1_s \Rightarrow_r e^1_s
        \]
        \[
            \wedge~\SI{}{} p^1_s;\assume{z}{e^1_s} \Rightarrow_s t^2_s;\assume{z}{ae^2_s}
        \]
        Because $\SAC ae_1 \Rightarrow_{ex} ae^1_s, t^1_s$, $\SAC ae_1 \equiv ae'_1$, $ae_1 \Rightarrow_r e_1$, and 
        $\SI{}{} e_1 \Rightarrow_s \_, \_, ae'_1$, 
        \[
            \SAC ae_1, ae^2_s, t^2_s \Rightarrow_{st} ae'_1~\wedge~t^1_s \Rightarrow_r p^1_s~\wedge~ae^1_s \Rightarrow_r e^1_s
        \]
        \[
            \wedge~\SI{}{} p^1_s;\assume{z}{e^1_s} \Rightarrow_s t^2_s;\assume{z}{ae^2_s}
        \]

        \noindent{Consider} $ae'_s = \ndi{\mathsf{Dist}(ae^2_s\#id_e)=ae'_2}{v'}{id'}$, $t'_s = t^2_s$, $p_s = p^1_s$, 
        and $e_s = \mathsf{Dist}(e^1_s)$. Because $ \SAC ae_1, ae^2_s, t^2_s \Rightarrow_{st} ae'_1$ and $id_e \in \cS$, the 
        definition of $\Rightarrow_{st}$ implies
        \[
            \begin{array}{c}
            \SAC  \ndi{\mathsf{Dist}(ae_1\#id_e) = ae_2}{v}{id},  \ndi{\mathsf{Dist}(ae^2_s\#id'_e) = ae'_2}{v'}{id'}, t'_s 
            \\\Rightarrow_{st}  \ndi{\mathsf{Dist}(ae'_1\#id'_e) = ae'_2}{v'}{id'} 
            \end{array}
        \]
        Therefore
        \[
            \SAC ae, ae'_s, t'_s \Rightarrow_{st} ae'
        \]

        \noindent{Because} $t^1_s \Rightarrow_r p^1_s$, 
        \[
            t_s \Rightarrow_r p_s
        \]

        \noindent{Because} $ ae^1_s \Rightarrow_r e^1_s$, the definition of $\Rightarrow_r$ implies
        \[
            \ndi{\mathsf{Dist}(ae^1_s\#id_e) = ae_2}{v}{id} \implies \mathsf{Dist}(e^1_s) 
        \]
        Therefore
        \[
            ae_s \Rightarrow_r e_s
        \]

        \noindent{Because} $\SI{}{} p^1_s;\assume{z}{e^1_s} 
        \Rightarrow_s t^2_s;\assume{z}{ae^2_s}$, $\SI{}{} e_2 \Rightarrow_s \_, \_, ae'_2$, 
        and all variable names introduced by $t^2_s$ do not conflict with variable names in $ae'_2$ (Observation \ref{ex:obs}),
        the definition of $\Rightarrow_s$ implies
        \[
            \SI{}{} p_s;\assume{z}{\mathsf{Dist}(e^1_s)} \Rightarrow_s t'_s;\assume{z}{\ndi{\mathsf{Dist}(ae^1_s\#id'_e) = ae'_2}{v'}{id'}}
        \]
        Therefore
        \[
            \SI{}{} p_s;\assume{z}{e_s} \Rightarrow_s t'_s;\assume{z}{ae'_s}
        \]

        \noindent{Therefore} when $ae =  \ndi{\mathsf{Dist}(ae_1\#id_e) = ae_2}{v}{id}$ and  $id_e \in \cS$, and assuming induction hypothesis,
        \[
            \forall t_s, ae_s, ae'~\SAC ae \Rightarrow_{ex} ae_s, t_s~\wedge~\SAC ae \equiv ae'~\exists e.~ae \Rightarrow_r e~\wedge~\SI{}{} e \Rightarrow_s \_, \_, ae'
        \]
        \[
          \implies  \exists ae'_s, t'_s, p_s, e_s.~\SAC ae, ae'_s, t'_s \Rightarrow_{st} ae'~\wedge~t_s \Rightarrow_r p_s~\wedge~ae_s \Rightarrow_r e_s
        \]
        \[
            \wedge~\SI{}{} p_s;\assume{z}{e_s} \Rightarrow_s t'_s;\assume{z}{ae'_s}
        \]

        \noindent{Case 2:} $ae = \ndi{\mathsf{Dist}(ae_1\#id_e) = ae_2}{v}{id}$ and $id_e \notin \cS$

        \noindent{By assumption}
        \[
            \SAC ae \Rightarrow_{ex} ae_s, t_s
        \]
        By definition of $\Rightarrow_{ex}$
        \[
            \SAC \ndi{\mathsf{Dist}(ae_1\#id_e) = ae_2}{v}{id} \Rightarrow_{ex} \ndi{ae^3_s}{v}{id}, t_s
        \]
        Then $ae_s = \ndi{ae^3_s}{v}{id}$, $\SAC ae_1 \Rightarrow_{ex} ae^1_s, t^1_s$, 
        $\SAC ae_2 \Rightarrow_{ex} ae^3_s, t^3_s$, $ae_2 \Rightarrow_r e_2$, and $t_s = t^1_s;\observed{ae^1_s}{e_2} ;t^3_s$.

        \noindent{By assumption}
        \[
            \SAC ae \equiv ae'
        \]
        By definition of $\equiv$
        \[
            \SAC \ndi{\mathsf{Dist}(ae_1\#id_e) = ae_2}{v}{id} \equiv \ndi{\mathsf{Dist}(ae'_1\#id_e) = ae'_2}{v'}{id'}
        \]
        Then $ae' = \ndi{\mathsf{Dist}(ae'_1\#id_e) = ae'_2}{v'}{id'}$, $\SAC ae_1 \equiv ae'_1$, $\SAC ae_2 \equiv ae'_2$.

        \noindent{By assumption}
        \[
            \exists e.~ae \Rightarrow_r e
        \]
        By definition of $\Rightarrow_r$
        \[
            \ndi{\mathsf{Dist}(ae_1\#id_e) = ae_2}{v}{id} \Rightarrow_r \mathsf{Dist}(e_1)
        \]
        Then $e = \mathsf{Dist}(e_1)$ and $ae_1 \Rightarrow_r e_1$.

        \noindent{By assumption}
        \[
            \SI{}{} e \Rightarrow_s \_, \_, ae'
        \]
        By definition of $\Rightarrow_s$
        \[
            \SI{}{} \mathsf{Dist}(e_1) \Rightarrow_s \_, \_, \ndi{\mathsf{Dist}(ae'_1\#id_e)=ae'_2}{v'}{id'}
        \]
        Then $\SI{}{} e_1 \Rightarrow_s \_, \_, ae'_1$, $e_2 \in \dom \mathsf{Dist}(v')$, and $\SI{}{} e_2 \Rightarrow_s \_, \_, ae'_2$.

        \noindent{By induction hypothesis}
        \[
            \SAC ae_1 \Rightarrow_{ex} ae^1_s, t^1_s~\wedge~\SAC ae_1 \equiv ae'_1~
            \wedge~ae_1 \Rightarrow_r e_1~\wedge \SI{}{} e_1 \Rightarrow_s \_, \_, ae'_1
        \]
        \[\implies
            \SAC ae_1, ae^2_s, t^2_s \Rightarrow_{st} ae'_1~\wedge~t^1_s \Rightarrow_r p^1_s~\wedge~ae^1_s \Rightarrow_r e^1_s
        \]
        \[
            \wedge~\SI{}{} p^1_s;\assume{z}{e^1_s} \Rightarrow_s t^2_s;\assume{z}{ae^2_s}
        \]
        Because $\SAC ae_1 \Rightarrow_{ex} ae^1_s, t^1_s$, $\SAC ae_1 \equiv ae'_1$, $ae_1 \Rightarrow_r e_1$, and 
        $\SI{}{} e_1 \Rightarrow_s \_, \_, ae'_1$, 
        \[
            \SAC ae_1, ae^2_s, t^2_s \Rightarrow_{st} ae'_1~\wedge~t^1_s \Rightarrow_r p^1_s~\wedge~ae^1_s \Rightarrow_r e^1_s
        \]
        \[
            \wedge~\SI{}{} p^1_s;\assume{z}{e^1_s} \Rightarrow_s t^2_s;\assume{z}{ae^2_s}
        \]

        \noindent{By induction hypothesis}
        \[
            \SAC ae_2 \Rightarrow_{ex} ae^3_s, t^3_s~\wedge~\SAC ae_2 \equiv ae'_2~
            \wedge~ae_2 \Rightarrow_r e_2~\wedge \SI{}{} e_2 \Rightarrow_s \_, \_, ae'_2
        \]
        \[\implies
            \SAC ae_2, ae^4_s, t^4_s \Rightarrow_{st} ae'_2~\wedge~t^3_s \Rightarrow_r p^2_s~\wedge~ae^3_s \Rightarrow_r e^2_s
        \]
        \[
            \wedge~\SI{}{} p^2_s;\assume{z}{e^2_s} \Rightarrow_s t^4_s;\assume{z}{ae^4_s}
        \]
        Because $\SAC ae_2 \Rightarrow_{ex} ae^3_s, t^3_s$, $\SAC ae_2 \equiv ae'_2$, $ae_2 \Rightarrow_r e_2$, and 
        $\SI{}{} e_2 \Rightarrow_s \_, \_, ae'_2$, 
        \[
            \SAC ae_2, ae^4_s, t^4_s \Rightarrow_{st} ae'_2~\wedge~t^3_s \Rightarrow_r p^2_s~\wedge~ae^3_s \Rightarrow_r e^2_s
        \]
        \[
            \wedge~\SI{}{} p^2_s;\assume{z}{e^2_s} \Rightarrow_s t^4_s;\assume{z}{ae^4_s}
        \]

        \noindent{Consider} $ae'_s = \ndi{ae^4_s}{v'}{id'}$, $t'_s = t^2_s;\observed{ae^2_s}{e_2};t^4_s$, $p_s = p^1_s;\observed{e^1_s}{e_2};p^2_s$, 
        and $e_s = e^2_s$. Because $ \SAC ae_1, ae^2_s, t^2_s \Rightarrow_{st} 
        ae'_1$, $\SAC ae_2, ae^4_s, t^4_s \Rightarrow_{st} ae'_s$, and $id_e \in \cS$, the 
        definition of $\Rightarrow_{st}$ implies
        \[
            \SAC  \ndi{\mathsf{Dist}(ae_1\#id_e) = ae_2}{v}{id},  \ndi{ae^4_s}{v'}{id'}, t'_s 
            \Rightarrow_{st}  \ndi{\mathsf{Dist}(ae'_1\#id_e) = ae'_2}{v'}{id'} 
        \]
        Therefore
        \[
            \SAC ae, ae'_s, t'_s \Rightarrow_{st} ae'
        \]

        \noindent{Because} $t^1_s \Rightarrow_r p^1_s$, $t^3_s \Rightarrow_r p^2_s$, $ae^1_s \Rightarrow_r e^1_s$,
        definition of $\Rightarrow_r$ implies
        \[
            t^1_s;\observed{ae^1_s}{e_2};t^3_s \Rightarrow_r p^1_s;\observed{e^1_s}{e_2};p^2_s
        \]
        Therefore
        \[
            t_s \Rightarrow_r p_s
        \]

        \noindent{Because} $ ae^3_s \Rightarrow_r e^2_s$, the definition of $\Rightarrow_r$ implies
        \[
            \ndi{ae^3_s}{v}{id} \implies e^2_s
        \]
        Therefore
        \[
            ae_s \Rightarrow_r e_s
        \]

        \noindent{Because} $\SI{}{} p^1_s;\assume{z}{e^1_s} 
        \Rightarrow_s t^2_s;\assume{z}{ae^2_s}$, $\SI{}{} p^2_s;\assume{z}{e^2_s} \Rightarrow_s t^4_s;\assume{z}{ae^4_s}$, 
        and all variable names introduced by $t^2_s$ do not conflict with variable names in $ae^4_s$ and $t^4_s$ (Observation \ref{ex:obs}),
        the definition of $\Rightarrow_s$ implies
        \[
            \SI{}{} p_s;\assume{z}{e^2_s} \Rightarrow_s t'_s;\assume{z}{ae^4_s}
        \]
        Therefore
        \[
            \SI{}{} p_s;\assume{z}{e_s} \Rightarrow_s t'_s;\assume{z}{ae'_s}
        \]

        \noindent{Therefore} when $ae =  \ndi{\mathsf{Dist}(ae_1\$id_e) = ae_2}{v}{id}$ and $id_e \in \cS$, and assuming induction hypothesis,
        \[
            \forall t_s, ae_s, ae'~\SAC ae \Rightarrow_{ex} ae_s, t_s~\wedge~
            \SAC ae \equiv ae'~\exists e.~ae \Rightarrow_r e~\wedge~\SI{}{} e \Rightarrow_s \_, \_, ae'
        \]
        \[
          \implies  \exists ae'_s, t'_s, p_s, e_s.~\SAC ae, ae'_s, t'_s \Rightarrow_{st} ae'~\wedge~t_s \Rightarrow_r p_s~\wedge~ae_s \Rightarrow_r e_s
        \]
        \[
            \wedge~\SI{}{} p_s;\assume{z}{e_s} \Rightarrow_s t'_s;\assume{z}{ae'_s}
        \]

        \noindent{Case 3:} $ae = \ndi{(ae_1~ae_2)aa}{v}{id}$ and $ID(ae_1) \in \cS$

        \noindent{By assumption}
        \[
            \SAC ae \Rightarrow_{ex} ae_s, t_s
        \]
        By definition of $\Rightarrow_{ex}$
        \[
            \SAC \ndi{(ae_1~ae_2)aa}{v}{id} \Rightarrow_{ex} \ndi{(ae^1_s~ae^3_s)aa}{v}{id}, t_s
        \]
        Then $ae_s = \ndi{(ae^1_s~ae^3_s)aa}{v}{id}$, $\SAC ae_1 \Rightarrow_{ex} ae^1_s, t^1_s$, 
        $\SAC ae_2 \Rightarrow_{ex} ae^3_s, t^3_s$, and $t_s = t^1_s;t^3_s$.

        \noindent{By assumption}
        \[
            \SAC ae \equiv ae'
        \]
        By definition of $\equiv$
        \[
            \SAC \ndi{(ae_1~ae_2)aa}{v}{id} \equiv \ndi{(ae'_1~ae'_2)aa'}{v'}{id'}
        \]
        Then $ae' = \ndi{(ae'_1~ae'_2)aa'}{v'}{id'}$, $\SAC ae_1 \equiv ae'_1$, $\SAC ae_2 \equiv ae'_2$.

        \noindent{By assumption}
        \[
            \exists e.~ae \Rightarrow_r e
        \]
        By definition of $\Rightarrow_r$
        \[
            \ndi{(ae_1~ae_2)aa}{v}{id} \Rightarrow_r (e_1~e_2)
        \]
        Then $e = (e_1~e_2)$, $ae_1 \Rightarrow_r e_1$, and $ae_2 \Rightarrow_r e_2$.

        \noindent{By assumption}
        \[
            \SI{}{} e \Rightarrow_s \_, \_, ae'
        \]
        By definition of $\Rightarrow_s$
        \[
            \SI{}{} (e_1~e_2) \Rightarrow_s \_, \_, \ndi{(ae'_1~ae'_2)aa'}{v'}{id'}
        \]
        Then $\SI{}{} e_1 \Rightarrow_s \_, \_, ae'_1$ and $\SI{}{} e_2 \Rightarrow_s \_, \_, ae'_2$.
        
        \noindent{If} $\cV(ae'_1) = \tup{\lambda.x~e', \sigma'_v, \sigma'_{id}}$, 
        $\SI{'[y \rightarrow \cV(ae'_2)]}{'[y \rightarrow ID(ae'_2)]} e'[y/x] \Rightarrow_s \_, \_, ae_e$, $aa = y=ae_e$,
        else
        $aa = \perp$.

        \noindent{By induction hypothesis}
        \[
            \SAC ae_1 \Rightarrow_{ex} ae^1_s, t^1_s~\wedge~\SAC ae_1 \equiv ae'_1~
            \wedge~ae_1 \Rightarrow_r e_1~\wedge \SI{}{} e_1 \Rightarrow_s \_, \_, ae'_1
        \]
        \[\implies
            \SAC ae_1, ae^2_s, t^2_s \Rightarrow_{st} ae'_1~\wedge~t^1_s \Rightarrow_r p^1_s~\wedge~ae^1_s \Rightarrow_r e^1_s
        \]
        \[
            \wedge~\SI{}{} p^1_s;\assume{z}{e^1_s} \Rightarrow_s t^2_s;\assume{z}{ae^2_s}
        \]
        Because $\SAC ae_1 \Rightarrow_{ex} ae^1_s, t^1_s$, $\SAC ae_1 \equiv ae'_1$, $ae_1 \Rightarrow_r e_1$, and 
        $\SI{}{} e_1 \Rightarrow_s \_, \_, ae'_1$, 
        \[
            \SAC ae_1, ae^2_s, t^2_s \Rightarrow_{st} ae'_1~\wedge~t^1_s \Rightarrow_r p^1_s~\wedge~ae^1_s \Rightarrow_r e^1_s
        \]
        \[
            \wedge~\SI{}{} p^1_s;\assume{z}{e^1_s} \Rightarrow_s t^2_s;\assume{z}{ae^2_s}
        \]

        \noindent{By induction hypothesis}
        \[
            \SAC ae_2 \Rightarrow_{ex} ae^3_s, t^3_s~\wedge~\SAC ae_2 \equiv ae'_2~
            \wedge~ae_2 \Rightarrow_r e_2~\wedge \SI{}{} e_2 \Rightarrow_s \_, \_, ae'_2
        \]
        \[\implies
            \SAC ae_2, ae^4_s, t^4_s \Rightarrow_{st} ae'_2~\wedge~t^3_s \Rightarrow_r p^2_s~\wedge~ae^3_s \Rightarrow_r e^2_s
        \]
        \[
            \wedge~\SI{}{} p^2_s;\assume{z}{e^2_s} \Rightarrow_s t^4_s;\assume{z}{ae^4_s}
        \]
        Because $\SAC ae_2 \Rightarrow_{ex} ae^3_s, t^3_s$, $\SAC ae_2 \equiv ae'_2$, $ae_2 \Rightarrow_r e_2$, and 
        $\SI{}{} e_2 \Rightarrow_s \_, \_, ae'_2$, 
        \[
            \SAC ae_2, ae^4_s, t^4_s \Rightarrow_{st} ae'_2~\wedge~t^3_s \Rightarrow_r p^2_s~\wedge~ae^3_s \Rightarrow_r e^2_s
        \]
        \[
            \wedge~\SI{}{} p^2_s;\assume{z}{e^2_s} \Rightarrow_s t^4_s;\assume{z}{ae^4_s}
        \]

        \noindent{Consider} $ae'_s = \ndi{(ae^2_s~ae^4_s)aa'}{v'}{id'}$, 
        $t'_s = t^2_s;t^4_s$, $p_s = p^1_s;p^2_s$, 
        and $e_s = (e^1_s~e^2_s)$. Because $ \SAC ae_1, ae^2_s, t^2_s \Rightarrow_{st} 
        ae'_1$, $\SAC ae_2, ae^4_s, t^4_s \Rightarrow_{st} ae'_2$, and $id_e \in \cS$, the 
        definition of $\Rightarrow_{st}$ implies
        \[
            \SAC  \ndi{(ae_1~ae_2)aa}{v}{id},  \ndi{(ae^2_s~ae^4_s)aa'}{v'}{id'}, t'_s 
            \Rightarrow_{st}  \ndi{(ae'_1~ae'_2)aa'}{v'}{id'} 
        \]
        Therefore
        \[
            \SAC ae, ae'_s, t'_s \Rightarrow_{st} ae'
        \]

        \noindent{Because} $t^1_s \Rightarrow_r p^1_s$ and $t^3_s \Rightarrow_r p^2_s$,
        definition of $\Rightarrow_r$ implies
        \[
            t^1_s;t^3_s \Rightarrow_r p^1_s;p^2_s
        \]
        Therefore
        \[
            t_s \Rightarrow_r p_s
        \]

        \noindent{Because} $ ae^3_s \Rightarrow_r e^2_s$ and $ae^1_s \Rightarrow_r e^1_s$, the definition of $\Rightarrow_r$ implies
        \[
            \ndi{(ae^1_s~ae^3_s)aa}{v}{id} \implies (e^1_s~e^2_s)
        \]
        Therefore
        \[
            ae_s \Rightarrow_r e_s
        \]

        \noindent{Because} $\SI{}{} p^1_s;\assume{z}{e^1_s} 
        \Rightarrow_s t^2_s;\assume{z}{ae^2_s}$, $\SI{}{} p^2_s;\assume{z}{e^2_s} \Rightarrow_s t^4_s;\assume{z}{ae^4_s}$, 
        and all variable names introduced by $t^2_s$ do not conflict with variable names in $ae^4_s$ and $t^4_s$ (Observation \ref{ex:obs}),
        the definition of $\Rightarrow_s$ implies
        \[
            \SI{}{} p_s;\assume{z}{(e^1_s~e^2_s)} \Rightarrow_s t'_s;\assume{z}{(ae^2_s~ae^4_s)aa'}
        \]
        Therefore
        \[
            \SI{}{} p_s;\assume{z}{e_s} \Rightarrow_s t'_s;\assume{z}{ae'_s}
        \]

        \noindent{Therefore} when $ae =  \ndi{(ae_1~ae_2)aa}{v}{id}$ and $id_e \in \cS$, and assuming induction hypothesis,
        \[
            \forall t_s, ae_s, ae'~\SAC ae \Rightarrow_{ex} ae_s, t_s~\wedge~
            \SAC ae \equiv ae'~\exists e.~ae \Rightarrow_r e~\wedge~\SI{}{} e \Rightarrow_s \_, \_, ae'
        \]
        \[
          \implies  \exists ae'_s, t'_s, p_s, e_s.~\SAC ae, ae'_s, t'_s \Rightarrow_{st} ae'~\wedge~t_s \Rightarrow_r p_s~\wedge~ae_s \Rightarrow_r e_s
        \]
        \[
            \wedge~\SI{}{} p_s;\assume{z}{e_s} \Rightarrow_s t'_s;\assume{z}{ae'_s}
        \]

        \noindent{Case 4:} $ae = \ndi{(ae_1~ae_2)\perp}{v}{id}$ and $ID(ae_1) \notin \cS$

        \noindent{By assumption}
        \[
            \SAC ae \Rightarrow_{ex} ae_s, t_s
        \]
        By definition of $\Rightarrow_{ex}$
        \[
            \SAC \ndi{(ae_1~ae_2)\perp}{v}{id} \Rightarrow_{ex} \ndi{(ae^1_s~ae^3_s)\perp}{v}{id}, t_s
        \]
        Then $ae_s = \ndi{(ae^1_s~ae^3_s)\perp}{v}{id}$, $\SAC ae_1 \Rightarrow_{ex} ae^1_s, t^1_s$, 
        $\SAC ae_2 \Rightarrow_{ex} ae^3_s, t^3_s$, and $t_s = t^1_s;t^3_s$.

        \noindent{By assumption}
        \[
            \SAC ae \equiv ae'
        \]
        By definition of $\equiv$
        \[
            \SAC \ndi{(ae_1~ae_2)\perp}{v}{id} \equiv \ndi{(ae'_1~ae'_2)\perp}{v'}{id'}
        \]
        Then $ae' = \ndi{(ae'_1~ae'_2)\perp}{v'}{id'}$, $\SAC ae_1 \equiv ae'_1$, $\SAC ae_2 \equiv ae'_2$.

        \noindent{By assumption}
        \[
            \exists e.~ae \Rightarrow_r e
        \]
        By definition of $\Rightarrow_r$
        \[
            \ndi{(ae_1~ae_2)\perp}{v}{id} \Rightarrow_r (e_1~e_2)
        \]
        Then $e = (e_1~e_2)$, $ae_1 \Rightarrow_r e_1$, and $ae_2 \Rightarrow_r e_2$.

        \noindent{By assumption}
        \[
            \SI{}{} e \Rightarrow_s \_, \_, ae'
        \]
        By definition of $\Rightarrow_s$
        \[
            \SI{}{} (e_1~e_2) \Rightarrow_s \_, \_, \ndi{(ae'_1~ae'_2)\perp}{v'}{id'}
        \]
        Then $\SI{}{} e_1 \Rightarrow_s \_, \_, ae'_1$ and $\SI{}{} e_2 \Rightarrow_s \_, \_, ae'_2$.

        \noindent{By induction hypothesis}
        \[
            \SAC ae_1 \Rightarrow_{ex} ae^1_s, t^1_s~\wedge~\SAC ae_1 \equiv ae'_1~
            \wedge~ae_1 \Rightarrow_r e_1~\wedge \SI{}{} e_1 \Rightarrow_s \_, \_, ae'_1
        \]
        \[\implies
            \SAC ae_1, ae^2_s, t^2_s \Rightarrow_{st} ae'_1~\wedge~t^1_s \Rightarrow_r p^1_s~\wedge~ae^1_s \Rightarrow_r e^1_s
        \]
        \[
            \wedge~\SI{}{} p^1_s;\assume{z}{e^1_s} \Rightarrow_s t^2_s;\assume{z}{ae^2_s}
        \]
        Because $\SAC ae_1 \Rightarrow_{ex} ae^1_s, t^1_s$, $\SAC ae_1 \equiv ae'_1$, $ae_1 \Rightarrow_r e_1$, and 
        $\SI{}{} e_1 \Rightarrow_s \_, \_, ae'_1$, 
        \[
            \SAC ae_1, ae^2_s, t^2_s \Rightarrow_{st} ae'_1~\wedge~t^1_s \Rightarrow_r p^1_s~\wedge~ae^1_s \Rightarrow_r e^1_s
        \]
        \[
            \wedge~\SI{}{} p^1_s;\assume{z}{e^1_s} \Rightarrow_s t^2_s;\assume{z}{ae^2_s}
        \]

        \noindent{By induction hypothesis}
        \[
            \SAC ae_2 \Rightarrow_{ex} ae^3_s, t^3_s~\wedge~\SAC ae_2 \equiv ae'_2~
            \wedge~ae_2 \Rightarrow_r e_2~\wedge \SI{}{} e_2 \Rightarrow_s \_, \_, ae'_2
        \]
        \[\implies
            \SAC ae_2, ae^4_s, t^4_s \Rightarrow_{st} ae'_2~\wedge~t^3_s \Rightarrow_r p^2_s~\wedge~ae^3_s \Rightarrow_r e^2_s
        \]
        \[
            \wedge~\SI{}{} p^2_s;\assume{z}{e^2_s} \Rightarrow_s t^4_s;\assume{z}{ae^4_s}
        \]
        Because $\SAC ae_2 \Rightarrow_{ex} ae^3_s, t^3_s$, $\SAC ae_2 \equiv ae'_2$, $ae_2 \Rightarrow_r e_2$, and 
        $\SI{}{} e_2 \Rightarrow_s \_, \_, ae'_2$, 
        \[
            \SAC ae_2, ae^4_s, t^4_s \Rightarrow_{st} ae'_2~\wedge~t^3_s \Rightarrow_r p^2_s~\wedge~ae^3_s \Rightarrow_r e^2_s
        \]
        \[
            \wedge~\SI{}{} p^2_s;\assume{z}{e^2_s} \Rightarrow_s t^4_s;\assume{z}{ae^4_s}
        \]

        \noindent{Consider} $ae'_s = \ndi{(ae^2_s~ae^4_s)\perp}{v'}{id'}$, 
        $t'_s = t^2_s;t^4_s$, $p_s = p^1_s;p^2_s$, 
        and $e_s = (e^1_s~e^2_s)$. Because $ \SAC ae_1, ae^2_s, t^2_s \Rightarrow_{st} 
        ae'_1$, $\SAC ae_2, ae^4_s, t^4_s \Rightarrow_{st} ae'_2$, and $id_e \in \cS$, the 
        definition of $\Rightarrow_{st}$ implies
        \[
            \SAC  \ndi{(ae_1~ae_2)\perp}{v}{id},  \ndi{(ae^2_s~ae^4_s)\perp}{v'}{id'}, t'_s 
            \Rightarrow_{st}  \ndi{(ae'_1~ae'_2)\perp}{v'}{id'} 
        \]
        Therefore
        \[
            \SAC ae, ae'_s, t'_s \Rightarrow_{st} ae'
        \]

        \noindent{Because} $t^1_s \Rightarrow_r p^1_s$ and $t^3_s \Rightarrow_r p^2_s$,
        definition of $\Rightarrow_r$ implies
        \[
            t^1_s;t^3_s \Rightarrow_r p^1_s;p^2_s
        \]
        Therefore
        \[
            t_s \Rightarrow_r p_s
        \]

        \noindent{Because} $ ae^3_s \Rightarrow_r e^2_s$ and $ae^1_s \Rightarrow_r e^1_s$, the definition of $\Rightarrow_r$ implies
        \[
            \ndi{(ae^1_s~ae^3_s)\perp}{v}{id} \implies (e^1_s~e^2_s)
        \]
        Therefore
        \[
            ae_s \Rightarrow_r e_s
        \]

        \noindent{Because} $\SI{}{} p^1_s;\assume{z}{e^1_s} 
        \Rightarrow_s t^2_s;\assume{z}{ae^2_s}$, $\SI{}{} p^2_s;\assume{z}{e^2_s} \Rightarrow_s t^4_s;\assume{z}{ae^4_s}$, 
        and all variable names introduced by $t^2_s$ do not conflict with variable names in $ae^4_s$ and $t^4_s$ (Observation \ref{ex:obs}),
        the definition of $\Rightarrow_s$ implies
        \[
            \SI{}{} p_s;\assume{z}{(e^1_s~e^2_s)} \Rightarrow_s t'_s;\assume{z}{(ae^2_s~ae^4_s)\perp}
        \]
        Therefore
        \[
            \SI{}{} p_s;\assume{z}{e_s} \Rightarrow_s t'_s;\assume{z}{ae'_s}
        \]

        \noindent{Therefore} when $ae =  \ndi{(ae_1~ae_2)\perp}{v}{id}$ and $id_e \notin \cS$, and assuming induction hypothesis,
        \[
            \forall t_s, ae_s, ae'~\SAC ae \Rightarrow_{ex} ae_s, t_s~\wedge~
            \SAC ae \equiv ae'~\exists e.~ae \Rightarrow_r e~\wedge~\SI{}{} e \Rightarrow_s \_, \_, ae'
        \]
        \[
          \implies  \exists ae'_s, t'_s, p_s, e_s.~\SAC ae, ae'_s, t'_s \Rightarrow_{st} ae'~\wedge~t_s \Rightarrow_r p_s~\wedge~ae_s \Rightarrow_r e_s
        \]
        \[
            \wedge~\SI{}{} p_s;\assume{z}{e_s} \Rightarrow_s t'_s;\assume{z}{ae'_s}
        \]

        \noindent{Case 5:} $ae = \ndi{(ae_1~ae_2)y=ae_3}{v}{id}$ and $ID(ae_1) \notin \cS$

        \noindent{By assumption}
        \[
            \SAC ae \Rightarrow_{ex} ae_s, t_s
        \]
        By definition of $\Rightarrow_{ex}$
        \[
            \SAC \ndi{(ae_1~ae_2)y=ae_3}{v}{id} \Rightarrow_{ex} \ndi{ae^5_s}{v}{id}, t_s
        \]
        Then $ae_s = \ndi{ae^5_s}{v}{id}$, $\SAC ae_1 \Rightarrow_{ex} ae^1_s, t^1_s$, 
        $\SAC ae_2 \Rightarrow_{ex} ae^3_s, t^3_s$, $\SAC ae_3 \rightarrow_{ex} ae^5_s, t^5_s$, 
        and $t_s = t^1_s;\assume{x}{ae^1_s};t^3_s;\assume{y}{ae^3_s};t^5_s$.

        \noindent{By assumption}
        \[
            \SAC ae \equiv ae'
        \]
        By definition of $\equiv$
        \[
            \SAC \ndi{(ae_1~ae_2)y=ae_3}{v}{id} \equiv \ndi{(ae'_1~ae'_2)y=ae'_3}{v'}{id'}
        \]
        Then $ae' = \ndi{(ae'_1~ae'_2)y=ae'_3}{v'}{id'}$, $\SAC ae_1 \equiv ae'_1$, $\SAC ae_2 \equiv ae'_2$, 
        and $\SAC ae_3 \equiv ae'_3$.

        \noindent{By assumption}
        \[
            \exists e.~ae \Rightarrow_r e
        \]
        By definition of $\Rightarrow_r$
        \[
            \ndi{(ae_1~ae_2)y=ae_3}{v}{id} \Rightarrow_r (e_1~e_2)
        \]
        Then $e = (e_1~e_2)$, $ae_1 \Rightarrow_r e_1$, and $ae_2 \Rightarrow_r e_2$.

        \noindent{By assumption}
        \[
            \SI{}{} e \Rightarrow_s \_, \_, ae'
        \]
        By definition of $\Rightarrow_s$
        \[
            \SI{}{} (e_1~e_2) \Rightarrow_s \_, \_, \ndi{(ae'_1~ae'_2)y=ae'_3}{v'}{id'}
        \]
        Then $\SI{}{} e_1 \Rightarrow_s \_, \_, ae'_1$ and $\SI{}{} e_2 \Rightarrow_s \_, \_, ae'_2$, 
        $\cV(ae'_1) = \tup{\lambda.x~e'_3, \sigma'_v, \sigma'_{id}}$, $e_3 = e'_3[y/x]$, and
        $\SI{'[y \rightarrow \cV(ae'_2)]}{'[y \rightarrow ID(ae'_2)]} e_3 \Rightarrow_s \_, \_, ae'_3$.

        \noindent{By induction hypothesis}
        \[
            \SAC ae_1 \Rightarrow_{ex} ae^1_s, t^1_s~\wedge~\SAC ae_1 \equiv ae'_1~
            \wedge~ae_1 \Rightarrow_r e_1~\wedge \SI{}{} e_1 \Rightarrow_s \_, \_, ae'_1
        \]
        \[\implies
            \SAC ae_1, ae^2_s, t^2_s \Rightarrow_{st} ae'_1~\wedge~t^1_s \Rightarrow_r p^1_s~\wedge~ae^1_s \Rightarrow_r e^1_s
        \]
        \[
            \wedge~\SI{}{} p^1_s;\assume{z}{e^1_s} \Rightarrow_s t^2_s;\assume{z}{ae^2_s}
        \]
        Because $\SAC ae_1 \Rightarrow_{ex} ae^1_s, t^1_s$, $\SAC ae_1 \equiv ae'_1$, $ae_1 \Rightarrow_r e_1$, and 
        $\SI{}{} e_1 \Rightarrow_s \_, \_, ae'_1$, 
        \[
            \SAC ae_1, ae^2_s, t^2_s \Rightarrow_{st} ae'_1~\wedge~t^1_s \Rightarrow_r p^1_s~\wedge~ae^1_s \Rightarrow_r e^1_s
        \]
        \[
            \wedge~\SI{}{} p^1_s;\assume{z}{e^1_s} \Rightarrow_s t^2_s;\assume{z}{ae^2_s}
        \]

        \noindent{By induction hypothesis}
        \[
            \SAC ae_2 \Rightarrow_{ex} ae^3_s, t^3_s~\wedge~\SAC ae_2 \equiv ae'_2~
            \wedge~ae_2 \Rightarrow_r e_2~\wedge \SI{}{} e_2 \Rightarrow_s \_, \_, ae'_2
        \]
        \[\implies
            \SAC ae_2, ae^4_s, t^4_s \Rightarrow_{st} ae'_2~\wedge~t^3_s \Rightarrow_r p^2_s~\wedge~ae^3_s \Rightarrow_r e^2_s
        \]
        \[
            \wedge~\SI{}{} p^2_s;\assume{z}{e^2_s} \Rightarrow_s t^4_s;\assume{z}{ae^4_s}
        \]
        Because $\SAC ae_2 \Rightarrow_{ex} ae^3_s, t^3_s$, $\SAC ae_2 \equiv ae'_2$, $ae_2 \Rightarrow_r e_2$, and 
        $\SI{}{} e_2 \Rightarrow_s \_, \_, ae'_2$, 
        \[
            \SAC ae_2, ae^4_s, t^4_s \Rightarrow_{st} ae'_2~\wedge~t^3_s \Rightarrow_r p^2_s~\wedge~ae^3_s \Rightarrow_r e^2_s
        \]
        \[
            \wedge~\SI{}{} p^2_s;\assume{z}{e^2_s} \Rightarrow_s t^4_s;\assume{z}{ae^4_s}
        \]

        \noindent{By induction hypothesis}
        \[
            \SAC ae_3 \Rightarrow_{ex} ae^5_s, t^5_s~\wedge~\SAC ae_3 \equiv ae'_3~
            \wedge~ae_3 \Rightarrow_r e_3~\wedge \SI{}{} e_3 \Rightarrow_s \_, \_, ae'_3
        \]
        \[\implies
            \SAC ae_3, ae^6_s, t^6_s \Rightarrow_{st} ae'_3~\wedge~t^5_s \Rightarrow_r p^3_s~\wedge~ae^5_s \Rightarrow_r e^3_s
        \]
        \[
            \wedge~\SI{}{} p^3_s;\assume{z}{e^3_s} \Rightarrow_s t^6_s;\assume{z}{ae^6_s}
        \]
        Because $\SAC ae_3 \Rightarrow_{ex} ae^5_s, t^5_s$, $\SAC ae_3 \equiv ae'_3$, $ae_3 \Rightarrow_r e_3$, and 
        $\SI{}{} e_3 \Rightarrow_s \_, \_, ae'_3$, 
        \[
            \SAC ae_3, ae^6_s, t^6_s \Rightarrow_{st} ae'_3~\wedge~t^5_s \Rightarrow_r p^3_s~\wedge~ae^5_s \Rightarrow_r e^3_s
        \]
        \[
            \wedge~\SI{}{} p^3_s;\assume{z}{e^3_s} \Rightarrow_s t^6_s;\assume{z}{ae^6_s}
        \]

        \noindent{Consider} $ae'_s = \ndi{ae^6_s}{v'}{id'}$, 
        $t'_s = t^2_s;\assume{x}{ae^2_s};t^4_s;\assume{y}{ae^4_s};t^6_s$, $p_s = p^1_s;\assume{x}{e^1_s};p^2_s;\assume{y}{e^2_s};p^3_s$, 
        and $e_s = e^3_s$. Because $ \SAC ae_1, ae^2_s, t^2_s \Rightarrow_{st} 
        ae'_1$, $\SAC ae_2, ae^4_s, t^4_s \Rightarrow_{st} ae'_2$, $\SAC ae_3, ae^6_s, t^6_s \Rightarrow_{st} ae'_3$, the 
        definition of $\Rightarrow_{st}$ implies
        \[
            \SAC  \ndi{(ae_1~ae_2)y=ae_3}{v}{id},  \ndi{ae^6_s}{v'}{id'}, t'_s 
            \Rightarrow_{st}  \ndi{(ae'_1~ae'_2)y=ae'_3}{v'}{id'} 
        \]
        Therefore
        \[
            \SAC ae, ae'_s, t'_s \Rightarrow_{st} ae'
        \]

        \noindent{Because} $t^1_s \Rightarrow_r p^1_s$, $t^3_s \Rightarrow_r p^2_s$,
        $t^5_s \Rightarrow_r p^3_s$, $ae^1_s \Rightarrow_r e^1_s$, and $ae^3_s \Rightarrow_r e^2_s$,
        definition of $\Rightarrow_r$ implies
        \[
            t^1_s;\assume{x}{ae^1_s};t^3_s;\assume{y}{ae^3_s}t^5_s \Rightarrow_r p^1_s;\assume{x}{e^1_s};p^2_s;\assume{y}{e^2_s};p^3_s
        \]
        Therefore
        \[
            t_s \Rightarrow_r p_s
        \]

        \noindent{Because} $ ae^5_s \Rightarrow_r e^3_s$, the definition of $\Rightarrow_r$ implies
        \[
            \ndi{ae^5_s}{v}{id} \implies e^3_s
        \]
        Therefore
        \[
            ae_s \Rightarrow_r e_s
        \]

        \noindent{Because} $\SI{}{} p^1_s;\assume{z}{e^1_s} 
        \Rightarrow_s t^2_s;\assume{z}{ae^2_s}$, $\SI{}{} p^2_s;\assume{z}{e^2_s} \Rightarrow_s t^4_s;\assume{z}{ae^4_s}$, 
        and all variable names introduced by $t^2_s$ do not conflict with variable names in $ae^4_s$ and $t^4_s$ (Observation \ref{ex:obs}),
        the definition of $\Rightarrow_s$ implies
        \[
            \SI{}{} p_s;\assume{z}{e^3_s} \Rightarrow_s t'_s;\assume{z}{ae^6_s}
        \]
        Therefore
        \[
            \SI{}{} p_s;\assume{z}{e_s} \Rightarrow_s t'_s;\assume{z}{ae'_s}
        \]

        \noindent{Therefore} when $ae =  \ndi{(ae_1~ae_2)y=ae_3}{v}{id}$ and $id_e \notin \cS$, and assuming induction hypothesis,
        \[
            \forall t_s, ae_s, ae'~\SAC ae \Rightarrow_{ex} ae_s, t_s~\wedge~
            \SAC ae \equiv ae'~\exists e.~ae \Rightarrow_r e~\wedge~\SI{}{} e \Rightarrow_s \_, \_, ae'
        \]
        \[
          \implies  \exists ae'_s, t'_s, p_s, e_s.~\SAC ae, ae'_s, t'_s \Rightarrow_{st} ae'~\wedge~t_s \Rightarrow_r p_s~\wedge~ae_s \Rightarrow_r e_s
        \]
        \[
            \wedge~\SI{}{} p_s;\assume{z}{e_s} \Rightarrow_s t'_s;\assume{z}{ae'_s}
        \]
 
        \noindent{Because we have covered all cases, using induction, the following statement is true}
    \[
            \forall t_s, ae_s, ae'~\SAC ae \Rightarrow_{ex} ae_s, t_s~\wedge~
            \SAC ae \equiv ae'~\exists e.~ae \Rightarrow_r e~\wedge~\SI{}{} e \Rightarrow_s \_, \_, ae'
        \]
        \[
          \implies  \exists ae'_s, t'_s, p_s, e_s.~\SAC ae, ae'_s, t'_s \Rightarrow_{st} ae'~\wedge~t_s \Rightarrow_r p_s~\wedge~ae_s \Rightarrow_r e_s
        \]
        \[
            \wedge~\SI{}{} p_s;\assume{z}{e_s} \Rightarrow_s t'_s;\assume{z}{ae'_s}
        \]
   \end{proof} 

\begin{lemma}
    \noindent{For any} two traces $t, t'$, a valid subproblem $\cS$ on $t$ and 
        a subtrace $\SAC t \Rightarrow_{ex} t_s$
        \[
            \SAC t \equiv t'  \wedge t \Rightarrow_r p \wedge \SI{}{} p \Rightarrow_s t'
        \]
        implies there exists a subtrace $t'_s$ such that
        \[
            \exists~p_s.t_s \Rightarrow_r p_s \wedge \SI{}{} p_s \Rightarrow_s t'_s \wedge \SAC t, t'_s \Rightarrow_{st} t'
        \]
        \label{lem:comptr}
\end{lemma}
\begin{proof}
        \noindent{Proof by Induction}

        \noindent{\bf Base Case:}
        $t = \emptyset$
        
        \noindent{By assumption}
        \[
            \SAC t \Rightarrow_{ex} t_s
        \]
        By definition of $\Rightarrow_{ex}$
        \[
            \SAC \emptyset \Rightarrow_{ex} \emptyset
        \]
        Then $t_s = \emptyset$.

        \noindent{By assumption}
        \[
            \SAC t \equiv t'
        \]
        By definition of $\equiv$
        \[
            \SAC \emptyset \equiv \emptyset 
        \]
        Then $t' = \emptyset$.

        \noindent{By assumption}
        \[
            t \Rightarrow_r p
        \]
        By definition of $\Rightarrow_r$
        \[
            \emptyset \Rightarrow_r \emptyset
        \]
        Then $p = \emptyset$.

        \noindent{By assumption}
        \[
            \SI{}{} p \Rightarrow_s t'
        \]
        By definition of $\Rightarrow_s$
        \[
            \SI{}{} \emptyset \Rightarrow_s \emptyset
        \]

        \noindent{Consider} $t'_s = \emptyset$, $p_s = \emptyset$.
        By definition of $\Rightarrow_r$
        \[
            \emptyset \Rightarrow_r \emptyset
        \]
        Therefore
        \[
            t_s \Rightarrow_r p_s
        \]

        \noindent{By definition of} $\Rightarrow_s$
        \[
            \SI{}{} \emptyset \Rightarrow_s \emptyset
        \]
        Therefore 
        \[
            \SI{}{} p_s \Rightarrow_s t'_s
        \]

        \noindent{By definition of} $\Rightarrow_{st}$
        \[
            \SAC \emptyset, \emptyset \Rightarrow_{st} \emptyset
        \]
        Therefore
        \[
            \SAC t, t'_s \Rightarrow_{st} t'
        \]

        \noindent{Therefore} when $t = \emptyset$,
        \[
            \forall~t', t_s, p.~\SAC t \Rightarrow_{ex} t_s
            ~\wedge~\SAC t \equiv t'~\wedge t \Rightarrow_r p \wedge  \SI{}{} p \Rightarrow_s t'
        \]
        \[
            \implies \exists~t'_s, p_s.t_s \Rightarrow_r p_s~\wedge~\SI{}{} p_s \Rightarrow_s t'_s~\wedge~\SAC t, t'_s \Rightarrow_{st} t'
        \]

        \noindent{\bf Induction Case:}

        \noindent{Case 1:}
         $t = \assume{x}{ae};t_1$
        
        \noindent{By assumption}
        \[
            \SAC t \Rightarrow_{ex} t_s
        \]
        By definition of $\Rightarrow_{ex}$
        \[
            \SAC \assume{x}{ae};t_1 \Rightarrow_{ex} t^1_s;\assume{x}{ae_s};t^3_s
        \]
        Then $t_s = t^1_s;\assume{x}{ae_s};t^3_s$, $\SAC ae \Rightarrow_{ex} ae_s, t^1_s$
        and $\SAC t_1 \Rightarrow_{ex} t^3_s$.

        \noindent{By assumption}
        \[
            \SAC t \equiv t'
        \]
        By definition of $\equiv$
        \[
            \SAC \assume{x}{ae};t_1  \equiv \assume{x}{ae'};t'_1 
        \]
        Then $t' = \assume{x}{ae'};t'_1 $, $\SAC ae \equiv ae'$, and $\SAC t_1 \equiv t'_1$.

        \noindent{By assumption}
        \[
            t \Rightarrow_r p
        \]
        By definition of $\Rightarrow_r$
        \[
            \assume{x}{ae};t_1  \Rightarrow_r \assume{x}{e};p_1 
        \]
        Then $p = \assume{x}{e};p_1$, $ae \Rightarrow_r e $, and $t_1 \Rightarrow_r p_1$.

        \noindent{By assumption}
        \[
            \SI{}{} p \Rightarrow_s t'
        \]
        By definition of $\Rightarrow_s$
        \[
            \SI{}{} \assume{x}{e};p_1  \Rightarrow_s \assume{x}{ae'};t'_1 
        \]
        Then $\SI{}{} e \Rightarrow_s \_, \_, ae'$, $\SI{'}{'} p_1 \Rightarrow_s t'_1$,
        and $\sigma'_v = \sigma_v[x \rightarrow \cV(ae')]$, $\sigma'_{id} = \sigma_{id}[x \rightarrow ID(ae')]$.

        \noindent{By induction hypothesis}
        \[
            \SAC t_1 \Rightarrow_{ex} t^3_s
            ~\wedge~\SAC t_1 \equiv t'_1~\wedge t_1 \Rightarrow_r p_1 \wedge  \SI{'}{'} p_1 \Rightarrow_s t'_1
        \]
        \[
            \implies \exists~t^4_s, p^2_s.t^3_s \Rightarrow_r p^2_s~\wedge~\SI{'}{'} p^2_s \Rightarrow_s t^4_s~\wedge~
            \SAC t_1, t^4_s \Rightarrow_{st} t'_1
        \]
        Because $\SAC t_1 \Rightarrow_{ex} t^3_s$, $\SAC t_1 \equiv t'_1$, $t_1 \Rightarrow_r p_1$, and $\SI{'}{'} p_1 \Rightarrow_s t'_1$,
        \[
            t^3_s \Rightarrow_r p^2_s~\wedge~\SI{'}{'} p^2_s \Rightarrow_s t^4_s~\wedge~
            \SAC t_1, t^4_s \Rightarrow_{st} t'_1
        \]

        \noindent{By induction hypothesis}
        \[
            \forall t^1_s, ae_s, ae'~\SAC ae \Rightarrow_{ex} ae_s, t^1_s~\wedge~
            \SAC ae \equiv ae'~\exists e.~ae \Rightarrow_r e~\wedge~\SI{}{} e \Rightarrow_s \_, \_, ae'
        \]
        \[
          \implies  \exists ae'_s, t^2_s, p^1_s, e_s.~\SAC ae, ae'_s, t^2_s 
          \Rightarrow_{st} ae'~\wedge~t^1_s \Rightarrow_r p^1_s~\wedge~ae_s \Rightarrow_r e_s
        \]
        \[
            \wedge~\SI{}{} p^1_s;\assume{z}{e_s} \Rightarrow_s t^2_s;\assume{z}{ae'_s}
        \]
        Because $\SAC ae \Rightarrow_{ex} ae_s, t^1_s$, $\SAC ae \equiv ae'$, 
        $ae \Rightarrow_r e$, and $\SI{}{} e \Rightarrow_s \_, \_, ae'$,
        \[
          \SAC ae, ae'_s, t^2_s 
          \Rightarrow_{st} ae'~\wedge~t^1_s \Rightarrow_r p^1_s~\wedge~ae_s \Rightarrow_r e_s
        \]
        \[
            \wedge~\SI{}{} p^1_s;\assume{z}{e_s} \Rightarrow_s t^2_s;\assume{z}{ae'_s}
        \]

        \noindent{Because} $\SAC ae, ae'_s, t^2_s \Rightarrow_{st} ae'$ and $\SAC t_1, t^4_s \Rightarrow_{st} t'_1$, 
        the definition of $\Rightarrow_{st}$ implies
        \[
            \SAC \assume{x}{ae};t_1, t^2_s;\assume{x}{ae'_s};t^4_s \Rightarrow_{st} \assume{x}{ae'};t'_1
        \]
        Therefore
        \[
            \SAC t, t'_s \Rightarrow_{st} t'
        \]

        \noindent{Because} $t^3_s \Rightarrow_r p^2_s$, $t^1_s \Rightarrow_r p^1_s$, and $ae_s \Rightarrow_r e_s$, the 
        definition of $\Rightarrow_r$ implies
        \[
            t^1_s;\assume{x}{ae_s};t^3_s \Rightarrow_r p^1_s;\assume{x}{e_s};p^2_s
        \]
        Therefore
        \[
            t \Rightarrow_r p
        \]

        \noindent{Because} $\SI{}{} p^1_s;\assume{z}{e_s} \Rightarrow_s t^2_s;\assume{z}{ae'_s}$, 
        $\SI{'}{'} p^2_s \Rightarrow_s t^4_s$, all variable names introduced 
        by $t^2_s$ do not conflict with variable names in $t^4_s$ (Observation~\ref{ex:obs}), 
        $\cV(ae') = \cV(ae'_s)$, and $ID(ae') = ID(ae'_s)$ (Observation~\ref{ex:obs2}), the 
        definition of $\Rightarrow_s$ implies
        \[
            \SI{}{} p^1_s;\assume{x}{e_s};p^2_s \Rightarrow_s t^2_s;\assume{x}{ae'_s};t^4_s
        \]
        Therefore
        \[
            \SI{}{} p_s \Rightarrow_s t'_s
        \]

        \noindent{Therefore} when $t = \assume{x}{ae};t_1$, and assuming the induction hypothesis,
        \[
            \forall~t', t_s, p.~\SAC t \Rightarrow_{ex} t_s
            ~\wedge~\SAC t \equiv t'~\wedge t \Rightarrow_r p \wedge  \SI{}{} p \Rightarrow_s t'
        \]
        \[
            \implies \exists~t'_s, p_s.t_s \Rightarrow_r p_s~\wedge~\SI{}{} p_s \Rightarrow_s t'_s~\wedge~\SAC t, t'_s \Rightarrow_{st} t'
        \]

        \noindent{Case 2:}
         $t = \observed{ae}{e_v};t_1$
        
        \noindent{By assumption}
        \[
            \SAC t \Rightarrow_{ex} t_s
        \]
        By definition of $\Rightarrow_{ex}$
        \[
            \SAC \observed{ae}{e_v};t_1 \Rightarrow_{ex} t^1_s;\observed{ae_s}{e_v};t^3_s
        \]
        Then $t_s = t^1_s;\observed{ae_s}{e_v};t^3_s$, $\SAC ae \Rightarrow_{ex} ae_s, t^1_s$
        and $\SAC t_1 \Rightarrow_{ex} t^3_s$.

        \noindent{By assumption}
        \[
            \SAC t \equiv t'
        \]
        By definition of $\equiv$
        \[
            \SAC \observed{ae}{e_v};t_1  \equiv \observed{ae'}{e_v};t'_1 
        \]
        Then $t' = \observed{ae'}{e_v};t'_1 $, $\SAC ae \equiv ae'$, and $\SAC t_1 \equiv t'_1$.

        \noindent{By assumption}
        \[
            t \Rightarrow_r p
        \]
        By definition of $\Rightarrow_r$
        \[
            \observed{ae}{e_v};t_1  \Rightarrow_r \observed{e}{e_v};p_1 
        \]
        Then $p = \observed{e}{e_v};p_1$, $ae \Rightarrow_r e $, and $t_1 \Rightarrow_r p_1$.

        \noindent{By assumption}
        \[
            \SI{}{} p \Rightarrow_s t'
        \]
        By definition of $\Rightarrow_s$
        \[
            \SI{}{} \observed{e}{e_v};p_1  \Rightarrow_s \observed{ae'}{e_v};t'_1 
        \]
        Then $\SI{}{} e \Rightarrow_s \_, \_, ae'$ and $\SI{}{} p_1 \Rightarrow_s t'_1$.
        
        \noindent{By induction hypothesis}
        \[
            \SAC t_1 \Rightarrow_{ex} t^3_s
            ~\wedge~\SAC t_1 \equiv t'_1~\wedge t_1 \Rightarrow_r p_1 \wedge  \SI{}{} p_1 \Rightarrow_s t'_1
        \]
        \[
            \implies \exists~t^4_s, p^2_s.t^3_s \Rightarrow_r p^2_s~\wedge~\SI{}{} p^2_s \Rightarrow_s t^4_s~\wedge~
            \SAC t_1, t^4_s \Rightarrow_{st} t'_1
        \]
        Because $\SAC t_1 \Rightarrow_{ex} t^3_s$, $\SAC t_1 \equiv t'_1$, $t_1 \Rightarrow_r p_1$, and $\SI{}{} p_1 \Rightarrow_s t'_1$,
        \[
            t^3_s \Rightarrow_r p^2_s~\wedge~\SI{}{} p^2_s \Rightarrow_s t^4_s~\wedge~
            \SAC t_1, t^4_s \Rightarrow_{st} t'_1
        \]

        \noindent{By induction hypothesis}
        \[
            \forall t^1_s, ae_s, ae'~\SAC ae \Rightarrow_{ex} ae_s, t^1_s~\wedge~
            \SAC ae \equiv ae'~\exists e.~ae \Rightarrow_r e~\wedge~\SI{}{} e \Rightarrow_s \_, \_, ae'
        \]
        \[
          \implies  \exists ae'_s, t^2_s, p^1_s, e_s.~\SAC ae, ae'_s, t^2_s 
          \Rightarrow_{st} ae'~\wedge~t^1_s \Rightarrow_r p^1_s~\wedge~ae_s \Rightarrow_r e_s
        \]
        \[
            \wedge~\SI{}{} p^1_s;\assume{z}{e_s} \Rightarrow_s t^2_s;\assume{z}{ae'_s}
        \]
        Because $\SAC ae \Rightarrow_{ex} ae_s, t^1_s$, $\SAC ae \equiv ae'$, 
        $ae \Rightarrow_r e$, and $\SI{}{} e \Rightarrow_s \_, \_, ae'$,
        \[
          \SAC ae, ae'_s, t^2_s 
          \Rightarrow_{st} ae'~\wedge~t^1_s \Rightarrow_r p^1_s~\wedge~ae_s \Rightarrow_r e_s
        \]
        \[
            \wedge~\SI{}{} p^1_s;\assume{z}{e_s} \Rightarrow_s t^2_s;\assume{z}{ae'_s}
        \]

        \noindent{Because} $\SAC ae, ae'_s, t^2_s \Rightarrow_{st} ae'$ and $\SAC t_1, t^4_s \Rightarrow_{st} t'_1$, 
        the definition of $\Rightarrow_{st}$ implies
        \[
            \SAC \observed{ae}{e_v};t_1, t^2_s;\observed{ae_s}{e_v};t^4_s \Rightarrow_{st} \observed{ae'}{e_v};t'_1
        \]
        Therefore
        \[
            \SAC t, t'_s \Rightarrow_{st} t'
        \]

        \noindent{Because} $t^3_s \Rightarrow_r p^2_s$, $t^1_s \Rightarrow_r p^1_s$, and $ae_s \Rightarrow_r e_s$, the 
        definition of $\Rightarrow_r$ implies
        \[
            t^1_s;\observed{ae_s}{e_v};t^3_s \Rightarrow_r p^1_s;\observed{e_s}{e_v};p^2_s
        \]
        Therefore
        \[
            t \Rightarrow_r p
        \]

        \noindent{Because} $\SI{}{} p^1_s;\assume{z}{e_s} \Rightarrow_s t^2_s;\assume{z}{ae'_s}$, 
        $\SI{}{} p^2_s \Rightarrow_s t^4_s$, and all variable names introduced 
        by $t^2_s$ do not conflict with variable names in $t^4_s$ (Observation~\ref{ex:obs}), 
         the 
        definition of $\Rightarrow_s$ implies
        \[
            \SI{}{} p^1_s;\observed{e_s}{e_v};p^2_s \Rightarrow_s t^2_s;\observed{ae'_s}{e_v};t^4_s
        \]
        Therefore
        \[
            \SI{}{} p_s \Rightarrow_s t'_s
        \]

        \noindent{Therefore} when $t = \observed{ae}{e_v};t_1$, and assuming the induction hypothesis,
        \[
            \forall~t', t_s, p.~\SAC t \Rightarrow_{ex} t_s
            ~\wedge~\SAC t \equiv t'~\wedge t \Rightarrow_r p \wedge  \SI{}{} p \Rightarrow_s t'
        \]
        \[
            \implies \exists~t'_s, p_s.t_s \Rightarrow_r p_s~\wedge~\SI{}{} p_s \Rightarrow_s t'_s~\wedge~\SAC t, t'_s \Rightarrow_{st} t'
        \]

    \noindent{All cases} have been covered therefore by induction the following statement is true.
        \[
            \forall~t', t_s, p.~\SAC t \Rightarrow_{ex} t_s
            ~\wedge~\SAC t \equiv t'~\wedge t \Rightarrow_r p \wedge  \SI{}{} p \Rightarrow_s t'
        \]
        \[
            \implies \exists~t'_s, p_s.t_s \Rightarrow_r p_s~\wedge~\SI{}{} p_s \Rightarrow_s t'_s~\wedge~\SAC t, t'_s \Rightarrow_{st} t'
        \]
\end{proof}

\begin{theorem}
    Given a valid trace $t$ and a valid subproblem $\cS$ of $t$ and a subtrace 
    $t_s = \mathsf{ExtractTrace}(t, \cS)$. For all possible
    traces $t'$,
    $
        t' \in \mathsf{Traces}(\mathsf{Program}(t))~\wedge~
    \SAC t \equiv t'$ 
implies 
there exists a subtrace $t'_s$ such that
    \begin{itemize}
        \item $t'_s \in \mathsf{Traces}(\mathsf{Program}(t_s))$
        \item $t' = \mathsf{StitchTrace}(t, t'_s, \cS)$
     \end{itemize}
    \label{thm:complete}
\end{theorem}
\begin{proof}
    From definitions of $\mathsf{Traces}$, $\mathsf{Program}$, $\mathsf{ExtractTrace}$, and $\mathsf{StitchTrace}$,
    given a trace $t$ and a valid subproblem $\cS$
    \[
       \exists~p. \SAC t \Rightarrow_{ex} t_s~\wedge~\SAC t \equiv t'~\wedge~\emptyset, \emptyset \vdash p \Rightarrow_s t~\wedge~\emptyset, \emptyset
       \vdash p \Rightarrow_s t'
    \]
    \[
        \implies \exists~t'_s. \SAC t, t'_s \Rightarrow_{st} t'~\wedge~t_s \Rightarrow_r p_s~\wedge~\emptyset, \emptyset \vdash p_s \Rightarrow_s t_s
    \]
    Because Lemma~\ref{lem:comptr} is true for all environments $\sigma_v, \sigma_{id}, \sigma'_v$, and $\sigma'_{id}$. 
    The theorem is equivalent to the lemma, but with $\sigma_v, \sigma_{id}, \sigma'_v$, and $\sigma'_{id}$ set to $\emptyset$.
\end{proof}

\newpage
\Comment{
\chapter{Metaprogramming Theorems}
}

\subsection{Metaprogramming}
\label{append:meta}

\begin{theorem}
    \label{thm:appendequi}
    A reversible subproblem selection strategy $\mathsf{SS}$ divides the trace space of program $p$
    into equivalence classes.
\end{theorem}
\begin{proof}
    $\mathsf{SS} \vdash t \equiv t'$  is an equivalence 
    relation over traces $t,t'\in T$.
         
         Reflexivity : $\mathsf{SS} \vdash t \equiv t$ is true by definition.
         
         Symmetry : $\mathsf{SS} \vdash t \equiv t \wedge \mathsf{SS} \vdash t \equiv t' \implies \mathsf{SS} \vdash t' \equiv t$.
         Hence its symmetric.
         
        Transitivity : $\mathsf{SS} \vdash t_1 \equiv t_2$, $\mathsf{SS} \vdash t_2 \equiv t_3$
        then $\mathsf{SS} \vdash t_1 \equiv t_3$ (by definition of reversibility and symmetry).
\end{proof}

    \begin{lemma}
    \noindent{For any} augmented expression $ae$ and subproblem $\cS$,
        \[
            \SAC ae \Rightarrow_{ex} ae_s, t_s \implies \sembr{ae} = \sembr{ae_s}*\sembr{t_s}
        \]
        \label{lem:denexp}
    \end{lemma}
\begin{proof}   
     \noindent{Proof by Induction}

    \noindent{\bf Base Case:}
    
    \noindent{Case 1:} $ae = x:x$, 
        
    \noindent{By assumption}
        \[
            \SAC ae \Rightarrow_{ex} ae_s, t_s
        \]
    By definition of $\Rightarrow_{ex}$
    \[
        \SAC x:x \Rightarrow_{ex} x:x, \emptyset
    \]
    Then $ae_s = x:x$ and $t_s = \emptyset$.

    \noindent{By definition} of $\mathsf{pdf}$, 
    $\sembr{x:x} = 1$ and $\sembr{\emptyset} = 1$.
    Therefore
    \[
        \sembr{ae} = \sembr{ae_s}*\sembr{t_s}
    \]

    \noindent{Case 2:} $ae = x:v$,
        
        \noindent{By assumption}
        \[
            \SAC ae \Rightarrow_{ex} ae_s, t_s
        \]
        By definition of $\Rightarrow_{ex}$
        \[
            \SAC x:v \Rightarrow_{ex} x:v, \emptyset
        \]
        Then $ae_s = x:v$ and $t_s = \emptyset$.

    \noindent{By definition} of $\mathsf{pdf}$, 
    $\sembr{x:v} = 1$ and $\sembr{\emptyset} = 1$.
    Therefore
    \[
        \sembr{ae} = \sembr{ae_s}*\sembr{t_s}
    \]
    
    \noindent{Case 3:} $ae = \lambda.x~e:v$,
        
        \noindent{By assumption}
        \[
            \SAC ae \Rightarrow_{ex} ae_s, t_s
        \]
        By definition of $\Rightarrow_{ex}$
        \[
            \SAC \lambda.x~e:v \Rightarrow_{ex} \lambda.x~e:v, \emptyset
        \]
        Then $\sembr{\lambda.x~e:v} = v, 1$ and $\sembr{\emptyset} = 1$. 
        Therefore
        \[
            \sembr{ae} = \sembr{ae_s}*\sembr{t_s}
        \]

    \noindent{\bf Induction Case:}

    \noindent{Case 1:} $ae = (ae_1~ae_2)\perp:v$,

    \noindent{By assumption}
    \[
        \SAC ae \Rightarrow_{ex} ae_s, t_s
    \]
    By definition of $\Rightarrow_{ex}$
    \[
        \SAC (ae_1~ae_2)\perp:v \Rightarrow_{ex} (ae^1_s~ae^2_s)\perp:v, t^1_s;t^2_s
    \]
    Then $ae_s = (ae^1_s~ae^2_s)\perp:v$, $t_s = t^1_s;t^2_s$, $\SAC ae_1 \Rightarrow_{ex} ae^1_s, t^1_s$,
    and $\SAC ae_2 \Rightarrow_{ex} ae^2_s, t^2_s$.

    \noindent{By induction hypothesis}
    \[
        \SAC ae_1 \Rightarrow_{ex} ae^1_s, t^1_s \implies \sembr{ae_1} = \sembr{ae^1_s}*\sembr{t^1_s}
    \]
    \noindent{Because} $\SAC ae_1 \Rightarrow_{ex} ae^1_s, t^1_s$,
    \[
       \sembr{ae_1} = \sembr{ae^1_s}*\sembr{t^1_s}
    \]

    \noindent{By induction hypothesis}
    \[
        \SAC ae_2 \Rightarrow_{ex} ae^2_s, t^2_s \implies \sembr{ae_2} = \sembr{ae^2_s}*\sembr{t^2_s}
    \]
    \noindent{Because} $\SAC ae_2 \Rightarrow_{ex} ae^2_s, t^2_s$
    \[
        \sembr{ae_2} = \sembr{ae^2_s}*\sembr{t^2_s}
    \]

    \noindent{From} definition of $\mathsf{pdf}$
    \[
        \sembr{ae} = \sembr{ae_1}*\sembr{ae_2}
    \]
    \[
        = \sembr{ae^1_s}*\sembr{ae^2_s}*\sembr{t^1_s}*\sembr{t^2_s}
    \]
    From definition of $\mathsf{pdf}$
    $
        \sembr{ae_s} = \sembr{ae^1_s}*\sembr{ae^2_s}
    $
    and
    $
        \sembr{t_s} = \sembr{t^1_s}*\sembr{t^2_s}
    $. Therefore
    \[
        \sembr{ae}  = \sembr{ae_s}*\sembr{t_s}
    \]

    \noindent{Case 2:} $ae = \ndi{(ae_1~ae_2)x=ae_3}{v}{id}$ and $ID(ae_1) \in \cS$

    \noindent{By assumption}
    \[
        \SAC ae \Rightarrow_{ex} ae_s, t_s
    \]
    By definition of $\Rightarrow_{ex}$
    \[
        \SAC (ae_1~ae_2)x=ae_3:v \Rightarrow_{ex} (ae^1_s~ae^2_s)x=ae_3:v, t^1_s;t^2_s
    \]
    Then $ae_s = (ae^1_s~ae^2_s)x=ae_3:v$, $t_s = t^1_s;t^2_s$, $\SAC ae_1 \Rightarrow_{ex} ae^1_s, t^1_s$,
    and $\SAC ae_2 \Rightarrow_{ex} ae^2_s, t^2_s$.

    \noindent{By induction hypothesis}
    \[
        \SAC ae_1 \Rightarrow_{ex} ae^1_s, t^1_s \implies \sembr{ae_1} = \sembr{ae^1_s}*\sembr{t^1_s}
    \]
    \noindent{Because} $\SAC ae_1 \Rightarrow_{ex} ae^1_s, t^1_s$,
    \[
       \sembr{ae_1} = \sembr{ae^1_s}*\sembr{t^1_s}
    \]

    \noindent{By induction hypothesis}
    \[
        \SAC ae_2 \Rightarrow_{ex} ae^2_s, t^2_s \implies \sembr{ae_2} = \sembr{ae^2_s}*\sembr{t^2_s}
    \]
    \noindent{Because} $\SAC ae_2 \Rightarrow_{ex} ae^2_s, t^2_s$
    \[
        \sembr{ae_2} = \sembr{ae^2_s}*\sembr{t^2_s}
    \]

    \noindent{From} definition of $\mathsf{pdf}$
    \[
        \sembr{ae} = \sembr{ae_1}*\sembr{ae_2}*\sembr{ae_3}
    \]
    \[
        = \sembr{ae^1_s}*\sembr{ae^2_s}*\sembr{ae_3}*\sembr{t^1_s}*\sembr{t^2_s}
    \]
    From definition of $\mathsf{pdf}$
    $
        \sembr{ae_s} = \sembr{ae^1_s}*\sembr{ae^2_s}*\sembr{ae_3}
    $
    and
    $
        \sembr{t_s} = \sembr{t^1_s}*\sembr{t^2_s}
    $. Therefore
    \[
        \sembr{ae}  = \sembr{ae_s}*\sembr{t_s}
    \]

    \noindent{Case 3:} $ae = \ndi{(ae_1~ae_2)x=ae_3}{v}{id}$ and $ID(ae_1)\notin \cS$

    \noindent{By assumption}
    \[
        \SAC ae \Rightarrow_{ex} ae_s, t_s
    \]
    By definition of $\Rightarrow_{ex}$
    \[
        \SAC (ae_1~ae_2)x=ae_3:v \Rightarrow_{ex} ae^3_s:v, t^1_s;\assume{y}{ae^1_s};t^2_s;\assume{x}{ae^2_s};t^3_s
    \]
    Then $ae_s = ae^3_s:v$, $t_s = t^1_s;\assume{y}{ae^1_s};t^2_s;\assume{x}{ae^2_s};t^3_s
$, $\SAC ae_1 \Rightarrow_{ex} ae^1_s, t^1_s$,
   $\SAC ae_2 \Rightarrow_{ex} ae^2_s, t^2_s$, $\SAC ae_3 \Rightarrow_{ex} ae^3_s, t^3_s$.

    \noindent{By induction hypothesis}
    \[
        \SAC ae_1 \Rightarrow_{ex} ae^1_s, t^1_s \implies \sembr{ae_1} = \sembr{ae^1_s}*\sembr{t^1_s}
    \]
    \noindent{Because} $\SAC ae_1 \Rightarrow_{ex} ae^1_s, t^1_s$,
    \[
       \sembr{ae_1} = \sembr{ae^1_s}*\sembr{t^1_s}
    \]

    \noindent{By induction hypothesis}
    \[
        \SAC ae_2 \Rightarrow_{ex} ae^2_s, t^2_s \implies \sembr{ae_2} = \sembr{ae^2_s}*\sembr{t^2_s}
    \]
    \noindent{Because} $\SAC ae_2 \Rightarrow_{ex} ae^2_s, t^2_s$
    \[
        \sembr{ae_2} = \sembr{ae^2_s}*\sembr{t^2_s}
    \]

    \noindent{By induction hypothesis}
    \[
        \SAC ae_3 \Rightarrow_{ex} ae^3_s, t^3_s \implies \sembr{ae_3} = \sembr{ae^3_s}*\sembr{t^3_s}
    \]
    \noindent{Because} $\SAC ae_3 \Rightarrow_{ex} ae^3_s, t^3_s$
    \[
        \sembr{ae_3} = \sembr{ae^3_s}*\sembr{t^3_s}
    \]

    \noindent{From} definition of $\mathsf{pdf}$
    \[
        \sembr{ae} = \sembr{ae_1}*\sembr{ae_2}*\sembr{ae_3}
    \]
    \[
        = \sembr{ae^1_s}*\sembr{ae^2_s}*\sembr{ae^3_s}*\sembr{t^1_s}*\sembr{t^2_s}*\sembr{t^3_s}
    \]
    From definition of $\mathsf{pdf}$
    $
        \sembr{ae_s} = \sembr{ae^3_s}
    $
    and
    $
        \sembr{t_s} = \sembr{t^1_s}*\sembr{ae^1_s}*\sembr{t^2_s}*\sembr{ae^2_s}*\sembr{t^3_s}
    $. Therefore
    \[
        \sembr{ae}  = \sembr{ae_s}*\sembr{t_s}
    \]

     \noindent{Case 4:} $ae = \mathsf{Dist}(ae_1\#id_e) = ae_2:v$ and $id_e \in \cS$

    \noindent{By assumption}
    \[
        \SAC ae \Rightarrow_{ex} ae_s, t_s
    \]
    By definition of $\Rightarrow_{ex}$
    \[
        \SAC \mathsf{Dist}(ae_1) = ae_2:v \Rightarrow_{ex} \mathsf{Dist}(ae^1_s) = ae_2:v, t^1_s
    \]
    Then $ae_s = \mathsf{Dist}(ae_1) = ae_2:v$, $t_s = t^1_s$, and $\SAC ae_1 \Rightarrow_{ex} ae^1_s, t^1_s$.

    \noindent{By induction hypothesis}
    \[
        \SAC ae_1 \Rightarrow_{ex} ae^1_s, t^1_s \implies \sembr{ae_1} = \sembr{ae^1_s}*\sembr{t^1_s}
    \]
    \noindent{Because} $\SAC ae_1 \Rightarrow_{ex} ae^1_s, t^1_s$,
    \[
       \sembr{ae_1} = \sembr{ae^1_s}*\sembr{t^1_s}
    \]

    \noindent{From} definition of $\mathsf{pdf}$, $ae_2 \Rightarrow_r e$
    \[
        \sembr{ae} = \sembr{ae_1}*\sembr{ae_2}*\mathsf{pdf}_{\mathsf{Dist}}(\cV(ae_1), e)
    \]
    \[
        = \sembr{ae^1_s}*\sembr{ae_2}*\sembr{t^1_s}*\mathsf{pdf}_{\mathsf{Dist}}(\cV(ae_1), e)
    \]
    From definition of $\mathsf{pdf}$
    $
        \sembr{ae_s} = \sembr{ae^1_s}*\sembr{ae_2}*\mathsf{pdf}_{\mathsf{Dist}}(\cV(ae_1), e)
    $
    and
    $
        \sembr{t_s} = \sembr{t^1_s}
    $. Therefore
    \[
        \sembr{ae}  = \sembr{ae_s}*\sembr{t_s}
    \]

    \noindent{Case 5:} $ae = \ndi{\mathsf{Dist}(ae_1\#id_e) = ae_2}{v}{id}$ and $id_e \notin \cS$
     
     \noindent{By assumption}
    \[
        \SAC ae \Rightarrow_{ex} ae_s, t_s
    \]
    By definition of $\Rightarrow_{ex}$
    \[
        \SAC \mathsf{Dist}(ae_1) = ae_2:v \Rightarrow_{ex} ae^2_s:v, t^1_s;\observed{ae^1_s}{e_v};t^2_s
    \]
    Then $ae_s = ae^2_s:v$, $t_s = t^1_s$, $\SAC ae_1 \Rightarrow_{ex} ae^1_s, t^1_s$,  $ae_2 \Rightarrow_r e_v$and
    $\SAC ae_2 \Rightarrow_{ex} ae^1_s, t^1_s$.

    \noindent{By induction hypothesis}
    \[
        \SAC ae_1 \Rightarrow_{ex} ae^1_s, t^1_s \implies \sembr{ae_1} = \sembr{ae^1_s}*\sembr{t^1_s}
    \]
    \noindent{Because} $\SAC ae_1 \Rightarrow_{ex} ae^1_s, t^1_s$,
    \[
       \sembr{ae_1} = \sembr{ae^1_s}*\sembr{t^1_s}
    \]
    
    \noindent{By induction hypothesis}
    \[
        \SAC ae_2 \Rightarrow_{ex} ae^2_s, t^2_s \implies \sembr{ae_2} = \sembr{ae^2_s}*\sembr{t^2_s}
    \]
    \noindent{Because} $\SAC ae_2 \Rightarrow_{ex} ae^2_s, t^2_s$,
    \[
       \sembr{ae_2} = \sembr{ae^2_s}*\sembr{t^2_s}
    \]

    \noindent{From} definition of $\mathsf{pdf}$,
    \[
        \sembr{ae} = \sembr{ae_1}*\sembr{ae_2}*\mathsf{pdf}_{\mathsf{Dist}}(\cV(ae_1), e)
    \]
    \[
        = \sembr{ae^1_s}*\sembr{ae^2_s}*\sembr{t^1_s}*\sembr{t^2_s}*\mathsf{pdf}_{\mathsf{Dist}}(\cV(ae_1), e)
    \]
    From definition of $\mathsf{pdf}$
    $
        \sembr{ae_s} = \sembr{ae^2_s}
    $,
    $
        \sembr{t_s} = \sembr{t^1_s}*\sembr{ae^1_s}*\sembr{t^2_s}*\mathsf{pdf}_{\mathsf{Dist}}(\cV(ae^1_s), e)
    $ and $\cV(ae^1_s) = \cV(ae_1)$ (Observation~\ref{ex:obs2}). Therefore
    \[
        \sembr{ae}  = \sembr{ae_s}*\sembr{t_s}
    \]
       
    \noindent{Because} we have considered all cases, by induction, for augmented expression $ae$, subproblem $\cS$, augmented subexpression $ae_s$, and a subtrace $t_s$,
    \[
        \SAC ae \Rightarrow_{ex} ae_s, t_s \implies \sembr{ae} = \sembr{ae_s}*\sembr{t_s}
    \]

\end{proof}

\begin{theorem}
    \label{appendthm:appdensity}
    Given a trace $t$ and a valid subproblem $\cS$ on $t$, 
    then for subtrace $t_s = \mathsf{ExtractTrace}(t, \cS)$, 
    \[
        \sembr{t} = \sembr{t_s}
    \]
    i.e. for the unnormalized density of $t$ and $t_s$ is equal.
\end{theorem}
\begin{proof}
\noindent{Proof by induction}
 
    \noindent{\bf Base Case:}
    $t = \emptyset$
   
   \noindent{By assumption}
   \[
        \SAC t \Rightarrow_{ex} t_s
    \]
    By definition of $\Rightarrow_{ex}$
    \[
        \SAC \emptyset \Rightarrow_{ex} \emptyset
    \]
    Then 
    $t_s = \emptyset$.
  
    By definition of $\mathsf{pdf}$, $\sembr{\emptyset} = 1$
    \[
        \sembr{t} = \sembr{t_s}
    \]

   \noindent{\bf Induction Case:}

   \noindent{Case 1:}
   $t = \assume{x}{ae};t_1$
 
    \noindent{By assumption}
    \[
        \SAC t \Rightarrow_{ex} t_s
    \]
    By definition of $\Rightarrow_{ex}$
    \[
        \SAC  \assume{x}{ae};t_1 \Rightarrow_{ex} t^1_s;\assume{x}{ae_s};t^2_s
    \]
    Then $\SAC ae \Rightarrow_{ex} ae_s, t^1_s$, $\SAC t_1 \Rightarrow_{ex} t^2_s$.

    \noindent{By induction hypothesis}
    \[
        \SAC t_1 \Rightarrow_{ex} t^2_s \implies \sembr{t_1} = \sembr{t^2_s}
    \]
    Because $\SAC t_1 \Rightarrow_{ex} t^2_s$,
    \[
        \sembr{t_1} = \sembr{t^2_s}
    \]

        \noindent{From} Lemma~\ref{lem:denexp} over augmented expressions
    \[
        \SAC ae \Rightarrow_{ex} ae_s, t^1_s \implies \sembr{ae} = \sembr{ae_s}*\sembr{t^1_s}
    \]
    Because $\SAC ae \Rightarrow_{ex} ae_s, t^1_s$
    \[
        \sembr{ae} = \sembr{ae_s}*\sembr{t^1_s}
    \]

    \noindent{From} definition of $\mathsf{pdf}$
    \[
        \sembr{t} = \sembr{ae}*\sembr{t_1} = \sembr{ae_s}*\sembr{t^1_s}*\sembr{t^2_s}
    \]
    Because $\sembr{t_s} = \sembr{ae_s}*\sembr{t^1_s}*\sembr{t^2_s}$
    \[
        \sembr{t} = \sembr{t_s}
    \]

    \noindent{Case 2:}
   $t = \observed{ae}{e};t_1$

    \noindent{By assumption}
    \[
        \SAC t \Rightarrow_{ex} t_s
    \]
    By definition of $\Rightarrow_{ex}$
    \[
        \SAC  \observed{ae}{e};t_1 \Rightarrow_{ex} t^1_s;\observed{ae_s}{e};t^2_s
    \]
    Then $\SAC ae \Rightarrow_{ex} ae_s, t^1_s$, $\SAC t_1 \Rightarrow_{ex} t^2_s$.

    \noindent{By induction hypothesis}
    \[
        \SAC t_1 \Rightarrow_{ex} t^2_s \implies \sembr{t_1} = \sembr{t^2_s}
    \]
    Because $\SAC t_1 \Rightarrow_{ex} t^2_s$,
    \[
        \sembr{t_1} = \sembr{t^2_s}
    \]

    \noindent{From} Lemma~\ref{lem:denexp} over augmented expressions
    \[
        \SAC ae \Rightarrow_{ex} ae_s, t^1_s \implies \sembr{ae} = \sembr{ae_s}*\sembr{t^1_s}
    \]
    Because $\SAC ae \Rightarrow_{ex} ae_s, t^1_s$
    \[
        \sembr{ae} = \sembr{ae_s}*\sembr{t^1_s}
    \]

    \noindent{From} definition of $\mathsf{pdf}$
    \[
        \sembr{t} = \sembr{ae}*\sembr{t_1}*\mathsf{pdf}_{\mathsf{Dist}}(\cV(ae), e) = \sembr{ae_s}*\sembr{t^1_s}*\sembr{t^2_s}*\mathsf{pdf}_{\mathsf{Dist}}(\cV(ae), e) 
    \]
    Because $\sembr{t_s} = \sembr{ae_s}*\sembr{t^1_s}*\sembr{t^2_s}*\mathsf{pdf}_{\mathsf{Dist}}(\cV(ae), e) $ and $\cV(ae) = \cV(ae_s)$ (Observation \ref{ex:obs2})
    \[
        \sembr{t} = \sembr{t_s}
    \]
    
    Because we have covered all cases, by induction, for any trace $t$ and subproblem $\cS$
    \[
        \SAC t \Rightarrow_{ex} t_s \implies \sembr{t} = \sembr{t_s}
    \]

    Therefore for any trace $t$ and subproblem $\cS$ 
    \[
        t_s = \mathsf{ExtractTrace}(t, \cS) \implies \sembr{t} = \sembr{t_s}
    \]
\end{proof}

\end{document}